\documentclass[12pt,oneside,reqno]{amsart}
\usepackage[latin9]{inputenc}
\usepackage{geometry}
\usepackage{mathrsfs}
\usepackage{mathtools}
\usepackage{amstext}
\usepackage{amsthm}
\usepackage{amssymb}
\usepackage{mathdots}
\usepackage{graphicx}
\usepackage[unicode=true,pdfusetitle,
 bookmarks=true,bookmarksnumbered=true,bookmarksopen=true,bookmarksopenlevel=3,
 breaklinks=false,pdfborder={0 0 1},backref=section,colorlinks=false]
 {hyperref}
\hypersetup{
 hypertex,bookmarksnumbered=true}

\makeatletter

\DeclareTextSymbolDefault{\textquotedbl}{T1}
\providecommand{\tabularnewline}{\\}

\numberwithin{equation}{section}
\numberwithin{figure}{section}
\theoremstyle{plain}
\newtheorem{thm}{\protect\theoremname}
\theoremstyle{definition}
\newtheorem{defn}[thm]{\protect\definitionname}
\theoremstyle{remark}
\newtheorem{rem}[thm]{\protect\remarkname}
\theoremstyle{plain}
\newtheorem{prop}[thm]{\protect\propositionname}
\theoremstyle{plain}
\newtheorem{cor}[thm]{\protect\corollaryname}
\theoremstyle{plain}
\newtheorem{lem}[thm]{\protect\lemmaname}

\usepackage{chngcntr}
\counterwithout{figure}{section}
\newtheorem{assumption}{Assumption}

\makeatother

\providecommand{\corollaryname}{Corollary}
\providecommand{\definitionname}{Definition}
\providecommand{\lemmaname}{Lemma}
\providecommand{\propositionname}{Proposition}
\providecommand{\remarkname}{Remark}
\providecommand{\theoremname}{Theorem}

\begin{document}
\title{Analytic theory of multicavity klystrons}
\author{Alexander Figotin}
\address{Department of Mathematics, University of California at Irvine, CA
92967, USA.}
\email{afigotin@uci.edu}
\begin{abstract}
Multicavity Klystron (MCK) is a high power microwave (HPM) vacuum
electronic device used to amplify radio-frequency (RF) signals. MCKs
have numerous applications, including radar, radio navigation, space
communication, television, radio repeaters, and charged particle accelerators.
The microwave-generating interactions in klystrons take place mostly
in coupled resonant cavities positioned periodically along the electron
beam axis. Importantly, there is no electromagnetic coupling between
cavities. The cavities are coupled only by the flow of bunched electrons
drifting from one cavity to the next. We advance here a Lagrangian
field theory theory of MCKs with the space being represented by one-dimensional
continuum. The theory integrates into it the space-charge effects
including the so-called debunching (electron-to-electron repulsion).
The corresponding Euler-Lagrange equations are ODEs with coefficients
varying periodically in the space. Utilizing the system periodicity
we develop the instrumental features of the Floquet theory including
the monodromy matrix and its Floquet multipliers. We use them to derive
closed form expressions for a number of physically significant quantities.
Those include in particular the dispersion relations and the frequency
dependent gain foundational to the RF signal amplification. We assume
that MCKs operate in voltage amplification mode associated with the
maximal gain. 
\end{abstract}

\maketitle

\section{Introduction\label{sec:int-twtj}}

A klystron is a specialized linear-beam vacuum tube, invented in 1935
by American electrical engineers Russell and Sigurd Varian. Klystron
is used as an amplifier for high radio frequencies, from UHF up into
the microwave range. It was the first genuine microwave electronic
device to take full advantage of the principle of bunching and phasing,
\cite[7.1]{Tsim}. The original description by brothers Varians of
the klystron concept is as follows, \cite{VarVar}:
\begin{quotation}
``A dc stream of cathode rays of constant current and speed is sent
through a pair of grids between which is an oscillating electric field,
parallel to the stream and of such strength as to change the speeds
of the cathode rays by appreciable but not too large fractions of
their initial speed. After passing these grids the electrons with
increased speeds begin to overtake those with decreased speeds ahead
of them. This motion groups the electrons into bunches separated by
relatively empty spaces. At any points beyond the grids, therefore,
the cathode ray current can be resolved into the original dc plus
a nonsinusoidal ac. A considerable fraction of its power can then
be converted into power of high frequency oscillations by running
the stream through a second pair of grids between which is an ac electric
field such as to take energy away from the electrons in bunches. These
two ac fields are best obtained by making the grids form parts of
the surfaces of resonators of type described in This journal by Hansen.''
\end{quotation}
Usage of cavity resonators in the klystron was a revolutionary idea
of Hansen and the Varians, \cite[7.1]{Tsim}. In the pursuit of higher
power and efficiency the original design of Vairan klystrons evolve
significantly over years featuring today multiple cavities and multiple
electron beam, \cite[7.7]{Tsim}. The advantages of klystrons are
their high power and efficiency, potentially wide bandwidth, phase
and amplitude stability, \cite[9.1]{BenSweScha}.

The distinct features of the klystron operation are as follows, \cite[9.1]{BenSweScha}:
\begin{quotation}
''Klystrons have two distinguishing features. First, the microwave-generating
interactions in these devices take place in resonant cavities at discrete
locations along the beam. Second, the drift tube connecting the cavities
is designed so that electromagnetic wave propagation at the operating
frequency is cut off between the cavities; without electromagnetic
coupling between cavities, they are coupled only by the bunched beam,
which drifts from one cavity to the next. This latter feature of these
devices, the lack of feedback between cavities, makes them perhaps
the best-suited of HPM devices to operate as amplifiers.''
\end{quotation}
Importantly, in klystrons the electron bunching is provided by cavity
resonators (often of toroidal shape) acting as $LC$-circuit resonators.
These cavities often utilize the lowest-frequency fundamental modes.
For these modes the electric field energy is localized near the cavity
gaps exposed to the e-beam whereas the magnetic field energy is stored
in cavity toroidal tubes, \cite[7.1]{Tsim}. Cavities (resonator cavities)
would interact with e-beam effectively if they satisfy the following
conditions, \cite[2.3]{Shev}:
\begin{quotation}
``In order to be used in an electron tube, a cavity resonator must
have a region with a relatively strong high-frequency field which
is polarized along the direction of electron flow. This region should,
in the majority of cases, be so small that the electron transit time
is less than the period of change of the field. Hollow toroidal resonators
satisfy these conditions. Toroidal resonators consist of cylinders
with a very prominent \textquotedbl bulge\textquotedbl{} in the middle.''
\end{quotation}
For more information on klystrons and their operation we refer the
reader to \cite[9]{BenSweScha}, \cite{ChoWes}, \cite[7.2]{Grigo},
\cite[3]{Nusi}, \cite[10]{GerWat}, \cite[10, 11]{Gilm1}, \cite{MAEAD},
\cite[4.3]{Paol}, \cite[7]{Tsim}, \cite[16]{ValMid}.

\begin{figure}[h]
\centering{}\includegraphics[bb=0cm 4cm 34cm 17cm,clip,scale=0.4]{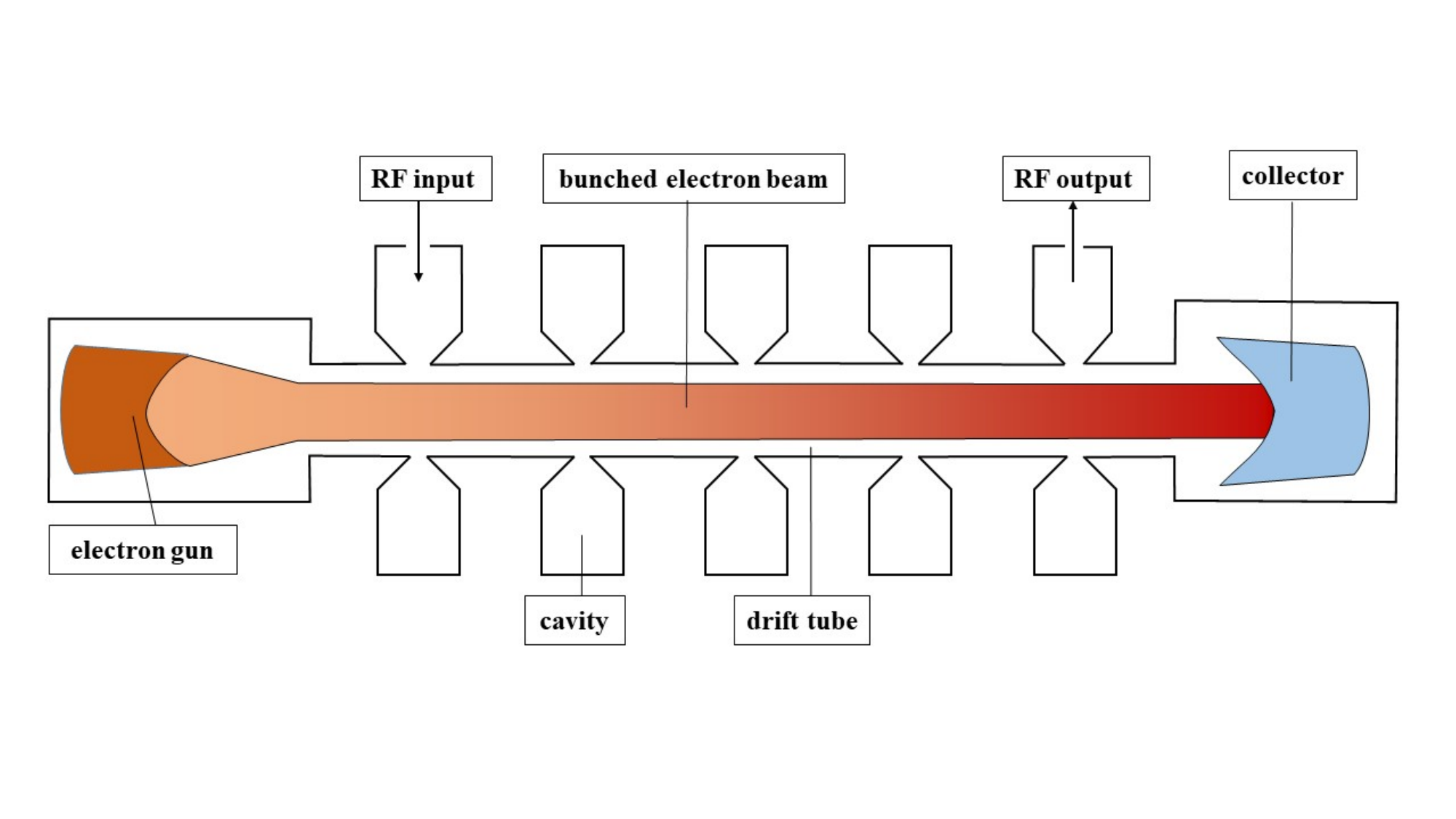}\caption{\label{fig:mck} A schematic presentation of a multicavity klystron
(also known as cascade amplifier) that exploits constructive interaction
between the pencil-like electron beam and an array of electromagnetic
cavities (often of toroidal shape). The interaction causes the electron
bunching and consequent amplification of the RF signal.}
\end{figure}

The conceptual design of a multicavity klystron (MCK) also known as
\emph{cascade amplifier}, \cite[IIb]{Werne}, is shown in Fig. \ref{fig:mck}
and more detailed description of its operation is as follows, \cite[9.3]{BenSweScha},
\cite[8]{Gilm1}, \cite[7.7]{Tsim}. The e-beam enters the gap region
of the first cavity (the \emph{buncher}) where the electron velocities
are modulated by the electric field in the gap driven by an RF signal.
The e-beam-cavity interaction through the cavity gap has the following
features: (i) the input RF voltage in the cavity gap generates the
electric field and that in turn initiates electron bunching (velocity
modulation) and RF current in the e-beam; (ii) the RF current in the
e-beam induces a current in the walls of the cavity and the induced
current acts back on the e-beam enhancing the e-beam modulation. Exiting
the gap region of the first cavity the velocity-modulated e-beam passes
through the \emph{drift region} and enters the gap region of the second
cavity. When drifting between cavities the faster electrons ``overtake''
the slower electron resulting in charge wave bunches on the e-beam.
\emph{Importantly, the drift tube separating the cavities is designed
so that there is no electromagnetic communication between the cavities
except for the bunched e-beam}. Under properly designed conditions
the e-beam charge wave interacts constructively in the gap of the
second cavity achieving an amplified electron bunching upon its exit.
This process of the e-beam charge wave amplification continues on
as the e-beam electrons pass through the drift region and the consequent
cavities. At the end of the process the e-beam enters the gap region
of the last cavity (\emph{extraction cavity, catcher}) where the power
output signal is extracted. Actual multicavity klystron is a very
complicated device with many independent parameters and with \emph{three
important modes to be considered in choosing these parameters: the
voltage amplifier, power amplifier, and bandwidth amplifier modes,}
\cite[7.7]{Tsim}\emph{.} MCKs can be broadband exceeding 10\% with
reasonably flat power output across the band, \cite{Kreu}, \cite[11.3]{Gilm1},
and their efficiency can exceed 70\%, \cite[11.1]{Gilm1}.

The subject our work here is the construction of an analytic theory
of multicavity klystrons operating in the voltage amplification mode
associated with the maximal gain. This theory features in particular
exact formulas for the MCK instability frequencies, its dispersion
relations and optimal values of the MCK parameters providing for maximal
gain.

The paper organized as follows. In Section \ref{sec:twt-mod} we concisely
review our prior work on the analytic theory of traveling wave tubes
(TWT) for its significant elements are utilized for the construction
of the analytic theory of MCK's. In Section \ref{sec:mck-mod} we
introduce the Lagrangian of the MCK system featuring a periodic array
of cavity resonators and derive the corresponding Euler-Lagrange evolution
equations. This Lagrangian has a term that integrates into it the
space-charge effects including the so-called debunching (electron-to-electron
repulsion). We show also in this section that the Euler-Lagrange equations
have the Hamiltonian structure and develop all elements of the Floquet
theory including the MCK monodromy matrix and its Floquet multipliers.
In Section \ref{sec:floqmul} we derive formulas for the Floquet multipliers
that provide a basis for the evaluation of the MCK dispersion relations.
In Section \ref{sec:instfreq} we introduce and study the MCK instability
parameter that determines the region of instability frequencies. In
Section \ref{sec:gain-freq} we derive formulas for the gain as a
function of frequency and its maximal value. In Section \ref{sec:typ}
we evaluate typical values of the MCK gain and its significant parameters.
In Section \ref{sec:disp} we derive explicit formulas for the MCK
dispersion relations. In Section \ref{sec:epd} we find the exceptional
points of degeneracy of the MCK dispersion relations and study their
properties. In Section \ref{sec:lagvar} we construct the Lagrangian
variational framework for the MCK system. In Appendices we provide
information on a number of mathematical subjects relevant to the construction
of the analytic theory of MCK's.

While quoting monographs we identify the relevant sections as follows.
Reference {[}X,Y{]} refers to Section/Chapter ``Y'' of monograph
(article) ``X'', whereas {[}X, p. Y{]} refers to page ``Y'' of
monograph (article) ``X''. For instance, reference {[}2, VI.3{]}
refers to monograph {[}2{]}, Section VI.3; reference {[}2, p. 131{]}
refers to page 131 of monograph {[}2{]}.

\section{Concise review of an analytic model of the traveling wave tube\label{sec:twt-mod}}

When constructing an analytic model of the multicavity klystron we
use some elements of an analytic model of the traveling wave tube
(TWT) introduced and studied in our monograph \cite[4, 24]{FigTWTbk}.
We concisely review here this model of TWT. According to the simplest
version of the model an ideal TWT is represented by a single-stream
e-beam interacting with single transmission line just as in the Pierce
model \cite[I]{Pier51}. The main parameter describing the single-stream
e-beam is e-beam intensity
\begin{equation}
\beta=\frac{\sigma_{\mathrm{B}}}{4\pi}R_{\mathrm{sc}}^{2}\omega_{\mathrm{p}}^{2}=\frac{e^{2}}{m}R_{\mathrm{sc}}^{2}\sigma_{\mathrm{B}}\mathring{n},\quad\omega_{\mathrm{p}}^{2}=\frac{4\pi\mathring{n}e^{2}}{m},\label{eq:T1B1betas1a}
\end{equation}
where $-e$ is electron charge with $e>0$, $m$ is the electron mass,
$\omega_{\mathrm{p}}$ is the e-beam plasma frequency, $\sigma_{\mathrm{B}}$
is the area of the cross-section of the e-beam, s $\mathring{v}>0$
is stationary velocity of electrons in the e-beam and $\mathring{n}$
is the density of the number of electrons. The constant $R_{\mathrm{sc}}$
is the\emph{ }plasma frequency reduction factor that accounts phenomenologically
for finite dimensions of the e-beam cylinder as well as geometric
features of the slow-wave structure, \cite{BraMih}, \cite[9.2]{Gilm1},
\cite[3.3.3]{Nusi}. The frequency
\begin{equation}
\omega_{\mathrm{rp}}=R_{\mathrm{sc}}\omega_{\mathrm{p}}\label{eq:omred1a}
\end{equation}
is known as reduced plasma frequency, \cite[9.2]{Gilm1}.

Assuming the Gaussian system of units the physical dimensions a complete
set of the e-beam parameters as in Tables \ref{tab:e-beam-units}
and \ref{tab:ebeam-par}.
\begin{table}[h]
\centering{}%
\begin{tabular}{|l||l||l|}
\hline 
\noalign{\vskip\doublerulesep}
Frequency & Plasma frequency & $\omega_{\mathrm{p}}=\sqrt{\frac{4\pi\mathring{n}e^{2}}{m}}$\tabularnewline[0.2cm]
\hline 
\hline 
\noalign{\vskip\doublerulesep}
Velocity & e-beam velocity & $\mathring{v}$\tabularnewline[0.2cm]
\hline 
\hline 
\noalign{\vskip\doublerulesep}
Wavenumber &  & $k_{\mathrm{q}}=\frac{\omega_{\mathrm{rp}}}{\mathring{v}}=\frac{R_{\mathrm{sc}}\omega_{\mathrm{p}}}{\mathring{v}}$\tabularnewline[0.2cm]
\hline 
\hline 
\noalign{\vskip\doublerulesep}
Length & Wavelength for $k_{\mathrm{q}}$ & $\lambda_{\mathrm{rp}}=\frac{2\pi\mathring{v}}{\omega_{\mathrm{rp}}},\:\omega_{\mathrm{rp}}=R_{\mathrm{sc}}\omega_{\mathrm{p}}$\tabularnewline[0.2cm]
\hline 
\noalign{\vskip\doublerulesep}
Time & Wave time period & $\mathring{\tau}=\frac{2\pi}{\omega_{\mathrm{p}}}$\tabularnewline[0.2cm]
\hline 
\end{tabular}\vspace{0.3cm}
\caption{\label{tab:e-beam-units} Natural units relevant to the e-beam.}
\end{table}
\begin{table}
\centering{}%
\begin{tabular}{|r||r||r|}
\hline 
\noalign{\vskip\doublerulesep}
$i$ & current & $\frac{\left[\text{charge}\right]}{\left[\text{time}\right]}$\tabularnewline[0.2cm]
\hline 
\noalign{\vskip\doublerulesep}
$q$ & charge & $\left[\text{charge}\right]$\tabularnewline[0.2cm]
\hline 
\noalign{\vskip\doublerulesep}
$\mathring{n}$ & number of electrons p/u of volume & $\frac{\left[\text{1}\right]}{\left[\text{length}\right]^{3}}$\tabularnewline[0.2cm]
\hline 
\noalign{\vskip\doublerulesep}
$\lambda_{\mathrm{rp}}=\frac{2\pi\mathring{v}}{\omega_{\mathrm{rp}}},\:\omega_{\mathrm{rp}}=R_{\mathrm{sc}}\omega_{\mathrm{p}}$ & the electron plasma wavelength & $\left[\text{length}\right]$\tabularnewline[0.2cm]
\hline 
\noalign{\vskip\doublerulesep}
$g_{\mathrm{B}}=\frac{\sigma_{\mathrm{B}}}{4\lambda_{\mathrm{rp}}}$ & the e-beam spatial scale & $\left[\text{length}\right]$\tabularnewline[0.2cm]
\hline 
\noalign{\vskip\doublerulesep}
$\beta=\frac{\sigma_{\mathrm{B}}}{4\pi}R_{\mathrm{sc}}^{2}\omega_{\mathrm{p}}^{2}=\frac{e^{2}}{m}R_{\mathrm{sc}}^{2}\sigma_{\mathrm{B}}\mathring{n}$ & e-beam intensity & $\frac{\left[\text{length}\right]^{2}}{\left[\text{time}\right]^{2}}$\tabularnewline[0.2cm]
\hline 
\noalign{\vskip\doublerulesep}
$\beta^{\prime}=\frac{\beta}{\mathring{v}^{2}}=\frac{\pi\sigma_{\mathrm{B}}}{\lambda_{\mathrm{rp}}^{2}}=\frac{4\pi g_{\mathrm{B}}}{\lambda_{\mathrm{rp}}}$ & dimensionless e-beam intensity & $\left[\text{dim-less}\right]$\tabularnewline[0.2cm]
\hline 
\end{tabular}\vspace{0.3cm}
\caption{\label{tab:ebeam-par}Physical dimensions of the e-beam parameters.
Abbreviations: dimensionless \textendash{} dim-less, p/u \textendash{}
per unit.}
\end{table}

We would like to point to an important spatial scale related to the
e-beam, namely
\begin{equation}
\lambda_{\mathrm{rp}}=\frac{2\pi\mathring{v}}{R_{\mathrm{sc}}\omega_{\mathrm{p}}},\quad\omega_{\mathrm{rp}}=R_{\mathrm{sc}}\omega_{\mathrm{p}},\label{eq:aBvom1a}
\end{equation}
which is the distance passed by an electron for the time period $\frac{2\pi}{\omega_{\mathrm{rp}}}$
associated with the plasma oscillations at the reduced plasma frequency
$\omega_{\mathrm{rp}}$. This scale is well known in the theory of
klystrons and is referred to as \emph{the electron plasma wavelength},
\cite[9.2]{Gilm1}. Another spatial scale related to the e-beam that
arises in our analysis later on is
\begin{equation}
g_{\mathrm{B}}=\frac{\sigma_{\mathrm{B}}}{4\lambda_{\mathrm{rp}}},\label{eq:Bvom1ba}
\end{equation}
and we will refer to it as\emph{ e-beam spatial scale}. Using these
spatial scales we obtain the following representation for the dimensionless
form $\beta^{\prime}$ of the e-beam intensity
\begin{equation}
\beta^{\prime}=\frac{\beta}{\mathring{v}^{2}}=\frac{\pi\sigma_{\mathrm{B}}}{\lambda_{\mathrm{rp}}^{2}}=\frac{4\pi g_{\mathrm{B}}}{\lambda_{\mathrm{rp}}}.\label{eq:Bvom1b}
\end{equation}

As for the single transmission line, its shunt capacitance per unit
of length is a real number $C>0$ and its inductance per unit of length
is another real number $L>0$. The coupling constant $0<b\leq1$ is
a number also, see \cite[3]{FigTWTbk} for more details. The TL single
characteristic velocity $w$ and the single \emph{TL principal coefficient}
$\theta$ defined by
\begin{equation}
w=\frac{1}{\sqrt{CL}},\quad\theta=\frac{b^{2}}{C}.\label{eq:T1B1betas1b}
\end{equation}
Following to \cite[3]{FigTWTbk} we assume that
\begin{equation}
0<\mathring{v}<w.\label{eq:T1B1betas1c}
\end{equation}

\subsection{TWT system Lagrangian and evolution equations\label{subsec:two-lag-ev}}

Following to developments in \cite{FigTWTbk} we introduce the \emph{TWT
principal parameter} $\bar{\gamma}=\theta\beta$. This parameter in
view of equations (\ref{eq:T1B1betas1a}) and (\ref{eq:T1B1betas1b})
can be represented as follows
\begin{equation}
\gamma=\theta\beta=\frac{b^{2}}{C}\frac{\sigma_{\mathrm{B}}}{4\pi}R_{\mathrm{sc}}^{2}\omega_{\mathrm{p}}^{2}=\frac{b^{2}}{C}\frac{e^{2}}{m}R_{\mathrm{sc}}^{2}\sigma_{\mathrm{B}}\mathring{n},\quad\theta=\frac{b^{2}}{C},\quad\beta=\frac{e^{2}}{m}R_{\mathrm{sc}}^{2}\sigma_{\mathrm{B}}\mathring{n}.\label{eq:DsT1B1se1a}
\end{equation}
The TWT-system Lagrangian $\mathcal{L}{}_{\mathrm{TB}}$ in the simplest
case of a single transmission line and one stream e-beam is of the
form, \cite[4, 24]{FigTWTbk}:
\begin{gather}
\mathcal{L}\left(\left\{ Q\right\} ,\left\{ q\right\} \right)=\mathcal{L}_{\mathrm{Tb}}\left(\left\{ Q\right\} ,\left\{ q\right\} \right)+\mathcal{L}_{\mathrm{B}}\left(\left\{ q\right\} \right),\label{eq:T1B1betase1b}\\
\mathcal{L}_{\mathrm{Tb}}=\frac{L}{2}\left(\partial_{t}Q\right)^{2}-\frac{1}{2C}\left(\partial_{z}Q+b\partial_{z}q\right)^{2},\;\mathcal{L}_{\mathrm{B}}=\frac{1}{2\beta}\left(\partial_{t}q+\mathring{v}\partial_{z}q\right)^{2}-\frac{2\pi}{\sigma_{\mathrm{B}}}q^{2},\nonumber 
\end{gather}
where
\begin{gather}
\left\{ Q\right\} =Q,\partial_{z}Q,\partial_{t}Q,\quad Q=Q\left(z,t\right);\quad\left\{ q\right\} =q,\partial_{z}q,\partial_{t}q,\quad q=q\left(z,t\right),\label{eq:T1B1betase1c}
\end{gather}
and $q\left(z,t\right)$ and $Q\left(z,t\right)$\emph{ are charges}
associated with the e-beam and the TL defined as time integrals of
the corresponding e-beam currents $i(z,t)$ and TL current $I(z,t)$,
that is
\begin{equation}
q(z,t)=\int^{t}i(z,t^{\prime})\,\mathrm{d}t^{\prime},\quad.Q(z,t)=\int^{t}I(z,t^{\prime})\,\mathrm{d}t^{\prime}.\label{eq:MTLQs1a}
\end{equation}
Note that term $-\frac{2\pi}{\sigma_{\mathrm{B}}}q^{2}$ in the Lagrangian
$\mathcal{L}_{\mathrm{B}}$ defined in equations (\ref{eq:T1B1betase1b})
represents space-charge effects including the so-called debunching
(electron-to-electron repulsion). The corresponding Euler-Lagrange
equations is the following system of second-order differential equations
\begin{gather}
L\partial_{t}^{2}Q-\partial_{z}\left[C^{-1}\left(\partial_{z}Q+b\partial_{z}q\right)\right]=0,\label{eq:T1B1betasr2a}\\
\frac{1}{\beta}\left(\partial_{t}+\mathring{v}\partial_{z}\right)^{2}q+\frac{4\pi}{\sigma_{\mathrm{B}}}q-b\partial_{z}\left[C^{-1}\left(\partial_{z}Q+b\partial_{z}q\right)\right]=0,\label{eq:T1B1betasr2b}
\end{gather}
where $\mathring{v}$ is the stationary velocity of electrons in the
e-beam, $\sigma_{\mathrm{B}}$ is the area of the cross-section of
the e-beam and $\beta$ is the e-beam intensity defined by equations
(\ref{eq:DsT1B1se1a}).

\section{An analytic model of multicavity klystron\label{sec:mck-mod}}

In the pursuit of powerful pulse microwave radiation the synchronization
of multiple high-frequency sources emerged as a possible solution
to the problem, \cite[7.7]{Tsim}. High-power multicavity klystron
(MCK) is a powerful amplifier that employs this kind of synchronization
and it is the primary subject of our studies here. In particular,
we advance the Lagrangian variational framework that includes: (i)
the MCK system of evolution equations; (ii) closed form expressions
for the MCK dispersion relations derived based on the Floquet theory;
(iii) exact description of the frequency region of the MCK instability;
(iv) exact formulas for the MCK gain as well as for the optimal values
of the MCK parameters that yield the maximal gain. The proposed MCK
model utilizes some of the elements of our analytic model of the traveling
wave tube reviewed in Section \ref{sec:twt-mod}. As to the features
of electron bunching special to klystrons they are as follows, \cite[9.2]{Gilm1}:
\begin{quotation}
``A very important characteristic of the bunching process with space
charge forces is that all electrons are either speeded up or slowed
down to the same velocity (the dc beam velocity) at the same axial
position $\frac{\lambda_{\mathrm{rp}}}{4}$. In addition, even if
the amplitude of the modulating field is changed so that initial electron
velocities are changed, the axial position of the bunch remains the
same. This result is extremely important to the klystron engineer
because, unlike the situation when space charge forces are ignored,
the cavity location for maximum RF beam current is not a function
of signal level, of gap width, or of frequency of operation''.
\end{quotation}
As we already pointed out an actual multicavity klystron is a complicated
device that can be designed to operate in one of \emph{three modes:
the voltage amplifier, power amplifier, and bandwidth amplifier,}
\cite[7.7]{Tsim}\emph{.} We are interested here in the voltage amplifier
mode for it yields the maximal gain, \cite[7.7.1]{Tsim}. \emph{When
in this mode the resonance frequencies of all cavities are identical
and equal to the input operating frequency, a design known as the
synchronous tuning regime with high amplification for a sufficiently
small beam current}. 

When integrating into the mathematical model the identified significant
features of MCKs we make a number of simplifying assumptions. In particular,
we use the following basic assumptions of one-dimensional model of
space-charge waves in velocity-modulated beams: (i) all quantities
of interest depend only on a single space variable $z$; (ii) the
electric field has only an $z$-component; (iii) there are no transverse
velocities of electrons; (iv) ac values are small compared with dc
values; (v) electrons have a constant dc velocity which is much smaller
than the speed of light; (vi) electron beams are nondense, \cite[7.6.1]{Tsim}.

\begin{assumption} \label{ass:mckmod}(ideal model of the e-beam
and cavities interaction).
\begin{enumerate}
\item E-beam is a flow of electrons confined effectively to $z$-axis (see
Fig. \ref{fig:mck}) in consistency with the MCK operation when all
significant energy transport is confined to $z$-axis.
\item The e-beam interacts with a periodic array of cavity resonators of
toroidal shape through their electric field along $z$-axis in small
cavity gaps. The cavity gap centers form a set of equidistant points
on $z$ axis which is a lattice:
\begin{equation}
a\mathbb{Z}:\mathbb{Z}=\left\{ \ldots,-2,-1,0,1,2,\ldots\right\} ,\quad a>0,\label{eq:kaZep1a}
\end{equation}
where $a$ is the MCK \emph{period}. The cavity resonators do not
interact with each other directly but they interact only with the
e-beam at the lattice points as in (\ref{eq:kaZep1a}). This interaction
feature is accomplished by designing the \emph{electron drift tube
(drift space)} so that its low cutoff frequency is above the klystron
operating frequencies.
\item Each cavity interacts with the e-beam at the corresponding lattice
points $a\ell$, $\ell\in\mathbb{Z}$ only by utilizing the single
resonating cavity mode at frequency $\omega_{0}=\frac{1}{\sqrt{l_{0}c_{0}}}$
where $c_{0}$ and $l_{0}$ are respectively the capacitance and the
inductance of each cavity resonator. This assumption enforces the
voltage amplifier mode yielding the maximal gain.
\end{enumerate}
\end{assumption}

An MCK state is described by charges $q=q\left(z,t\right)$, $z\in\mathbb{R}$
and $Q=Q\left(z,t\right)$, $z\in a\mathbb{Z}$ associated with respectively
the e-beam charge-wave and the cavity resonators defined as the time
integrals of the relevant currents
\begin{equation}
Q=Q\left(z,t\right)=\int^{t}I\left(z,t^{\prime}\right)\,\mathrm{d}t^{\prime},\quad.q=q\left(z,t\right)=\int^{t}i\left(z,t^{\prime}\right)\,\mathrm{d}t^{\prime}.\label{eq:kaZep1b}
\end{equation}
Since according to Assumptions \ref{ass:mckmod} the interaction occurs
only at the discrete set $a\mathbb{Z}$ (lattice) of points embedded
into one-dimensional continuum of real numbers $\mathbb{R}$ some
degree of singularity of function $q\left(z,t\right)$ is expected.
As the analysis shows it is appropriate to impose the following \emph{jump-continuity
conditions} on charge function $q\left(z,t\right)$.

\begin{assumption} \label{ass:jumpcon}(jump-continuity of charge
functions).
\begin{enumerate}
\item Functions $q\left(z,t\right)$, $z\in\mathbb{R}$ and their time derivatives
$\partial_{t}^{j}q\left(z,t\right)$ for $j=1,2$ are continuous for
all real $t$ and $z$.
\item Functions $Q\left(z,t\right)$, $z\in a\mathbb{Z}$ and their time
derivatives $\partial_{t}^{j}Q\left(z,t\right)$ for $j=1,2$ are
continuous for all real $t$.
\item Derivatives $\partial_{t}^{j}q\left(z,t\right)$, $\partial_{z}^{j}q\left(z,t\right)$
for $j=1,2$, and the mixed derivatives $\partial_{z}\partial_{t}q\left(z,t\right)=\partial_{t}\partial_{z}q\left(z,t\right)$
exist and continuous for all real real $t$ and $z$ except for the
interaction points on the lattice $a\mathbb{Z}$.
\item For a function $F\left(z\right)$ and a real number $b$ symbols $F$$\left(b-0\right)$
and $F$$\left(b+0\right)$ stand for its left and right limit at
$b$ assuming their existence, that is
\begin{equation}
F\left(b\pm0\right)=\lim_{z\rightarrow b\pm0}F\left(z\right).\label{eq:klimpm1a}
\end{equation}
We also denote by $\left[F\right]\left(b\right)$ the jump of function
$F\left(z\right)$ at $b$, that is
\begin{equation}
\left[F\right]\left(b\right)=F\left(b+0\right)-F\left(b-0\right).\label{eq:klimpm1b}
\end{equation}
\item The following right and left limits exist
\begin{equation}
\partial_{z}^{j}q\left(a\ell\pm0,t\right),\;j=1,2;\;\ell\in\mathbb{Z},\label{eq:klimpm1c}
\end{equation}
and these limits are continuously differentiable functions of $t$.
The values $\partial_{z}q\left(a\ell\pm0,t\right)$ can be different
and consequently the jumps $\left[\partial_{z}q\right]\left(a\ell,t\right)$
can be nonzero.
\end{enumerate}
\end{assumption}

The physical dimensions of quantities related to cavities are summarized
in Table \ref{tab:cavity-par}.
\begin{table}
\centering{}%
\begin{tabular}{|r||r||r|}
\hline 
\noalign{\vskip\doublerulesep}
$I$ & Current & $\frac{\left[\text{charge}\right]}{\left[\text{time}\right]}$\tabularnewline
\hline 
\noalign{\vskip\doublerulesep}
$Q$ & Charge & $\left[\text{charge}\right]$\tabularnewline
\hline 
\noalign{\vskip\doublerulesep}
$c_{0}$ & Cavity capacitance & $\left[\text{length}\right]$\tabularnewline
\hline 
\noalign{\vskip\doublerulesep}
$l_{0}$ & Cavity inductance & $\frac{\left[\text{time}\right]^{2}}{\left[\text{length}\right]}$\tabularnewline
\hline 
\noalign{\vskip\doublerulesep}
$b$ & Coupling parameter & $\left[\text{dim-less}\right]$\tabularnewline
\hline 
\end{tabular}\vspace{0.3cm}
\caption{\label{tab:cavity-par}Physical dimensions of cavity related quantities.
Abbreviations: dimensionless \textendash{} dim-less}
\end{table}

\subsection{MCK Lagrangian and Euler-Lagrange equations\label{subsec:mck-lag}}

To simplify expressions of quantities of interest we use notations
\begin{equation}
\left\{ Q\right\} =Q,\:\:\partial_{t}Q,\quad Q=Q\left(z,t\right),\quad z\in a\mathbb{Z};\quad\left\{ q\right\} =q,\:\partial_{z}q,\:\partial_{t}q,\quad q=q\left(z,t\right),\quad z\in\mathbb{R},\label{eq:aZep1Q}
\end{equation}
\begin{equation}
\left\{ x\right\} =Q,\:\partial_{t}Q,\:q,\:\partial_{t}q,\:\partial_{z}q.\label{eq:eZep1xk}
\end{equation}

The dynamical properties of our MCK model are implemented through
the Lagrangian variational formalism. Namely, the MCK system Lagrangian
$\mathcal{L}$ is defined as the sum of its two components: (i) $\mathcal{L}_{\mathrm{B}}$
is the e-beam Lagrangian; (ii) $\mathcal{L}_{\mathrm{CB}}$ the cavities
and e-beam interaction Lagrangian. That is

\begin{equation}
\mathcal{L}\left(\left\{ x\right\} \right)=\mathcal{L}_{\mathrm{B}}\left(\left\{ q\right\} \right)+\mathcal{L}_{\mathrm{CB}}\left(x\right),\label{eq:aZep1Lk}
\end{equation}
where we used notations (\ref{eq:aZep1Q}) and (\ref{eq:eZep1xk}).
The expressions for $\mathcal{L}_{\mathrm{B}}$ are similar to the
Lagrangian components in equations (\ref{eq:T1B1betase1b}), (\ref{eq:T1B1betase1c}),
namely

\begin{gather}
\mathcal{L}_{\mathrm{B}}\left(\left\{ q\right\} \right)=\frac{1}{2\beta}\left(\partial_{t}q+\mathring{v}\partial_{z}q\right)^{2}-\frac{2\pi}{\sigma_{\mathrm{B}}}q^{2},\label{eq:aZep1ck}
\end{gather}
and the interaction Lagrangian $\mathcal{L}_{\mathrm{CB}}$ is defined
by
\begin{equation}
\mathcal{L}_{\mathrm{CB}}\left(x\right)=\sum_{\ell=-\infty}^{\infty}\delta\left(z-a\ell\right)\left\{ \frac{l_{0}}{2}\left(\partial_{t}Q\left(a\ell\right)\right)^{2}-\frac{1}{2c_{0}}\left[Q\left(a\ell\right)+bq\left(a\ell\right)\right]^{2}\right\} .\label{eq:aZep1dk}
\end{equation}
Note that term $-\frac{2\pi}{\sigma_{\mathrm{B}}}q^{2}$ in the Lagrangian
$\mathcal{L}_{\mathrm{B}}$ defined in equations (\ref{eq:aZep1ck})
represents space-charge effects including the so-called debunching
(electron-to-electron repulsion). Parameters $\beta$ and $\sigma_{\mathrm{B}}$
are respectively the e-beam intensity and the area of the cross-section
of the e-beam defined in Section \ref{sec:twt-mod}.

We would like to point out that: (i) expression (\ref{eq:aZep1dk})
for the interaction Lagrangian $\mathcal{L}_{\mathrm{CB}}$ limits
the interaction by design to points $a\ell$ as indicated by delta
functions $\delta\left(z-a\ell\right)$ and (ii) the factors before
delta functions $\delta\left(z-a\ell\right)$ are expressions similar
to density $\mathcal{L}_{\mathrm{Tb}}$ in equations (\ref{eq:T1B1betase1b})
adapted to set of discrete interaction points $a\ell$; (iii) cavity
capacitance $c_{0}$ is of particular significance for the interaction
between the cavities and the e-beam. \emph{Note that according to
equations (\ref{eq:aZep1Lk}), (\ref{eq:aZep1ck}) and (\ref{eq:aZep1dk})
Lagrangian $\mathcal{L}$ is a periodic function of $z$ of the period
$a$. }

As we derive in Section \ref{sec:lagvar} the Euler-Lagrange (EL)
equations for points $z$ outside the lattice $a\mathbb{Z}$ are

\begin{gather}
\frac{1}{\beta}\left(\partial_{t}+\mathring{v}\partial_{z}\right)^{2}q+\frac{4\pi}{\sigma_{\mathrm{B}}}q=0,\quad z\neq a\ell,\quad\ell\in\mathbb{Z},\label{eq:eZep1ek}
\end{gather}
or equivalently
\begin{gather}
\left(\frac{1}{v}\partial_{t}+\partial_{z}\right)^{2}q+\frac{4\pi\beta}{\sigma_{\mathrm{B}}\mathring{v}^{2}}q=0,\quad z\neq a\ell,\quad\ell\in\mathbb{Z}.\label{eq:eZep1eak}
\end{gather}
The EL equations at the interaction points $a\ell$ (see equations
(\ref{eq:Sact3ek})) are
\begin{equation}
\left[q\right]\left(a\ell\right)=0,\label{eq:eZep1fak}
\end{equation}
\begin{gather}
\partial_{t}^{2}Q\left(a\ell\right)+\omega_{0}^{2}\left[Q\left(a\ell,t\right)+bq\left(a\ell,t\right)\right]=0,\quad\left[\partial_{z}q\right]\left(a\ell\right)=-\frac{b\beta_{0}}{\mathring{v}^{2}}\left[Q\left(a\ell\right)+bq\left(a\ell\right)\right],\label{eq:eZep1fk}
\end{gather}
where we make use of parameters
\begin{equation}
\omega_{0}=\frac{1}{\sqrt{l_{0}c_{0}}},\quad\beta_{0}=\frac{\beta}{c_{0}},\label{eq:eZep1hk}
\end{equation}
and jumps $\left[q\right]\left(a\ell\right)$ are defined by equation
(\ref{eq:klimpm1b}). We refer to $\beta_{0}$ as \emph{cavity e-beam
interaction parameter} and to $\omega_{0}$ as \emph{cavity resonance
frequency}. Note that equations (\ref{eq:eZep1fak}) is just an acknowledgment
of the continuity of charges $q\left(z,t\right)$ at the interaction
points in consistency with Assumption \ref{ass:jumpcon}. Equations
(\ref{eq:eZep1fak}), (\ref{eq:eZep1fk}) can be viewed as the boundary
conditions that are complementary to the differential equations (\ref{eq:eZep1ek})
and (\ref{eq:eZep1eak}).

In what follows the following parameters play an important role in
the analysis
\begin{equation}
f_{\mathrm{B}}=\sqrt{\frac{4\pi\beta}{\sigma_{\mathrm{B}}\mathring{v}^{2}}}=\frac{\omega_{\mathrm{rp}}}{\mathring{v}}=\frac{2\pi}{\lambda_{\mathrm{rp}}},\quad\lambda_{\mathrm{rp}}=\frac{2\pi\mathring{v}}{\omega_{\mathrm{rp}}},\quad\omega_{\mathrm{rp}}=R_{\mathrm{sc}}\omega_{\mathrm{p}},\label{eq:eZep2a}
\end{equation}
where $\omega_{\mathrm{rp}}$ and $\lambda_{\mathrm{rp}}$ are respectively
the \emph{reduced plasma frequency} and \emph{the electron plasma
wavelength}.

The Fourier transform in $t$ (see Appendix \ref{sec:four}) of equations
(\ref{eq:eZep1eak}), (\ref{eq:eZep1fak}) and (\ref{eq:eZep1fk})
(see also equations (\ref{eq:Lagdim1ek}) and (\ref{eq:Lagdim1fk}))
yields the following ordinary differential equations in $z$
\begin{equation}
\left(\partial_{z}-\mathrm{i}\frac{\omega}{\mathring{v}}\right)^{2}\check{q}+f_{\mathrm{B}}^{2}\check{q}=0,\quad f_{\mathrm{B}}=\sqrt{\frac{4\pi\beta}{\sigma_{\mathrm{B}}\mathring{v}^{2}}}=\frac{\omega_{\mathrm{rp}}}{\mathring{v}}=\frac{2\pi}{\lambda_{\mathrm{rp}}},\quad z\neq a\ell,\quad\ell\in\mathbb{Z},\label{eq:eZep3a}
\end{equation}
subjects to the boundary conditions at the interaction points
\begin{gather}
\left[\check{q}\right]\left(a\ell\right)=0,\quad\ell\in\mathbb{Z},\label{eq:eZep3b}\\
\left[\partial_{z}\check{q}\right]\left(a\ell\right)=-\frac{b^{2}\beta_{0}}{\mathring{v}^{2}}\frac{\omega^{2}}{\omega^{2}-\omega_{0}^{2}}\check{q}\left(a\ell\right),\quad\omega_{0}=\frac{1}{\sqrt{l_{0}c_{0}}},\quad\beta_{0}=\frac{\beta}{c_{0}},\label{eq:eZep3c}
\end{gather}
where $\check{q}$ is the time Fourier transform of $q$. Note also
that
\begin{equation}
\check{Q}\left(a\ell\right)=\frac{\omega_{0}^{2}}{\omega^{2}-\omega_{0}^{2}}b\check{q}\left(a\ell\right),\quad\ell\in\mathbb{Z},.\label{eq:eZep3d}
\end{equation}
where $\check{Q}$ is the time Fourier transform of $Q$.

\emph{Hence, the EL differential equations (\ref{eq:eZep3a}) together
with the boundary conditions (\ref{eq:eZep3b}) and (\ref{eq:eZep3c})
form the complete set of equation describing the MCK evolution.} Boundary
conditions (\ref{eq:eZep3b}) and (\ref{eq:eZep3b}) can be recast
into matrix form as follows
\begin{equation}
X\left(a\ell+0\right)=\mathsf{S}_{\mathrm{b}}X\left(a\ell-0\right),\quad\mathsf{S}_{\mathrm{b}}=\left[\begin{array}{rr}
1 & 0\\
-\frac{b^{2}\beta_{0}}{\mathring{v}^{2}}\frac{\omega^{2}}{\omega^{2}-\omega_{0}^{2}} & 1
\end{array}\right],\quad X=\left[\begin{array}{r}
\check{q}\\
\partial_{z}\check{q}
\end{array}\right].\label{eq:eZep3e}
\end{equation}

\subsection{Euler-Lagrange equations in dimensionless variables\label{subsec:EL-dim-var}}

As to basic variables related to the e-beam and the klystron cavities
we refer the reader to Section \ref{sec:twt-mod} and Tables \ref{tab:ebeam-par},
\ref{tab:cavity-par}. The primary dimensionless variables of importance
are
\begin{equation}
z^{\prime}=\frac{z}{a},\quad\partial_{z^{\prime}}=a\partial_{z},\quad t^{\prime}=\frac{\mathring{v}}{a}t,\quad\partial_{t^{\prime}}=\frac{a}{\mathring{v}}\partial_{t},\quad\omega^{\prime}=\frac{\omega}{\omega_{a}}=\frac{a}{2\pi\mathring{v}}\omega,\quad\omega_{0}^{\prime}=\frac{\omega_{0}}{\omega_{a}}=\frac{a}{2\pi\mathring{v}}\omega_{0},\label{eq:unitL1bk}
\end{equation}
\begin{equation}
\beta^{\prime}=\frac{\beta}{\mathring{v}^{2}},\quad\sigma_{\mathrm{B}}^{\prime}=\frac{\sigma_{\mathrm{B}}}{a^{2}},\quad c_{0}^{\prime}=\frac{c_{0}}{a},\quad l_{0}^{\prime}=\frac{\mathring{v}^{2}}{a}l_{0},\label{eq:unitL1ck}
\end{equation}
\begin{equation}
\beta_{0}^{\prime}=\frac{\beta^{\prime}}{c_{0}^{\prime}}=\frac{a\beta}{c_{0}\mathring{v}^{2}},\quad f_{\mathrm{B}}^{\prime}=af_{\mathrm{B}}=\sqrt{\frac{4\pi\beta^{\prime}}{\sigma_{\mathrm{B}}^{\prime}}}=\frac{aR_{\mathrm{sc}}\omega_{\mathrm{p}}}{\mathring{v}}=\frac{2\pi R_{\mathrm{sc}}\omega_{\mathrm{p}}}{\omega_{a}}=\frac{2\pi a}{\lambda_{\mathrm{rp}}},\label{eq:unitL1dk}
\end{equation}
\begin{equation}
a\delta\left(z\right)=\delta\left(z^{\prime}\right),\quad z=az^{\prime}.\label{eq:unitL1ek}
\end{equation}
For the reader convenience we collected in Table \ref{tab:mck-par}
all significant parameters associated with MCK.
\begin{table}
\centering{}%
\begin{tabular}{|r||r||r|}
\hline 
\noalign{\vskip\doublerulesep}
$a$ & the MCK period & $\left[\text{length}\right]$\tabularnewline[0.2cm]
\hline 
\noalign{\vskip\doublerulesep}
$\mathring{v}$ & the e-beam stationary velocity & $\frac{\left[\text{length}\right]}{\left[\text{time}\right]}$\tabularnewline[0.2cm]
\hline 
\noalign{\vskip\doublerulesep}
$\omega_{a}=\frac{2\pi\mathring{v}}{a}$ & the period frequency & $\frac{\left[\text{1}\right]}{\left[\text{time}\right]}$\tabularnewline[0.2cm]
\hline 
\noalign{\vskip\doublerulesep}
$\omega_{\mathrm{p}}=\sqrt{\frac{4\pi\mathring{n}e^{2}}{m}}$ & the plasma frequency & $\frac{\left[\text{1}\right]}{\left[\text{time}\right]}$\tabularnewline[0.2cm]
\hline 
\noalign{\vskip\doublerulesep}
$\lambda_{\mathrm{rp}}=\frac{2\pi\mathring{v}}{\omega_{\mathrm{rp}}},\:\omega_{\mathrm{rp}}=R_{\mathrm{sc}}\omega_{\mathrm{p}}$ & the electron plasma wavelength & $\left[\text{length}\right]$\tabularnewline[0.2cm]
\hline 
\noalign{\vskip\doublerulesep}
$g_{\mathrm{B}}=\frac{\sigma_{\mathrm{B}}}{4\lambda_{\mathrm{rp}}}$ & the e-beam spatial scale & $\left[\text{length}\right]$\tabularnewline[0.2cm]
\hline 
\noalign{\vskip\doublerulesep}
$f^{\prime}=\frac{2\pi\omega_{\mathrm{rp}}}{\omega_{a}}=\frac{2\pi a}{\lambda_{\mathrm{rp}}}$ & normalized period in units of $\frac{\lambda_{\mathrm{rp}}}{2\pi}$ & $\left[\text{dim-less}\right]$\tabularnewline[0.2cm]
\hline 
\noalign{\vskip\doublerulesep}
$\mathring{n}$ & the number of electrons p/u of volume & $\frac{\left[\text{1}\right]}{\left[\text{length}\right]^{3}}$\tabularnewline[0.2cm]
\hline 
\noalign{\vskip\doublerulesep}
$c_{0},\;l_{0}$ & the cavity capacitance, inductance & $\left[\text{length}\right]$$,\;\frac{\left[\text{time}\right]^{2}}{\left[\text{length}\right]}$\tabularnewline[0.2cm]
\hline 
\noalign{\vskip\doublerulesep}
$\omega_{0}=\frac{1}{\sqrt{l_{0}c_{0}}}$ & the cavity resonant frequency & $\frac{\left[\text{1}\right]}{\left[\text{time}\right]}$\tabularnewline[0.2cm]
\hline 
\noalign{\vskip\doublerulesep}
$\beta=\frac{\sigma_{\mathrm{B}}R_{\mathrm{sc}}^{2}\omega_{\mathrm{p}}^{2}}{4\pi}=\frac{e^{2}R_{\mathrm{sc}}^{2}\sigma_{\mathrm{B}}\mathring{n}}{m}=\frac{\pi\sigma_{\mathrm{B}}\mathring{v}^{2}}{\lambda_{\mathrm{rp}}^{2}}$ & the e-beam intensity & $\frac{\left[\text{length}\right]^{2}}{\left[\text{time}\right]^{2}}$\tabularnewline[0.2cm]
\hline 
\noalign{\vskip\doublerulesep}
$\beta^{\prime}=\frac{\beta}{\mathring{v}^{2}}=\frac{\pi\sigma_{\mathrm{B}}}{\lambda_{\mathrm{rp}}^{2}}=\frac{4\pi g_{\mathrm{B}}}{\lambda_{\mathrm{rp}}}$ & dim-less e-beam intensity & $\left[\text{dim-less}\right]$\tabularnewline[0.2cm]
\hline 
\noalign{\vskip\doublerulesep}
$\beta_{0}^{\prime}=\frac{\beta^{\prime}}{c_{0}^{\prime}}=\frac{a\beta}{c_{0}\mathring{v}^{2}}$ & the first interaction par. & $\left[\text{dim-less}\right]$\tabularnewline[0.2cm]
\hline 
\noalign{\vskip\doublerulesep}
$B\left(\omega\right)=B_{0}\frac{\omega^{2}}{\omega^{2}-\omega_{0}^{2}},\quad B_{0}=b^{2}\beta_{0}^{\prime}$ & the second interaction par. & $\left[\text{dim-less}\right]$\tabularnewline[0.2cm]
\hline 
\noalign{\vskip\doublerulesep}
$K_{0}=\frac{B_{0}}{2f}=\frac{b^{2}\beta_{0}^{\prime}}{2f}=\frac{b^{2}\sigma_{\mathrm{B}}}{4\lambda_{\mathrm{rp}}c_{0}}=\frac{b^{2}g_{\mathrm{B}}}{c_{0}}$ & the MCK gain coefficient & $\left[\text{dim-less}\right]$\tabularnewline[0.2cm]
\hline 
\noalign{\vskip\doublerulesep}
$K\left(\omega\right)=\frac{B\left(\omega\right)}{2f}=K_{0}\frac{\omega^{2}}{\omega^{2}-\omega_{0}^{2}}$ & the MCK gain par. & $\left[\text{dim-less}\right]$\tabularnewline[0.2cm]
\hline 
\end{tabular}\vspace{0.3cm}
\caption{\label{tab:mck-par} MCK significant parameters. Abbreviations: dimensionless
\textendash{} dim-less, p/u \textendash{} per unit, par. - parameter.
For the sake of notation simplicity we often omit \textquotedblleft prime\textquotedblright{}
super-index indicating that the dimensionless version of the relevant
parameter is involved when it is clear from the context.}
\end{table}

The dimensionless form $\mathcal{L}^{\prime}$ of the Lagrangians
is as follows:
\begin{equation}
\mathcal{L}^{\prime}=\mathcal{L}_{\mathrm{B}}^{\prime}+\mathcal{L}_{\mathrm{CB}}^{\prime};\quad\mathcal{L}_{\mathrm{B}}^{\prime}=\frac{1}{2\beta^{\prime}}\left(\partial_{t^{\prime}}q+\partial_{z^{\prime}}q\right)^{2}-\frac{2\pi}{\sigma_{\mathrm{B}}^{\prime}}q^{2},\quad\ell\in\mathbb{Z},\label{eq:Lagdim1ak}
\end{equation}
\begin{equation}
\mathcal{L}_{\mathrm{CB}}^{\prime}=\sum_{\ell=-\infty}^{\infty}\delta\left(z^{\prime}-\ell\right)\left\{ \frac{l_{0}^{\prime}}{2}\left(\partial_{t}Q\left(a\ell\right)\right)^{2}-\frac{1}{2c_{0}^{\prime}}\left[Q\left(a\ell\right)+bq\left(a\ell\right)\right]^{2}\right\} .\label{eq:Lagdim1bk}
\end{equation}
The dimensionless form of the EL equations in between interaction
points that corresponds to the Lagrangian $\mathcal{L}^{\prime}$
defined by equations (\ref{eq:Lagdim1ak}) and (\ref{eq:Lagdim1bk})
is
\begin{gather}
\left(\partial_{t^{\prime}}+\partial_{z^{\prime}}\right)^{2}q+f_{\mathrm{B}}^{\prime2}q=0,\quad z^{\prime}\neq\ell,\quad\ell\in\mathbb{Z};\quad f_{\mathrm{B}}^{\prime}=\sqrt{\frac{4\pi\beta^{\prime}}{\sigma_{\mathrm{B}}^{\prime}}}=\frac{2\pi R_{\mathrm{sc}}\omega_{\mathrm{p}}}{\omega_{a}}=\frac{2\pi a}{\lambda_{\mathrm{rp}}},\label{eq:Lagdim1ck}
\end{gather}
and the EL equations at the interaction points $\ell$ are
\begin{gather}
\partial_{t^{\prime}}^{2}Q\left(a\ell\right)+\omega_{0}^{\prime2}\left[Q\left(a\ell\right)+bq\left(a\ell\right)\right]=0,\quad\omega_{0}^{\prime}=\frac{1}{\sqrt{l_{0}^{\prime}c_{0}^{\prime}}},\quad\beta_{0}^{\prime}=\frac{\beta^{\prime}}{c_{0}^{\prime}},\label{eq:Lagdim1dk}\\
\left[\partial_{z^{\prime}}q\right]\left(a\ell\right)=-b\beta_{0}^{\prime}\left[Q\left(a\ell\right)+bq\left(a\ell\right)\right],\quad\ell\in\mathbb{Z},\nonumber 
\end{gather}
where jumps $\left[q\right]\left(a\ell\right)$ are defined by equation
(\ref{eq:klimpm1b}). Note that equations (\ref{eq:Lagdim1dk}) can
naturally be viewed as boundary (interface) conditions complimentary
to the ordinary differential equations (\ref{eq:Lagdim1ck}).

\emph{To simplify notations we will omit prime symbol in equations
but rather will simply acknowledge their dimensionless form}. So we
will use from now on the following \emph{dimensionless form of the
EL equations} (\ref{eq:Lagdim1ck}), (\ref{eq:Lagdim1dk})
\begin{gather}
\left(\partial_{t}+\partial_{z}\right)^{2}q+f^{2}q=0,\quad z\neq\ell,\quad\ell\in\mathbb{Z};\quad f=\frac{2\pi R_{\mathrm{sc}}\omega_{\mathrm{p}}}{\omega_{a}}=\frac{2\pi a}{\lambda_{\mathrm{rp}}},\label{eq:Lagdim1ek}
\end{gather}
\begin{gather}
\partial_{\prime}^{2}Q\left(a\ell\right)+\omega_{0}^{2}\left[Q\left(a\ell\right)+bq\left(a\ell\right)\right]=0,\quad\omega_{0}=\frac{1}{\sqrt{l_{0}c_{0}}},\quad\beta_{0}=\frac{\beta}{c_{0}},.\label{eq:Lagdim1fk}\\
\left[\partial_{z}q\right]\left(a\ell\right)=-b\beta_{0}\left[Q\left(a\ell\right)+bq\left(a\ell\right)\right],\quad\ell\in\mathbb{Z}.\nonumber 
\end{gather}
\emph{Note that in view of the definition of normalized frequency
$f=\frac{2\pi a}{\lambda_{\mathrm{rp}}}$ in equation (\ref{eq:Lagdim1ek})
we may view parameter $\frac{f}{2\pi}=\frac{a}{\lambda_{\mathrm{rp}}}$
in equations (\ref{eq:Lagdim1ek}) as the MCK period measured in natural
to the e-beam spatial unit $\lambda_{\mathrm{rp}}$.}

The Fourier transform in $t$ (see Appendix \ref{sec:four}) of equations
(\ref{eq:Lagdim1ek}), (\ref{eq:Lagdim1fk}) yields
\begin{equation}
\left(\partial_{z}-\mathrm{i}\omega\right)^{2}\check{q}+f^{2}\check{q}=0,\quad z\neq\ell,\label{eq:Lagdim2ak}
\end{equation}
subjects to the boundary conditions at the interaction points
\begin{gather}
\left[\check{q}\right]\left(a\ell\right)=0,\quad\left[\partial_{z}\check{q}\right]\left(a\ell\right)=-B\left(\omega\right)\check{q}\left(a\ell\right)\quad\ell\in\mathbb{Z},\label{eq:Lagdim2aak}
\end{gather}
where $\check{q}$ is the time Fourier transform of $q$ and $B\left(\omega\right)$
is a new important parameter defined by
\begin{equation}
B=B\left(\omega\right)=\frac{b^{2}\beta_{0}\omega^{2}}{\omega^{2}-\omega_{0}^{2}}=B_{0}\frac{\omega^{2}}{\omega^{2}-\omega_{0}^{2}},\quad B_{0}=b^{2}\beta_{0}=\frac{b^{2}\beta}{c_{0}},\label{eq:Lagdim2B1}
\end{equation}
we refer to it as\emph{ cavity e-beam interaction parameter}. The
representation of coefficient $B_{0}$ by equations (\ref{eq:Lagdim2B1})
suggests that it similar to the TWT principal parameter $\gamma$
defined by equations (\ref{eq:DsT1B1se1a}). Note that according to
equations (\ref{eq:Lagdim2B1}) the following representations hold
for the parameter $B_{0}$
\begin{equation}
B_{0}=b^{2}\beta_{0}=\lim_{\omega\rightarrow\infty}B\left(\omega\right)=\left.B\left(\omega\right)\right|_{\omega_{0}=0},\label{eq:Lagdim2B2}
\end{equation}
indicating that $B_{0}$ is the high-frequency limit of $B\left(\omega\right)$
and at the same time it is the value of $B\left(\omega\right)$ when
the cavity resonant frequency $\omega_{0}$ vanishes, that $\omega_{0}=0$.

The Fourier transform in time of equation (\ref{eq:Lagdim1fk}) yields
\begin{equation}
\check{Q}\left(a\ell\right)=\frac{\omega_{0}^{2}}{\omega^{2}-\omega_{0}^{2}}b\check{q}\left(a\ell\right),\quad\ell\in\mathbb{Z},\label{eq:Lagdim2abk}
\end{equation}
where $\check{Q}$ is the time Fourier transform of $Q$, and equation
(\ref{eq:Lagdim2abk}) was used to obtain the second equation in (\ref{eq:Lagdim2aak}).

Boundary conditions (\ref{eq:Lagdim2aak}) can be recast into matrix
form as follows
\begin{equation}
X\left(a\ell+0\right)=\mathsf{S}_{\mathrm{b}}X\left(a\ell-0\right),\quad\mathsf{S}_{\mathrm{b}}=\left[\begin{array}{rr}
1 & 0\\
-B\left(\omega\right) & 1
\end{array}\right],\quad X=\left[\begin{array}{r}
\check{q}\\
\partial_{z}\check{q}
\end{array}\right],\quad B\left(\omega\right)=\frac{b^{2}\beta_{0}\omega^{2}}{\omega^{2}-\omega_{0}^{2}}.\label{eq:Lagdim2ack}
\end{equation}

In order to use the standard form of the Floquet theory reviewed in
Appendix \ref{sec:floquet} we recast the ordinary differential equations
(\ref{eq:Lagdim2ak}) with boundary (interface) conditions (\ref{eq:Lagdim2aak})
as the following single second-order ordinary differential equation
with singular, frequency dependent, periodic potential:
\begin{equation}
\partial_{z}^{2}\check{q}-2\mathrm{i}\omega\partial_{z}\check{q}+\left(f^{2}-\omega^{2}\right)\check{q}-B\left(\omega\right)p\left(z\right)\check{q}=0,\quad p\left(z\right)=\sum_{\ell=-\infty}^{\infty}\delta\left(z-\ell\right),\quad\check{q}=\check{q}\left(z\right),\label{eq:Lagdim4a}
\end{equation}
where the second interaction parameter $B\left(\omega\right)$ is
defined by equation (\ref{eq:Lagdim2B1}).

Analysis of equations (\ref{eq:Lagdim4a}) based on the Floquet theory
(see Appendix \ref{sec:floquet}) becomes now the primary subject
of our studies. The second-order ordinary differential equation (\ref{eq:Lagdim4a})
can in turn be recast into the following matrix ordinary differential
equation
\begin{gather}
\partial_{z}X=A\left(z\right)X,\quad A\left(z\right)=A\left(z,\omega\right)=\left[\begin{array}{rr}
0 & 1\\
\omega^{2}-f^{2}+B\left(\omega\right)p\left(z\right) & 2\mathrm{i}\omega
\end{array}\right],\quad X=\left[\begin{array}{r}
q\\
\partial_{z}q
\end{array}\right],\label{eq:Lagdim4b}\\
B\left(\omega\right)=\frac{b^{2}\beta_{0}\omega^{2}}{\omega^{2}-\omega_{0}^{2}}=2fK_{0}\frac{\omega^{2}}{\omega^{2}-\omega_{0}^{2}},\quad p\left(z\right)=\sum_{\ell=-\infty}^{\infty}\delta\left(z-\ell\right).\nonumber 
\end{gather}
Note that normalized period $f=\frac{2\pi a}{\lambda_{\mathrm{rp}}}$
and the MCK gain coefficient $K_{0}=\frac{b^{2}g_{\mathrm{B}}}{c_{0}}$
play particularly significant roles for the MCK properties. 

One can verify by straightforward evaluation that equation (\ref{eq:Lagdim4a})
has the Hamiltonian structure (see Appendix \ref{sec:Ham}) with the
following selection for the metric matrix
\begin{equation}
G=G^{*}=\left[\begin{array}{rr}
2\omega & \mathrm{i}\\
-\mathrm{i} & 0
\end{array}\right],\quad G^{-1}=\left[\begin{array}{rr}
0 & \mathrm{i}\\
-\mathrm{i} & -2\omega
\end{array}\right],\quad\det\left[G\right]=-1.\label{eq:Lagdim4c}
\end{equation}
The eigenvalues are eigenvectors of metric matrix $G$ are as follows:
\begin{equation}
\omega+\sqrt{\omega^{2}+1},\;\left[\begin{array}{r}
\frac{\mathrm{i}}{-\omega+\sqrt{\omega^{2}+1}}\\
1
\end{array}\right];\quad\omega-\sqrt{\omega^{2}+1},\;\left[\begin{array}{r}
-\frac{\mathrm{i}}{\omega+\sqrt{\omega^{2}+1}}\\
1
\end{array}\right].\label{eq:Lagdim4d}
\end{equation}
Using expressions (\ref{eq:Lagdim4b}) and (\ref{eq:Lagdim4c}) for
respectively matrices $A\left(z\right)$ and $G$ one can readily
verify that $A\left(z\right)$ is $G$-skew-Hermitian matrix, that
is
\begin{equation}
GA\left(z\right)+A^{*}\left(z\right)G=0,\label{eq:Lagdim4e}
\end{equation}
and that according to Appendix \ref{sec:Ham} implies that the system
(\ref{eq:Lagdim4b}) is Hamiltonian. Consequently, according to Appendix
\ref{sec:Ham} the matrizant $\Phi\left(z\right)$ of the Hamiltonian
system (\ref{eq:Lagdim4b}) ) is $G$-unitary and its spectrum $\sigma\left\{ \Phi\left(z\right)\right\} $
is symmetric with respect to the unit circle, that is
\begin{equation}
\Phi^{*}\left(z\right)G\Phi\left(z\right)=G,\quad\zeta\in\sigma\left\{ \Phi\right\} \Rightarrow\frac{1}{\bar{\zeta}}\in\sigma\left\{ \Phi\right\} .\label{eq:Lagdim4f}
\end{equation}

\subsection{The monodromy matrix\label{subsec:monodr-mck}}

We remind that according to the Floquet theory reviewed in Appendix
\ref{sec:floquet} the monodromy matrix encodes significant information
about the relevant first order periodic ODE related to the eigenmodes. 

All analysis here is carried out for dimensionless variables. We begin
with an observation that characteristic polynomial associated with
equation (\ref{eq:Lagdim2ak}) is
\begin{equation}
A_{\mathrm{B}}\left(s\right)=s^{2}-2\mathrm{i}{\it \omega}s+f^{2}-{\it {\it \omega}^{2}},\quad s=\exp\left\{ \mathrm{i}k\right\} ,\quad k=k\left(\omega\right),\label{eq:ebcha3akL}
\end{equation}
where $s$ is the spectral parameter and $k$ can be interpreted as
complex-valued wave number. Note that in accordance with the general
theory of differential equations (see Appendices \ref{sec:co-mat},
\ref{sec:matpol} and \ref{sec:dif-jord}) the spectral parameter
$s$ in the expression (\ref{eq:ebcha3akL}) of the characteristic
polynomial $A_{\mathrm{B}}\left(s\right)$ represents symbolically
the differential operator $\partial_{z}$.

The expression of $2\times2$ companion matrix $\mathscr{C}_{\mathrm{B}}$
(see Appendix \ref{sec:matpol}) of the scalar characteristic polynomial
$A_{\mathrm{B}}\left(s\right)$ defined by equation (\ref{eq:ebcha3akL})
and the corresponding exponential $\exp\left\{ z\mathscr{C}_{\mathrm{B}}\right\} $
are
\begin{equation}
\mathscr{C}_{\mathrm{B}}=\left[\begin{array}{rr}
0 & 1\\
{\it {\it \omega}^{2}}-f^{2} & 2\mathrm{i}{\it \omega}
\end{array}\right]=\mathscr{Z}_{\mathrm{B}}\left[\begin{array}{rr}
\mathrm{i}\left(\omega-f\right) & 0\\
0 & \mathrm{i}\left(\omega+f\right)
\end{array}\right]\mathscr{Z}_{\mathrm{B}}^{-1},\quad\mathscr{Z}_{\mathrm{B}}=\left[\begin{array}{rr}
-\frac{\mathrm{i}}{\omega-f} & -\frac{\mathrm{i}}{\omega+f}\\
1 & 1
\end{array}\right],\label{eq:ebcha3bkL}
\end{equation}
\begin{gather}
\exp\left\{ z\mathscr{C}_{\mathrm{B}}\right\} =\frac{1}{f}\exp\left\{ \mathrm{i}\omega z\right\} \left[\begin{array}{rr}
f\cos\left(fz\right)-\mathrm{i}{\it \omega}\sin\left(fz\right) & \sin\left(fz\right)\\
\left({\it {\it \omega}^{2}}-f^{2}\right)\sin\left(fz\right) & f\cos\left(fz\right)+\mathrm{i}{\it \omega}\sin\left(fz\right)
\end{array}\right]=\label{eq:abcha3ckL}\\
=\mathscr{Z}_{\mathrm{B}}\exp\left\{ \left[\begin{array}{rr}
\mathrm{i}\left(\omega-f\right)z & 0\\
0 & \mathrm{i}\left(\omega+f\right)z
\end{array}\right]\right\} \mathscr{Z}_{\mathrm{B}}^{-1}.\nonumber 
\end{gather}

Applying now the Floquet theory (see Appendix \ref{sec:floquet})
to the first-order ODE equivalent of the EL equations (\ref{eq:Lagdim2ak}),
(\ref{eq:Lagdim2aak}) and (\ref{eq:Lagdim2B1}) and their alternative
form (\ref{eq:Lagdim4a}) we find that the corresponding $2\times2$
matrizant $\Phi\left(z\right)$ satisfies
\begin{gather}
\Phi\left(z\right)=\exp\left\{ \left(z-\ell\right)\mathscr{C}_{\mathrm{B}}\right\} \Phi\left(\ell+0\right),\quad\ell<z<\ell+1,\label{eq:mckPhi1akL}\\
\Phi\left(0+0\right)=\mathbb{I},\quad\Phi\left(\ell+0\right)=\mathsf{S}_{\mathrm{b}}\Phi\left(\ell-0\right),\quad\ell\in\mathbb{Z},\nonumber 
\end{gather}
where $\exp\left\{ z\mathscr{C}_{\mathrm{TB}}\right\} $ is defined
by equation (\ref{eq:abcha3ckL}) and $2\times2$ matrix $\mathsf{S}_{\mathrm{b}}$
is defined by equations (\ref{eq:Lagdim2ack}). Equations (\ref{eq:mckPhi1akL})
imply in turn the following expression for the monodromy matrix $\mathscr{T}=\Phi\left(1+0\right)$:
\begin{gather}
\mathscr{T}=\Phi\left(1+0\right)=\mathsf{S}_{\mathrm{b}}\exp\left\{ \mathscr{C}_{\mathrm{B}}\right\} =\nonumber \\
=e^{\mathrm{i}\omega}\left[\begin{array}{rr}
\cos\left(f\right)-{\it \mathrm{i}\omega}\mathrm{sinc}\,\left(f\right) & \mathrm{sinc}\,\left(f\right)\\
\mathrm{sinc}\,\left(f\right)\omega^{2}+2{\it \mathrm{i}\omega\left(\cos\left(f\right)+b_{f}\right)}-\frac{2b_{f}\cos\left(f\right)+\cos^{2}\left(f\right)+1}{\mathrm{sinc}\,\left(f\right)} & \mathrm{i}\omega\mathrm{sinc}\,\left(f\right)-\cos\left(f\right)-2b_{f}
\end{array}\right],\label{eq:abcha3dkL}
\end{gather}
where
\begin{equation}
b_{f}=\frac{B\left(\omega\right)}{2}\mathrm{sinc}\,\left(f\right)-\cos\left(f\right),\quad\mathrm{sinc}\,\left(f\right)=\frac{\sin\left(f\right)}{f},\quad B\left(\omega\right)=\frac{b^{2}\beta_{0}\omega^{2}}{\omega^{2}-\omega_{0}^{2}}.\label{eq:abcha3ekL}
\end{equation}

Note according to relations (\ref{eq:Lagdim4f}) the \emph{monodromy
matrix $\mathscr{T}=\varPhi\left(1+0\right)$ represented by equations}
(\ref{eq:abcha3dkL})\emph{ is $G$-unitary for metric matrix $G$
defined by equations} (\ref{eq:Lagdim4c})\emph{ and its spectrum
$\sigma\left\{ \mathscr{T}\right\} $ is symmetric with respect to
the unit circle, that is it satisfies relations} (\ref{eq:Lagdim4f}).

We show in Section \ref{sec:floqmul} that the parameter $b_{f}$
defined in equations (\ref{eq:abcha3ekL}) completely determines the
Floquet multipliers (eigenvalues) of monodromy matrix $\mathscr{T}$and
consequently plays a key role in the analysis of the MCK instability.

\section{The Floquet multipliers, the instability and the gain\label{sec:floqmul}}

The MCK instability is manifested by exponentially growing eigenmodes.
In the case of periodic differential equation (\ref{eq:Lagdim4a})
according to the Floquet theory (see Appendix \ref{sec:floquet})
and particularly Remark \ref{rem:disprel}) a Floquet eigenmode grows
exponentially if and only if the absolute value of the corresponding
Floquet multiplier $s$ (that is an eigenvalue of the MCK monodromy
matrix) is greater than 1, that is $\left|s\right|>1$. Consequently
the MCK instability is reduced mathematically to the identification
of conditions when the Floquet multipliers $s$ satisfy inequality
$\left|s\right|>1$. With that in mind we proceed as follows.

We introduce first new frequency dependent parameter
\begin{equation}
K=K\left(\omega\right)=\frac{B\left(\omega\right)}{2f}=K_{0}\frac{\omega^{2}}{\omega^{2}-\omega_{0}^{2}},\quad K_{0}=\frac{b^{2}\beta_{0}}{2f}=\frac{b^{2}\sigma_{\mathrm{B}}}{4\lambda_{\mathrm{rp}}c_{0}}=\frac{b^{2}g_{\mathrm{B}}}{c_{0}},\quad g_{\mathrm{B}}=\frac{\sigma_{\mathrm{B}}}{4\lambda_{\mathrm{rp}}},\label{eq:bfKpsi1a}
\end{equation}
where $\lambda_{\mathrm{rp}}$ and $g_{\mathrm{B}}$ are respectively
the electron plasma wavelength and the e-beam spatial scale (see Table
\ref{tab:mck-par}). We refer to $K=K\left(\omega\right)$ as the
\emph{gain parameter} and to $K_{0}$ the \emph{coefficient of the
gain parameter} for they determine the MCK gain as we show below (see
Definition \ref{def:gain}). Note that gain parameter $K\left(\omega\right)$
defined by equations (\ref{eq:bfKpsi1a}) is evidently a function
of $\frac{\omega}{\omega_{0}}$:
\begin{equation}
K\left(\omega\right)=K_{0}\frac{\omega^{2}}{\omega^{2}-\omega_{0}^{2}}=K_{0}\frac{\check{\omega}^{2}}{\check{\omega}^{2}-1},\quad\check{\omega}=\frac{\omega}{\omega_{0}},\quad K_{0}=\frac{b^{2}\beta_{0}}{2f}=\frac{b^{2}g_{\mathrm{B}}}{c_{0}},\quad g_{\mathrm{B}}=\frac{\sigma_{\mathrm{B}}}{4\lambda_{\mathrm{rp}}}.\label{eq:bfKpsi1c}
\end{equation}
As to the gain coefficient $K_{0}$ in view of equations (\ref{eq:bfKpsi1a})
it satisfies the following identities
\begin{equation}
K_{0}=\lim_{\omega\rightarrow\infty}K\left(\omega\right)=\left.K\left(\omega\right)\right|_{\omega_{0}=0}=\frac{b^{2}\beta_{0}}{2f}=\frac{b^{2}g_{\mathrm{B}}}{c_{0}},\quad g_{\mathrm{B}}=\frac{\sigma_{\mathrm{B}}}{4\lambda_{\mathrm{rp}}}.\label{eq:bfKpsi1ca}
\end{equation}
Using the gain parameter $K$ we recast parameter $b_{f}$ in equations
(\ref{eq:abcha3ekL}) as follows:
\begin{equation}
b_{f}=b_{f}\left(\omega\right)=K\left(\omega\right)\sin\left(f\right)-\cos\left(f\right),\quad K\left(\omega\right)=K_{0}\frac{\omega^{2}}{\omega^{2}-\omega_{0}^{2}},\quad K_{0}=\frac{b^{2}g_{\mathrm{B}}}{c_{0}}.\label{eq:bfKpsi1b}
\end{equation}
and acknowledging its key role for the MCK instability we name $b_{f}$
\emph{instability parameter}. It turns out that the high-frequency
limit $b_{f}^{\infty}$ of instability parameter $b_{f}$ defined
by
\begin{equation}
b_{f}^{\infty}=\lim_{\omega\rightarrow\infty}b_{f}\left(\omega\right)=K_{0}\sin\left(f\right)-\cos\left(f\right),\quad K_{0}=\frac{b^{2}g_{\mathrm{B}}}{c_{0}},\label{eq:bfKinf1}
\end{equation}
plays significant role in the analysis of the MCK instability and
its gain.

We turn now to two Floquet multipliers $s_{\pm}$ which are the eigenvalues
of the monodromy matrix $\mathscr{T}$ defined by equations (\ref{eq:abcha3dkL})
and (\ref{eq:abcha3ekL}). Hence, $s_{\pm}$ are solutions to the
characteristic equation $\det\left\{ \mathscr{T}-s\mathbb{I}\right\} =0$
which is the following quadratic equation:
\begin{gather}
s=e^{\mathrm{i}k}=e^{\mathrm{i}\omega}S:\;S^{2}+2b_{f}S+1=0;\quad S=S_{\pm}=-b_{f}\pm\sqrt{b_{f}^{2}-1},\label{eq:bfKpsi1ba}
\end{gather}
readily implying
\begin{gather}
s_{\pm}=e^{\mathrm{i}k_{\pm}}=e^{\mathrm{i}\omega}S_{\pm}=e^{\mathrm{i}\omega}\left(-b_{f}\pm\sqrt{b_{f}^{2}-1}\right),\quad b_{f}=b_{f}\left(\omega\right)=K\left(\omega\right)\sin\left(f\right)-\cos\left(f\right),\label{eq:bfKpsi1bb}\\
K\left(\omega\right)=K_{0}\frac{\omega^{2}}{\omega^{2}-\omega_{0}^{2}},\quad K_{0}=\frac{b^{2}\beta_{0}}{2f}=\frac{b^{2}g_{\mathrm{B}}}{c_{0}},\nonumber 
\end{gather}
We will also use the following normalized form of the characteristic
equation (\ref{eq:bfKpsi1ba})
\begin{equation}
s^{2}+2b_{f}e^{\mathrm{i}\omega}s+e^{2\mathrm{i}\omega}=0,\quad b_{f}=b_{f}\left(\omega\right)=K_{0}\frac{\omega^{2}}{\omega^{2}-\omega_{0}^{2}}\sin\left(f\right)-\cos\left(f\right).\label{eq:bfKpsi1bc}
\end{equation}

Equations (\ref{eq:bfKpsi1ba}) show that\emph{ parameter $b_{f}$
completely determines the two Floquet multipliers $s_{\pm}$} justifying
its its name the instability parameter.\emph{ Importantly, the characteristic
equation (\ref{eq:bfKpsi1ba}) can be viewed as an expression of the
dispersion relations between the frequency $\omega$ and the wavenumber
$k$} as we discuss in Section \ref{sec:disp}.

Note that quantities $S_{\pm}$ that solve quadratic equation (\ref{eq:bfKpsi1ba})
and $s_{\pm}=e^{\mathrm{i}\omega}S_{\pm}$ satisfy the following elementary
identities
\begin{equation}
S_{+}S_{-}=1;\quad s_{+}s_{-}=e^{2\mathrm{i}\omega},\quad\left|s_{+}\right|\left|s_{-}\right|=1,\label{eq:sSbf1a}
\end{equation}
implying that if the pairs $S_{+}$ and $S_{-}$ and consequently
$s_{+}$ and $s_{-}$ are outside the unit circle then the both pairs
are symmetric with respect to it as illustrated in Figures \ref{fig:mck-eig1}
(a) and \ref{fig:mck-eig1} (a). Not also that expressions (\ref{eq:bfKpsi1bb})
for $s_{\pm}$ imply that $s_{+}$ and $s_{-}$ are equal if and only
if $b_{f}^{2}=1$, that is
\begin{equation}
s_{+}=s_{-}\Leftrightarrow b_{f}=\pm1.\label{eq:sSbf1b}
\end{equation}
\begin{figure}[h]
\begin{centering}
\includegraphics[scale=0.35]{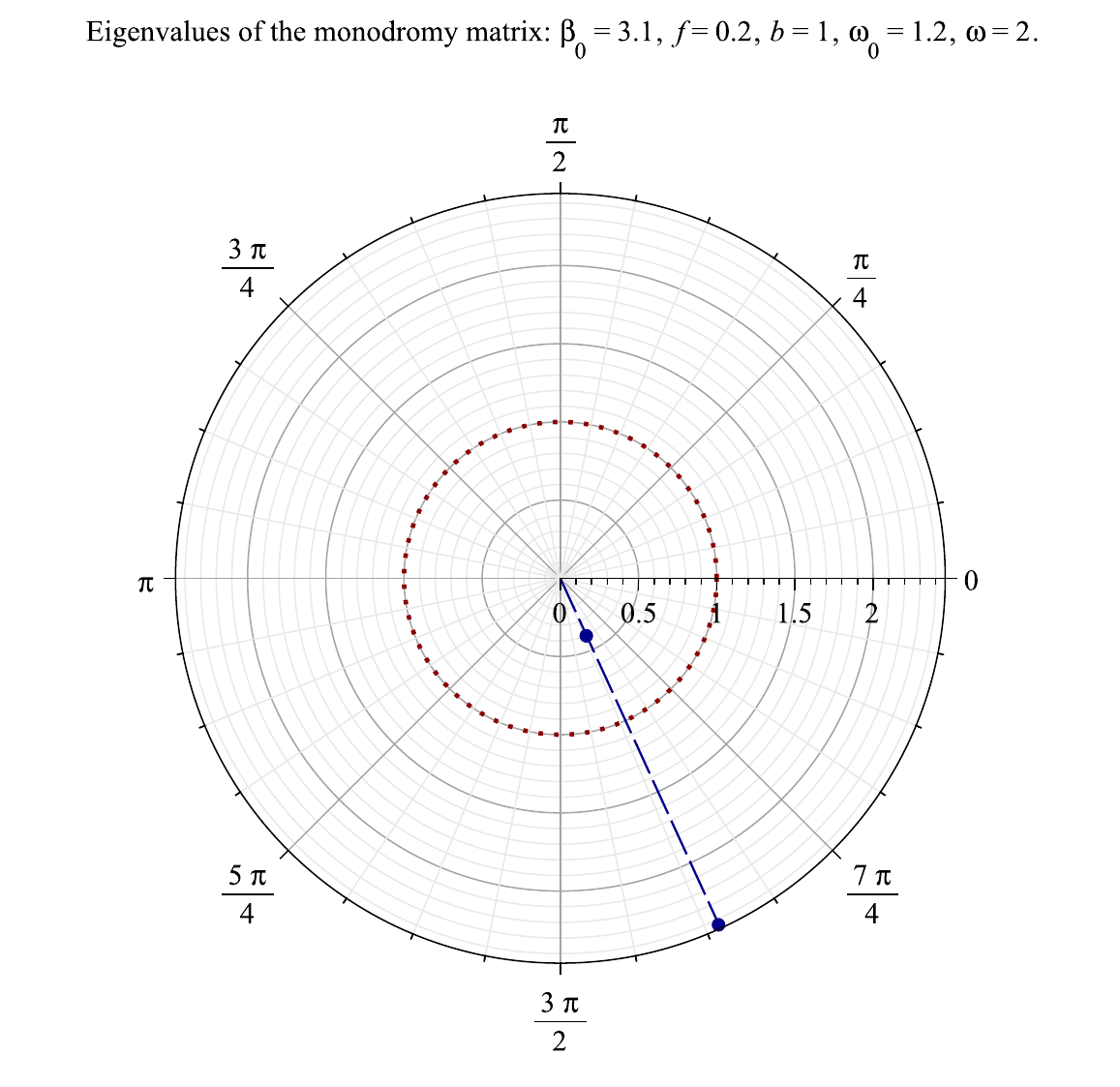}\hspace{0.5cm}\includegraphics[scale=0.33]{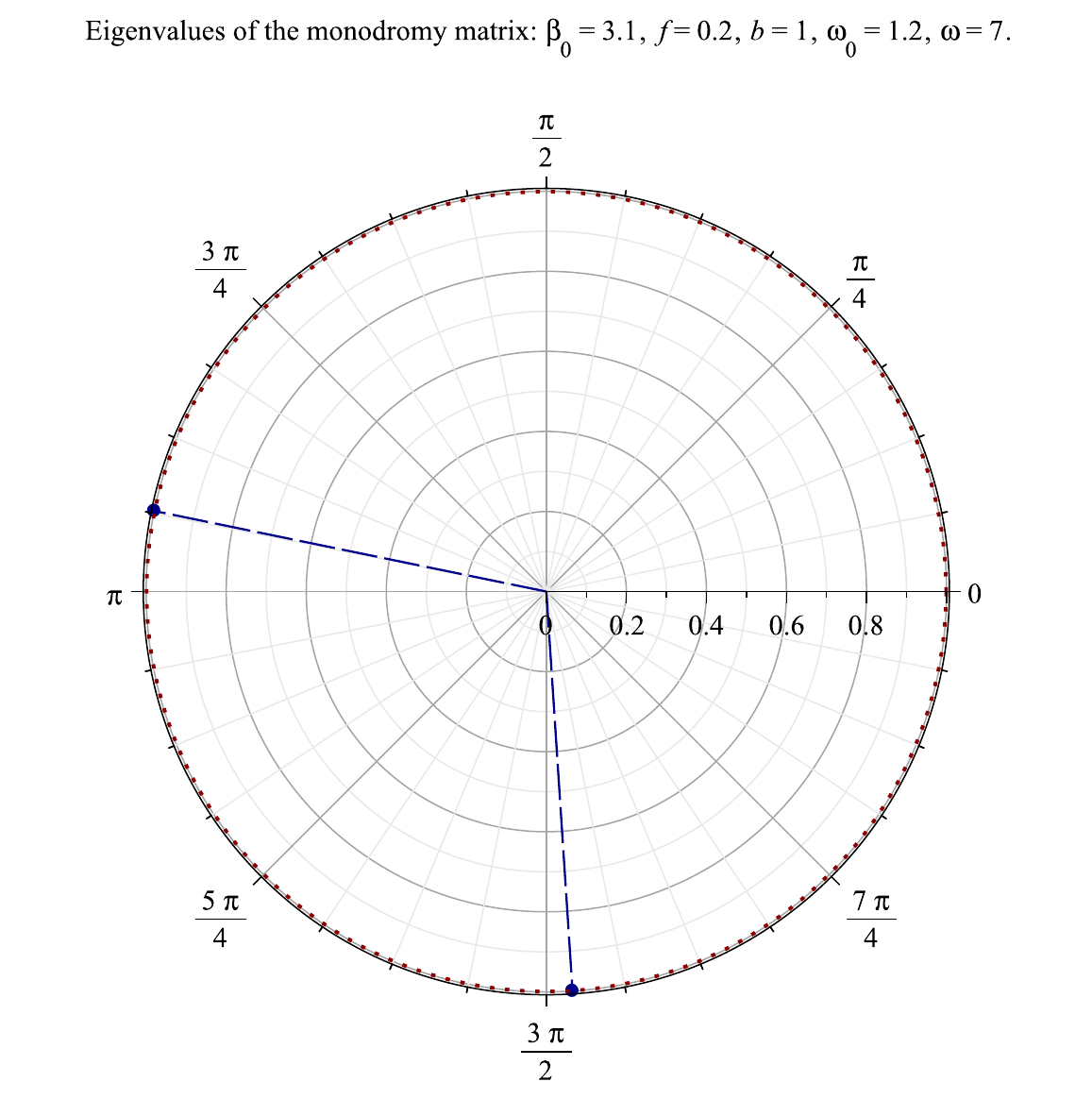}
\par\end{centering}
\centering{}(a)\hspace{7cm}(b)\caption{\label{fig:mck-eig1} The two complex eigenvalues (the Floquet multipliers)
$s_{\pm}=e^{\mathrm{i}\omega}S_{\pm}$ of the monodromy matrix $\mathscr{T}$
defined by equations (\ref{eq:bfKpsi1bb}) for $\beta_{0}=3.1$, $f=0.2$,
$b=1$, $\omega_{0}=1.2$ and two values of $\omega$: (a) $\omega_{0}=1.2<\omega=2<\varOmega_{0.2}^{+}$:
two eigenvalues shown as solid (blue) dots reside outside the the
unit circle; (b) $\omega=7>\varOmega_{0.2}^{+}$: two eigenvalues
shown as solid (blue) dots reside on the unit circle. The horizontal
and vertical axes represent respectively $\Re\left\{ s\right\} $
and $\Im\left\{ s\right\} $. In both cases $\omega=2,7>\omega_{0}=1.2$,
that is the chosen frequencies $\omega$ are above the resonant frequency
$\omega_{0}=1.2$. The doted red circle represents the unit circle.
See Fig. \ref{fig:Om-plusmin} showing plots of functions $\varOmega_{f}^{\pm}$.}
\end{figure}
\begin{figure}[h]
\begin{centering}
\includegraphics[scale=0.32]{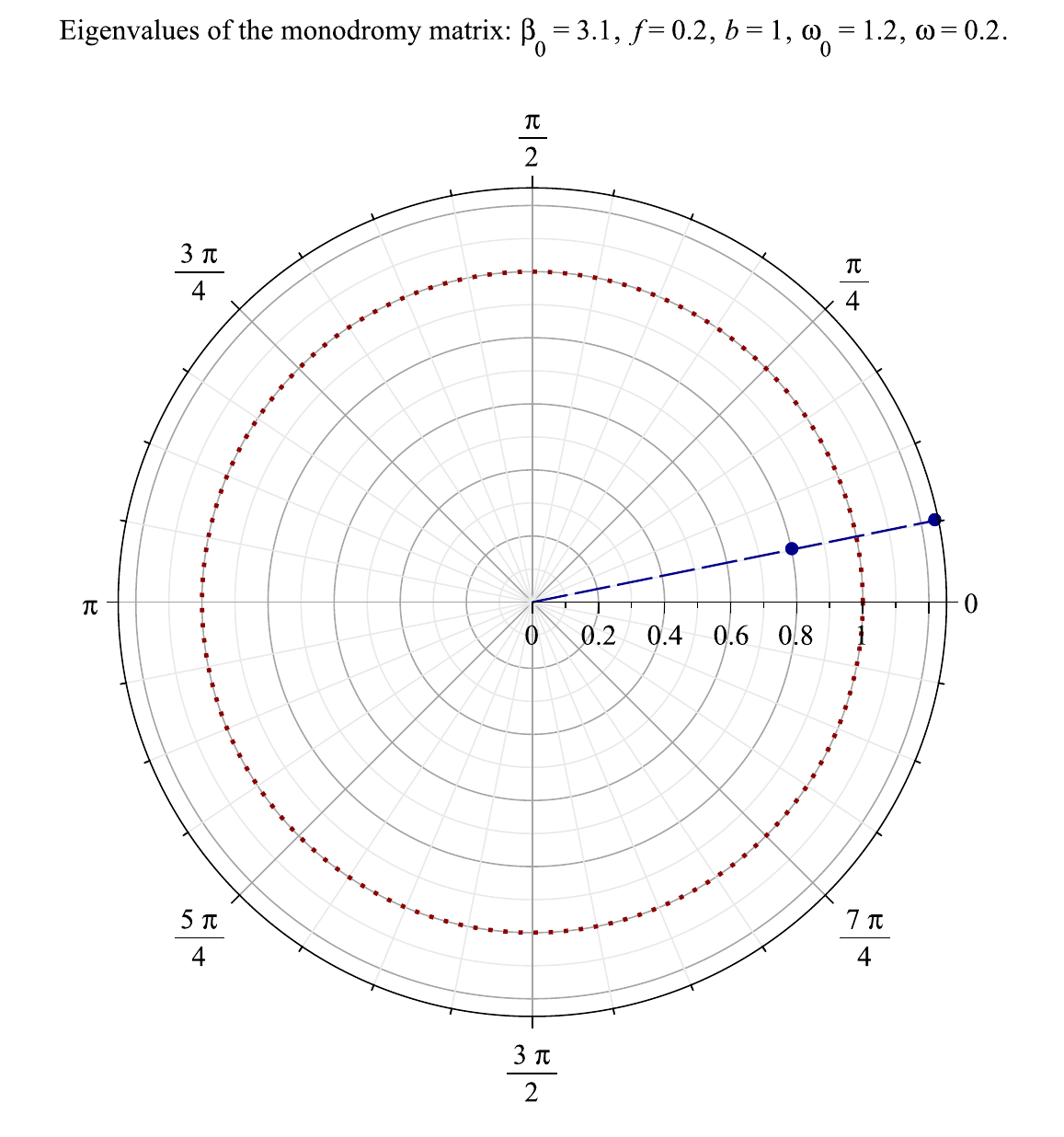}\hspace{1.2cm}\includegraphics[scale=0.3]{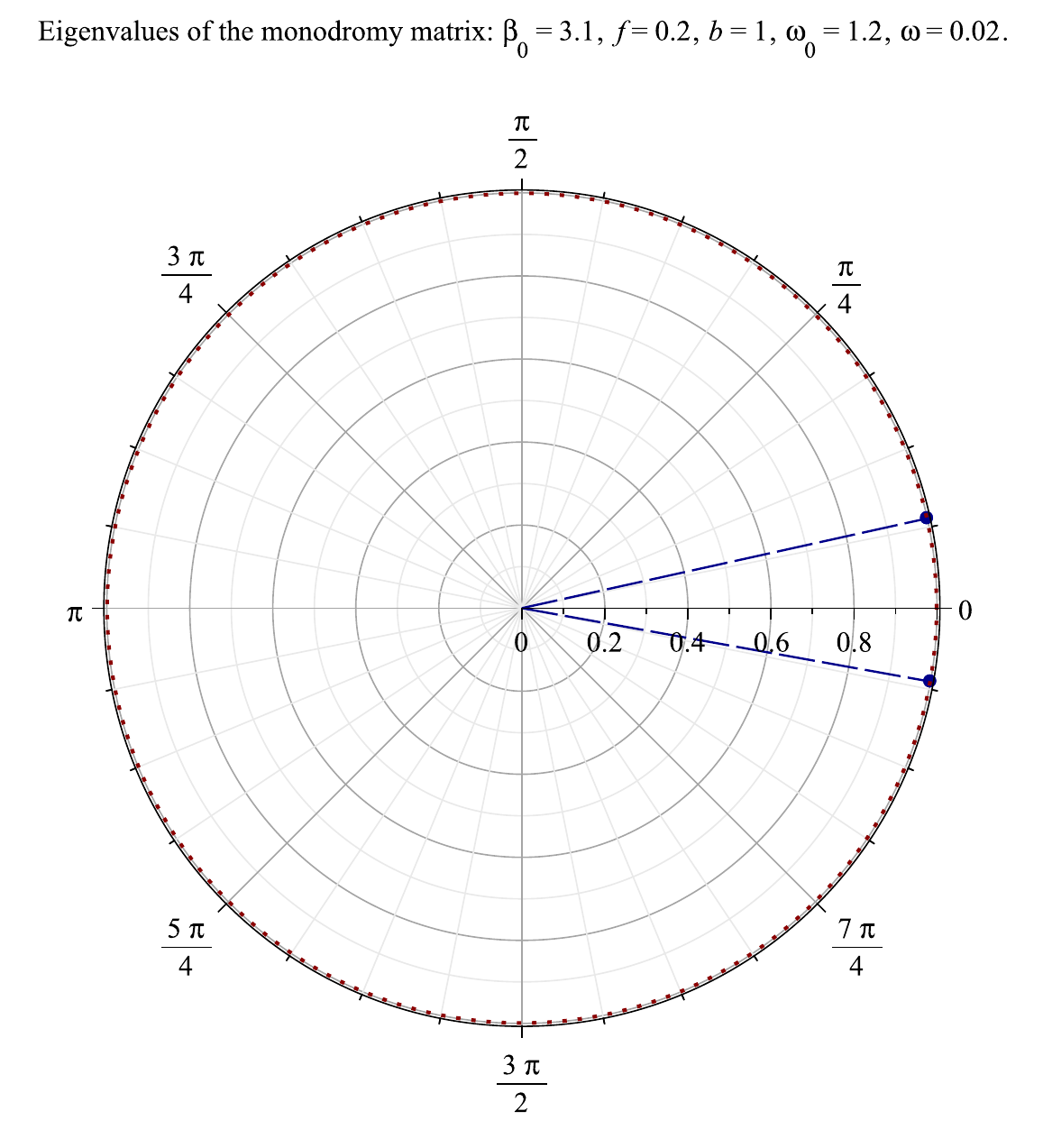}
\par\end{centering}
\centering{}(a)\hspace{7cm}(b)\caption{\label{fig:mck-eig2} The two complex eigenvalues (the Floquet multipliers)
$s_{\pm}=e^{\mathrm{i}\omega}S_{\pm}$ of the monodromy matrix $\mathscr{T}$
defined by equations (\ref{eq:abcha3dkL}) for $\beta_{0}=3.1$, $f=0.2$,
$b=1$, $\omega_{0}=1.2$ and two values of $\omega$: (a) $\varOmega_{0.2}^{-}<\omega=0.2<\omega_{0}=1.2$:
two eigenvalues shown as solid (blue) dots reside outside the the
unit circle; (b) $\omega=0.02<\varOmega_{0.2}^{-}$: two eigenvalues
shown as solid (blue) dots reside on the unit circle. The horizontal
and vertical axes represent respectively $\Re\left\{ s\right\} $
and $\Im\left\{ s\right\} $. In both cases $\omega=0.02,0.2<\omega_{0}=1.2$,
that is the chosen frequencies $\omega$ are below the resonant frequency
$\omega_{0}=1.2$. The doted red circle represents the unit circle.
See Fig. \ref{fig:Om-plusmin} showing plots of functions $\varOmega_{f}^{\pm}$.}
\end{figure}

Figures \ref{fig:mck-eig1} and \ref{fig:mck-eig1} illustrate possible
locations of the two Floquet multipliers $s_{\pm}$ in the complex
plane $\mathbb{C}$.

Using (\ref{eq:bfKpsi1bb}) and carrying out elementary algebraic
transformations we obtain the following statement.
\begin{thm}[Floquet multipliers]
\label{thm:floqmultbf} The instability parameter $b_{f}=b_{f}\left(\omega\right)$
defined by equations (\ref{eq:bfKpsi1b}) and its absolute value $\left|b_{f}\right|$
are respectively $2\pi$-periodic and $\pi$-periodic functions of
$f$, that is
\begin{equation}
b_{f+2\pi}=b_{f};\quad b_{f+\pi}=-b_{f};\quad\left|b_{f+\pi}\right|=\left|b_{f}\right|,\quad b_{f}=b_{f}\left(\omega\right).\label{eq:bfspm1a}
\end{equation}
Let also $s_{\pm}$ be the two MCK Floquet multipliers solving quadratic
equation (\ref{eq:bfKpsi1ba}). Then exactly one of the following
three possibilities always occurs:
\begin{gather}
\left|b_{f}\right|>1:\quad s_{\pm}=-\mathrm{sign}\,\left\{ b_{f}\right\} \left(\left|b_{f}\right|\pm\sqrt{b_{f}^{2}-1}\right)\exp\left\{ \mathrm{i}\omega\right\} ,\quad\left|s_{-}\right|<1<\left|s_{+}\right|;\label{eq:bfspm1b}
\end{gather}
\begin{equation}
\left|b_{f}\right|<1:\quad s_{\pm}=-\mathrm{sign}\,\left\{ b_{f}\right\} \exp\left\{ \mathrm{i}\left[\omega\pm\arccos\left(\left|b_{f}\right|\right)\right]\right\} ,\quad\left|s_{\pm}\right|=1;\label{eq:bfspm1c}
\end{equation}
\begin{equation}
\left|b_{f}\right|=1:\quad s_{\pm}=-\mathrm{sign}\,\left\{ b_{f}\right\} \exp\left\{ \mathrm{i}\omega\right\} ,\quad s_{+}=s_{-},\quad\left|s_{\pm}\right|=1.\label{eq:bfspm1d}
\end{equation}
where
\begin{equation}
b_{f}=b_{f}\left(\omega\right)=K_{0}\frac{\omega^{2}}{\omega^{2}-\omega_{0}^{2}}\sin\left(f\right)-\cos\left(f\right),\quad K_{0}=\frac{b^{2}\beta_{0}}{2}=\frac{b^{2}g_{\mathrm{B}}}{c_{0}},\quad g_{\mathrm{B}}=\frac{\sigma_{\mathrm{B}}}{4\lambda_{\mathrm{rp}}}.\label{eq:bfspm1e}
\end{equation}
Relations (\ref{eq:bfspm1a})-(\ref{eq:bfspm1e}) imply
\begin{equation}
s_{\pm}\left(f+2\pi\right)=s_{\pm}\left(f\right);\quad s_{\pm}\left(f+\pi\right)=-s_{\pm}\left(f\right);\quad\left|s_{\pm}\left(f+\pi\right)\right|=\left|s_{\pm}\left(f\right)\right|.\label{eq:bfspm1g}
\end{equation}
\end{thm}

As it was already pointed out $\left|s_{\pm}\right|\neq1$ determines
the onset of the MCK instability. According to Theorem \ref{thm:floqmultbf}
the absolute value of each of the Floquet multipliers $\left|s_{\pm}\left(f\right)\right|$
is $\pi$-periodic functions of $f$. Consequently, if we are interested
in smaller values of normalized period $f=\frac{2\pi a}{\lambda_{\mathrm{rp}}}$
, the parameter that effects the instability, we may impose the following
assumption.

\begin{assumption} \label{ass:fpi}(smaller MCK period). The MCK
normalized period $f$ satisfies the following bounds:
\begin{equation}
0<f=\frac{2\pi a}{\lambda_{\mathrm{rp}}}<\pi.\label{eq:bfspm2a}
\end{equation}

\end{assumption}

\section{Instability parameter and instability frequencies\label{sec:instfreq}}

We assume here that Assumption \ref{ass:fpi}, that is $0<f<\pi$,
holds. As to the MCK instability according to Theorem \ref{thm:floqmultbf}
its presence is determined entirely by the instability parameter $b_{f}\left(\omega\right)$
defined by equations (\ref{eq:bfKpsi1b}). More precisely, the instability
occurs if and only if $\left|b_{f}\left(\omega\right)\right|>1$ and
we want to identify all points $\left(f,\omega\right)$ when it is
the case and Figure \ref{fig:Om-plusmin} illustrates ultimate results
of our analysis of the matter.

When searching for all points $\left(f,\omega\right)$ such that $\left|b_{f}\left(\omega\right)\right|>1$
we want to identify first all points $\left(f,\omega\right)$ for
which $b_{f}\left(\omega\right)=\pm1$, that is
\begin{equation}
b_{f}\left(\omega\right)=K_{0}\frac{\omega^{2}}{\omega^{2}-\omega_{0}^{2}}\sin\left(f\right)-\cos\left(f\right)=\pm1.\label{eq:bfKom1a}
\end{equation}
To separate out variables $f$ and $\omega$ in equations (\ref{eq:bfKom1a})
we recast them into
\begin{equation}
r\left(\omega\right)=\frac{\omega^{2}}{\omega^{2}-\omega_{0}^{2}}=\frac{1+\cos\left(f\right)}{K_{0}\sin\left(f\right)}=\frac{1}{K_{0}\tan\left(\frac{f}{2}\right)},\quad\omega>0,\quad0<f<\pi,\label{eq:bfKom1b}
\end{equation}
\begin{equation}
r\left(\omega\right)=\frac{\omega^{2}}{\omega^{2}-\omega_{0}^{2}}=\frac{-1+\cos\left(f\right)}{K_{0}\sin\left(f\right)}=-\frac{\tan\left(\frac{f}{2}\right)}{K_{0}},\quad\omega>0,\quad0<f<\pi.\label{eq:bfKom1c}
\end{equation}
Note that function $r\left(\omega\right)$ in equations (\ref{eq:bfKom1b})
and (\ref{eq:bfKom1c}) has the following properties: (i) it is a
monotonically decreasing function of $\omega\geq0$ with a simple
pole at $\omega=\omega_{0}$; (ii) it maps one-to-one interval $\left(\omega_{0},+\infty\right)$
onto $\left(1,+\infty\right)$ and interval $\left(0,\omega_{0}\right)$
onto $\left(-\infty,0\right)$. The monotonicity of $r\left(\omega\right)$
and expressions for $b_{f}\left(\omega\right)$ in equations (\ref{eq:bfKom1a})
and $b_{f}^{\infty}$ in equations (\ref{eq:bfKinf1}) readily imply
the following low bound
\begin{equation}
b_{f}\left(\omega\right),b_{f}^{\infty}>-\cos\left(f\right)>-1,\quad0<f<\pi,\quad\omega>\omega_{0}.\label{eq:bfomf1a}
\end{equation}
Equations (\ref{eq:bfKinf1}) for $b_{f}^{\infty}$ imply also the
following factorization

\begin{equation}
b_{f}^{\infty}-1=K_{0}\sin\left(f\right)-\cos\left(f\right)-1=K_{0}\sin\left(f\right)\left[1-\frac{1}{K_{0}\tan\left(\frac{f}{2}\right)}\right].\label{eq:bfomf1b}
\end{equation}
The high-frequency limit $b_{f}^{\infty}$ is of importance in the
analysis of the MCK instabilities and its significant properties are
collected in the following statement.
\begin{thm}[the high-frequency limit of the instability parameter]
\label{thm:bfinf} Let the high-frequency limit $b_{f}^{\infty}$
of the instability parameter be defined by equations (\ref{eq:bfKinf1}).
Then there exists a unique value $f_{\mathrm{cr}}$ on interval $\left(0,\pi\right)$
of the normalized period $f$ such that

\begin{equation}
b_{f_{\mathrm{cr}}}^{\infty}=1,\quad0<f_{\mathrm{cr}}<\pi,\label{eq:bfomf1c}
\end{equation}
and we refer to it as the\emph{ critical value} and the following
representation holds
\begin{equation}
f_{\mathrm{cr}}=2\arctan\left(\frac{1}{K_{0}}\right),\quad\text{ where }K_{0}=\frac{b^{2}g_{\mathrm{B}}}{c_{0}},\quad g_{\mathrm{B}}=\frac{\sigma_{\mathrm{B}}}{4\lambda_{\mathrm{rp}}},\quad\left|\arctan\left(*\right)\right|<\frac{\pi}{2}.\label{eq:bfomf1d}
\end{equation}
The following identities hold for $f_{\mathrm{cr}}$:
\begin{equation}
\tan\left(\frac{f_{\mathrm{cr}}}{2}\right)=\frac{1}{K_{0}},\quad\sin\left(\frac{f_{\mathrm{cr}}}{2}\right)=\frac{1}{\sqrt{1+K_{0}^{2}}},\quad\cos\left(\frac{f_{\mathrm{cr}}}{2}\right)=\frac{K_{0}}{\sqrt{1+K_{0}^{2}}}.\label{eq:bfomf1e}
\end{equation}
Figure \ref{fig:mck-fcr} shows how $f_{\mathrm{cr}}=2\arctan\left(\frac{1}{K_{0}}\right)$
varies with $K_{0}$. 

The high-frequency limit $b_{f}^{\infty}$ can be alternatively represented
by the following equations:
\begin{equation}
b_{f}^{\infty}=\sqrt{1+K_{0}^{2}}\sin\left(f-\frac{f_{\mathrm{cr}}}{2}\right)=\sqrt{1+K_{0}^{2}}\sin\left(f-\arctan\left(\frac{1}{K_{0}}\right)\right),\label{eq:bfomf1f}
\end{equation}
\begin{equation}
b_{f}^{\infty}=-\sqrt{1+K_{0}^{2}}\cos\left(f+\arctan\left(K_{0}\right)\right),\quad\left|\arctan\left(*\right)\right|<\frac{\pi}{2}.\label{eq:bfomf1g}
\end{equation}
In addition to that $b_{f}^{\infty}$ satisfies the following relations
\begin{equation}
b_{f}^{\infty}>-1,\quad0<f<\pi;\label{eq:bfomf2a}
\end{equation}
\begin{equation}
b_{f}^{\infty}<1,\quad0<f<f_{\mathrm{cr}};\quad b_{f_{\mathrm{cr}}}^{\infty}=1;\quad b_{f}^{\infty}>1,\quad f_{\mathrm{cr}}<f<\pi.\label{eq:bfomf2b}
\end{equation}
Importantly, the MCK normalized period $f_{\mathrm{cr}}$ signifies
the onset of the MCK instability for all frequencies $\omega>\omega_{0}$,
that is for $f_{\mathrm{cr}}<f<\pi$ the MCK system is unstable for
all $\omega>\omega_{0}$, see Fig. \ref{fig:mck-gain1}.
\end{thm}

\begin{proof}
Note that equation (\ref{eq:bfomf1c}) can be recast as
\begin{equation}
b_{f}^{\infty}-1=K_{0}\sin\left(f\right)c\left(f\right),\quad c\left(f\right)=\left[1-\frac{1}{K_{0}\tan\left(\frac{f}{2}\right)}\right],\quad0<f<\pi.\label{eq:bfomf2c}
\end{equation}
Since $\sin\left(f\right)\neq0$ for $0<f<\pi$ equation (\ref{eq:bfomf2c})
is equivalent to 
\begin{equation}
c\left(f\right)=1-\frac{1}{K_{0}\tan\left(\frac{f}{2}\right)}=0,\quad0<f<\pi.\label{eq:bfomf2d}
\end{equation}
Function $c\left(f\right)$ is monotonically increasing on interval
$\left(0,\pi\right)$ since
\begin{equation}
\partial_{f}c\left(f\right)=\frac{1}{2K_{0}\sin^{2}\left(\frac{f}{2}\right)}>0,\quad0<f<\pi,\label{eq:bfomf2e}
\end{equation}
and it varies from $-\infty$ to $1$ on this interval. Hence equation
(\ref{eq:bfomf2d}) and consequently equation (\ref{eq:bfomf1b})
have a unique solution $f_{\mathrm{cr}}$ on interval $\left(0,\pi\right)$
satisfying equation (\ref{eq:bfomf1d}). 

The validity of equations (\ref{eq:bfomf1e})-(\ref{eq:bfomf1g})
can be verified by carrying out elementary trigonometric transformations
of involved expressions. Inequality (\ref{eq:bfomf2a}) readily follows
from the first equation in (\ref{eq:bfomf1b}) $\sin\left(f\right)>0$
for $0<f<\pi$. Relations (\ref{eq:bfomf2a}) follow straightforwardly
from equations (\ref{eq:bfomf2c}) and (\ref{eq:bfomf2d}) and already
established monotonicity of function $c\left(f\right)$.
\end{proof}
The following asymptotic formulas hold for $f_{\mathrm{cr}}$ defined
by equations (\ref{eq:bfomf1d}) :
\begin{equation}
f_{\mathrm{cr}}=\pi-2K_{0}+\frac{2K_{0}^{3}}{3}+O\left(K_{0}^{5}\right),\quad K_{0}\rightarrow0,\label{eq:bfomf2f}
\end{equation}
\begin{equation}
f_{\mathrm{cr}}=\frac{2}{K_{0}}-\frac{2}{3K_{0}^{3}}+O\left(K_{0}^{5}\right),\quad K_{0}\rightarrow+\infty.\label{eq:befomf2g}
\end{equation}
\begin{figure}[h]
\begin{centering}
\includegraphics[scale=0.23]{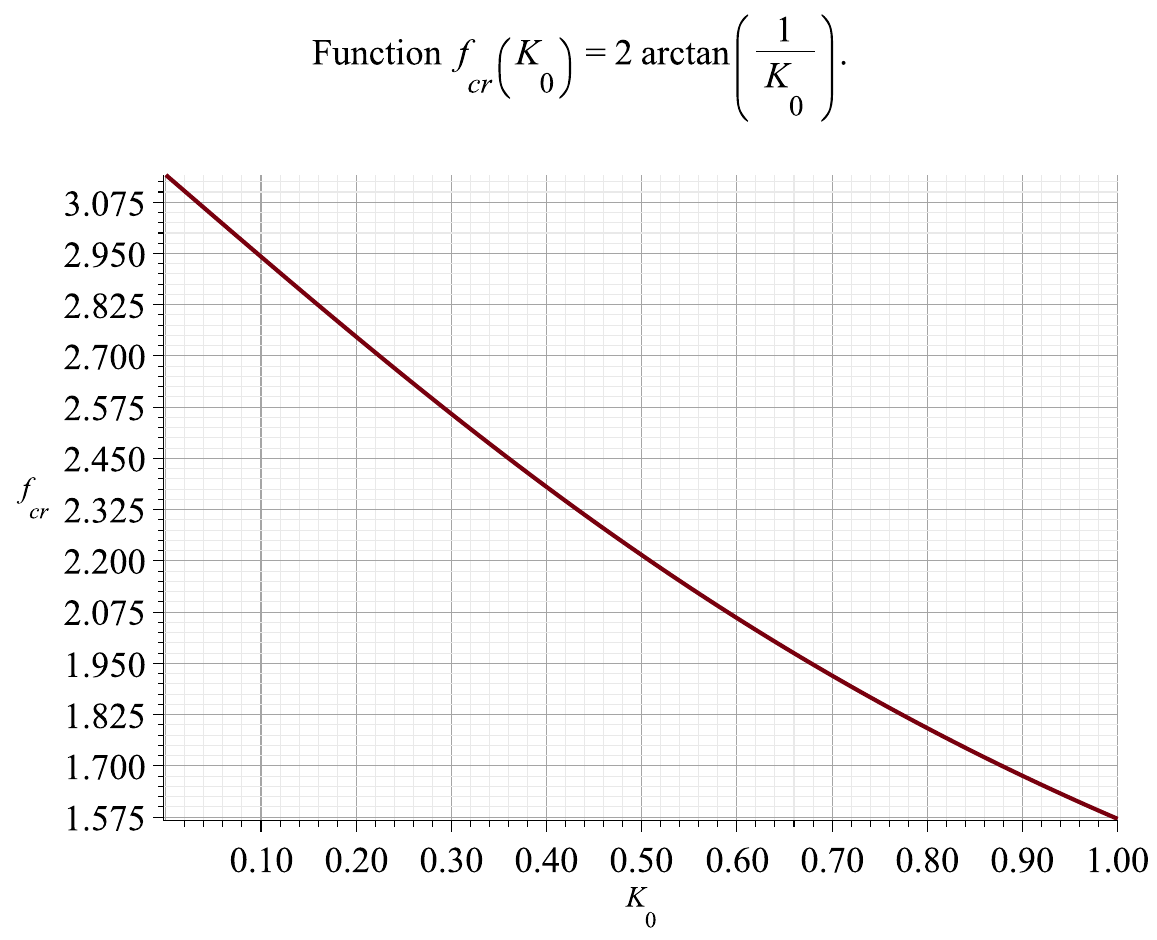}\hspace{0.4cm}\includegraphics[scale=0.23]{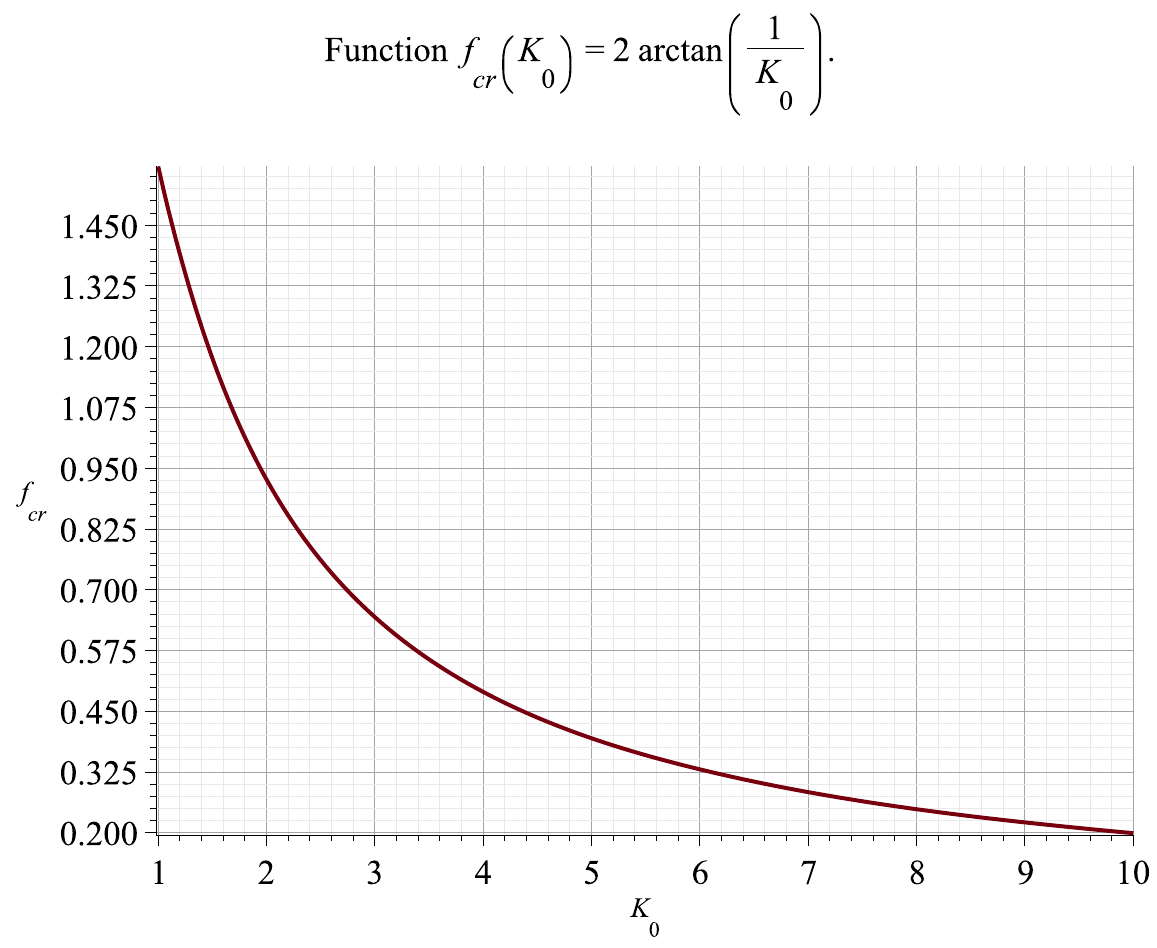}\hspace{0.4cm}\includegraphics[scale=0.25]{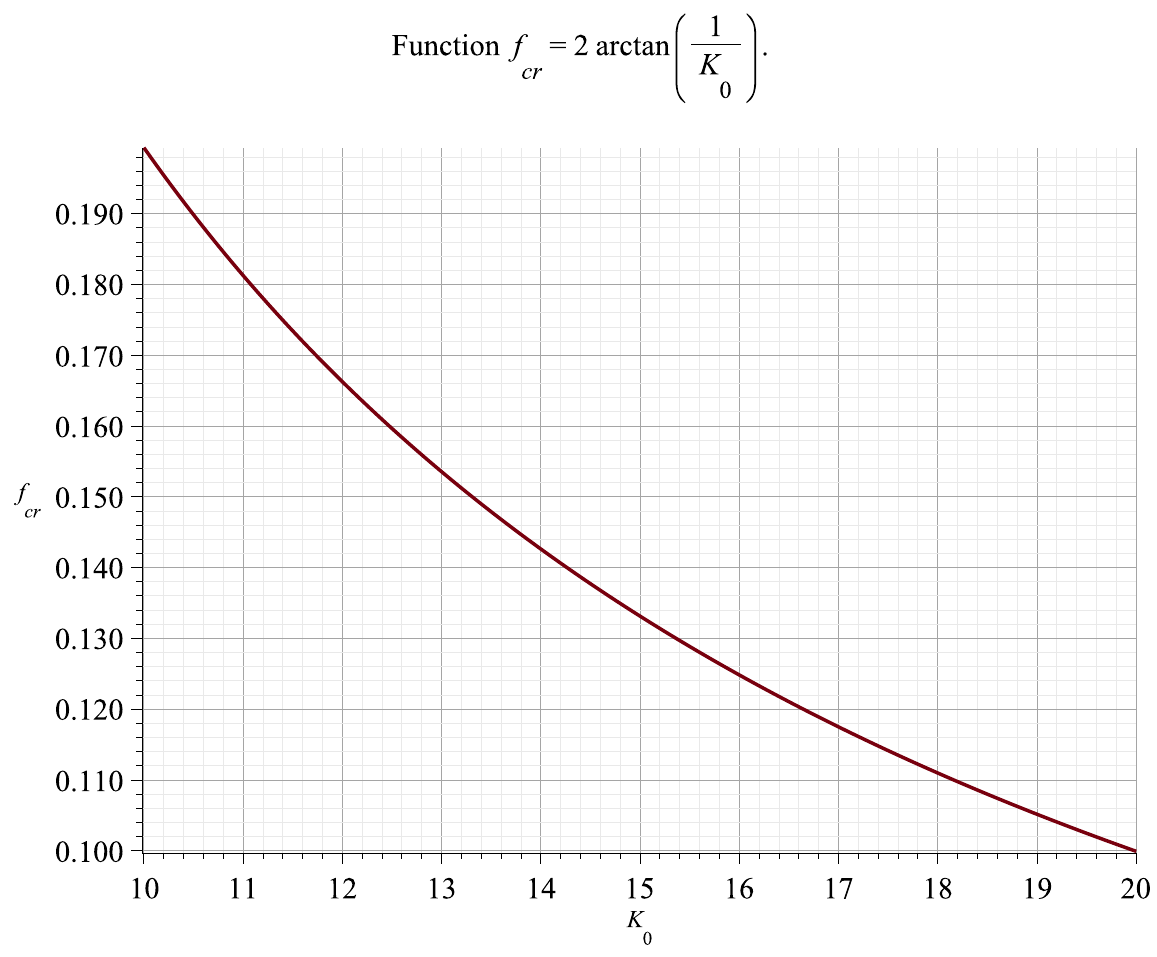}
\par\end{centering}
\centering{}(a)\hspace{4.5cm}(b)\hspace{4.5cm}(c)\caption{\label{fig:mck-fcr} Plots of $f_{\mathrm{cr}}=2\arctan\left(\frac{1}{K_{0}}\right)$
for different ranges of $K_{0}$: (a) $0\protect\leq K_{0}\protect\leq1$;
(b) $1\protect\leq K_{0}\protect\leq10$; (c) $10\protect\leq K_{0}\protect\leq20$.
In all plots the horizontal and vertical axes represent respectively
$K_{0}$ and $f_{\mathrm{cr}}$.}
\end{figure}

The next statement specifies the set of points $\left(f,\omega\right)$
associated with the instability, namely the points for which $\left|b_{f}\left(\omega\right)\right|>1$.
\begin{thm}[instability frequencies]
\label{thm:Omegapm} Let functions $\varOmega_{f}^{\pm}>0$ of $f$
for $0<f<\pi$ be defined by the following relations:
\begin{gather}
\varOmega_{f}^{+}=\left\{ \begin{array}{ccc}
\omega_{0}\sqrt{\frac{1}{1-K_{0}\tan\left(\frac{f}{2}\right)}}>\omega_{0} & \text{if} & 0<f<f_{\mathrm{cr}},\\
+\infty & \text{if} & f_{\mathrm{cr}}\leq f<\pi
\end{array}\right.,\label{eq:Ompom1a}
\end{gather}
\begin{equation}
\varOmega_{f}^{-}=\omega_{0}\sqrt{\frac{\tan\left(\frac{f}{2}\right)}{\tan\left(\frac{f}{2}\right)+K_{0}}}<\omega_{0},\quad0<f<\pi.\label{eq:Ompom1b}
\end{equation}
Then the values of the instability parameter $b_{f}\left(\omega\right)$
satisfy the following relations:
\begin{gather}
\omega>\omega_{0},\quad0<f<f_{\mathrm{cr}}:\quad b_{f}\left(\varOmega_{f}^{+}\right)=1;\label{eq:Ompom1c}\\
b_{f}\left(\omega\right)>1,\quad\omega_{0}<\omega<\varOmega_{f}^{+};\quad-1<b_{f}\left(\omega\right)<1,\quad\omega>\varOmega_{f}^{+};\nonumber 
\end{gather}
\begin{equation}
\omega>\omega_{0},\quad f_{\mathrm{cr}}\leq f<\pi:\quad b_{f}\left(\omega\right)>b_{f}^{\infty}>1;\label{eq:Ompom1d}
\end{equation}
\begin{gather}
0<\omega<\omega_{0},\quad0<f<\pi:\quad b_{f}\left(\varOmega_{f}^{-}\right)=-1;\label{eq:Ompom1e}\\
-1<b_{f}\left(\omega\right)<1,\quad0<\omega<\varOmega_{f}^{-};\quad b_{f}\left(\omega\right)<-1,\quad\varOmega_{f}^{-}<\omega<\omega_{0}.\nonumber 
\end{gather}
Hence, for any $0<f<\pi$ interval $\left(\varOmega_{f}^{-},\varOmega_{f}^{+}\right)$
identifies frequencies of the MCK instability and gain, that is
\begin{equation}
\left(\varOmega_{f}^{-},\varOmega_{f}^{+}\right)=\left\{ \omega:\;\left|b_{f}\left(\omega\right)\right|>1\right\} ,\quad0<f<\pi,\label{eq:Ompom3a}
\end{equation}
and, in particular, for $f_{\mathrm{cr}}\leq f<\pi$ the frequency
instability interval extends to infinity, that is
\begin{equation}
\left(\varOmega_{f}^{-},\varOmega_{f}^{+}\right)=\left(\varOmega_{f}^{-},+\infty\right),\quad f_{\mathrm{cr}}\leq f<\pi.\label{eq:Ompom3b}
\end{equation}
Figure \ref{fig:Om-plusmin} provides for graphical representation
of functions $\varOmega_{f}^{\pm}>0$ with shadowed area identifying
points $\left(f,\omega\right)$ of instability where $\left|b_{f}\left(\omega\right)\right|>1$.
\end{thm}

\begin{proof}
Let us start with points $\left(f,\omega\right)$ with $\omega>\omega_{0}$.
Note that in this case $b_{f}\left(\omega\right)$ is a monotonically
decreasing function of $\omega$ on interval $\left(\omega_{0},+\infty\right)$
for any $0<f<\pi$. Consequently the set of its values on interval
$\left(\omega_{0},+\infty\right)$ satisfies
\begin{equation}
b_{f}\left(\left(\omega_{0},+\infty\right)\right)=\left(b_{f}^{\infty},+\infty\right).\label{eq:Ompom2a}
\end{equation}
Combining relation (\ref{eq:Ompom2a}) with the results of Theorem
\ref{thm:bfinf} we obtain relations (\ref{eq:Ompom1d}).

Let us consider now points $\left(f,\omega\right)$ with $0<\omega<\omega_{0}$
and $0<f<\pi$. In this case $b_{f}\left(\omega\right)$ is also a
monotonically decreasing function of $\omega$ on interval $\left(0,\omega_{0}\right)$
for any $0<f<\pi$. Consequently the set of its values on interval
$\left(0,\omega_{0}\right)$ satisfies
\begin{equation}
b_{f}\left(\left(0,\omega_{0}\right)\right)=\left(-\cos\left(f\right),-\infty\right).\label{eq:Ompom2b}
\end{equation}
Combining relation (\ref{eq:Ompom2b}) with the results of Theorem
\ref{thm:bfinf} we obtain relations (\ref{eq:Ompom1e}). Representations
(\ref{eq:Ompom3a}) and (\ref{eq:Ompom3b}) for the frequency instability
interval $\left(\varOmega_{f}^{-},\varOmega_{f}^{+}\right)$ follow
from equations (\ref{eq:Ompom1a}) for $\varOmega_{f}^{+}$ and relations
(\ref{eq:Ompom1d}).
\end{proof}
Note that function $\varOmega_{f}^{+}$ approaches $+\infty$ as $f\rightarrow f_{\mathrm{cr}}$
and the following asymptotic formula can be obtained. Combing equation
(\ref{eq:bfomf1d}) defining $f_{\mathrm{cr}}$ and equations (\ref{eq:Ompom1a})
defining $\varOmega_{f}^{+}$ we obtain the following alternative
representation for function $\varOmega_{f}^{+}$:
\begin{equation}
\varOmega_{f}^{+}=\frac{\omega_{0}}{\sqrt{K_{0}}}\sqrt{\frac{1}{\tan\left(\frac{f_{\mathrm{cr}}}{2}\right)-\tan\left(\frac{f}{2}\right)}}=\frac{\omega_{0}}{\sqrt{K_{0}}}\sqrt{\frac{\cos\left(\frac{f_{\mathrm{cr}}}{2}\right)\cos\left(\frac{f}{2}\right)}{\sin\left(\frac{f_{\mathrm{cr}}-f}{2}\right)}},\quad0<f<f_{\mathrm{cr}}.\label{eq:hpmK2c}
\end{equation}
Then using representation (\ref{eq:hpmK2c}) we find the following
asymptotic expansion of $\varOmega_{f}^{+}$ at $f_{\mathrm{cr}}$:
\begin{gather}
\varOmega_{f}^{+}=\frac{\omega_{0}}{\sqrt{K_{0}}}\left[\frac{\left|\cos\left(\frac{f_{\mathrm{cr}}}{2}\right)\right|}{\sqrt{\frac{f_{\mathrm{cr}}-f}{2}}}+\frac{\mathrm{sign}\,\left\{ \cos\left(\frac{f_{\mathrm{cr}}}{2}\right)\right\} \sin\left(\frac{f_{\mathrm{cr}}}{2}\right)}{2}\sqrt{\frac{f_{\mathrm{cr}}-f}{2}}+O\left(\left(f_{\mathrm{cr}}-f\right)^{\frac{3}{2}}\right)\right],\quad f\rightarrow f_{\mathrm{cr}}.\label{eq:hpmK2d}
\end{gather}
\begin{figure}[h]
\begin{centering}
\includegraphics[scale=0.35]{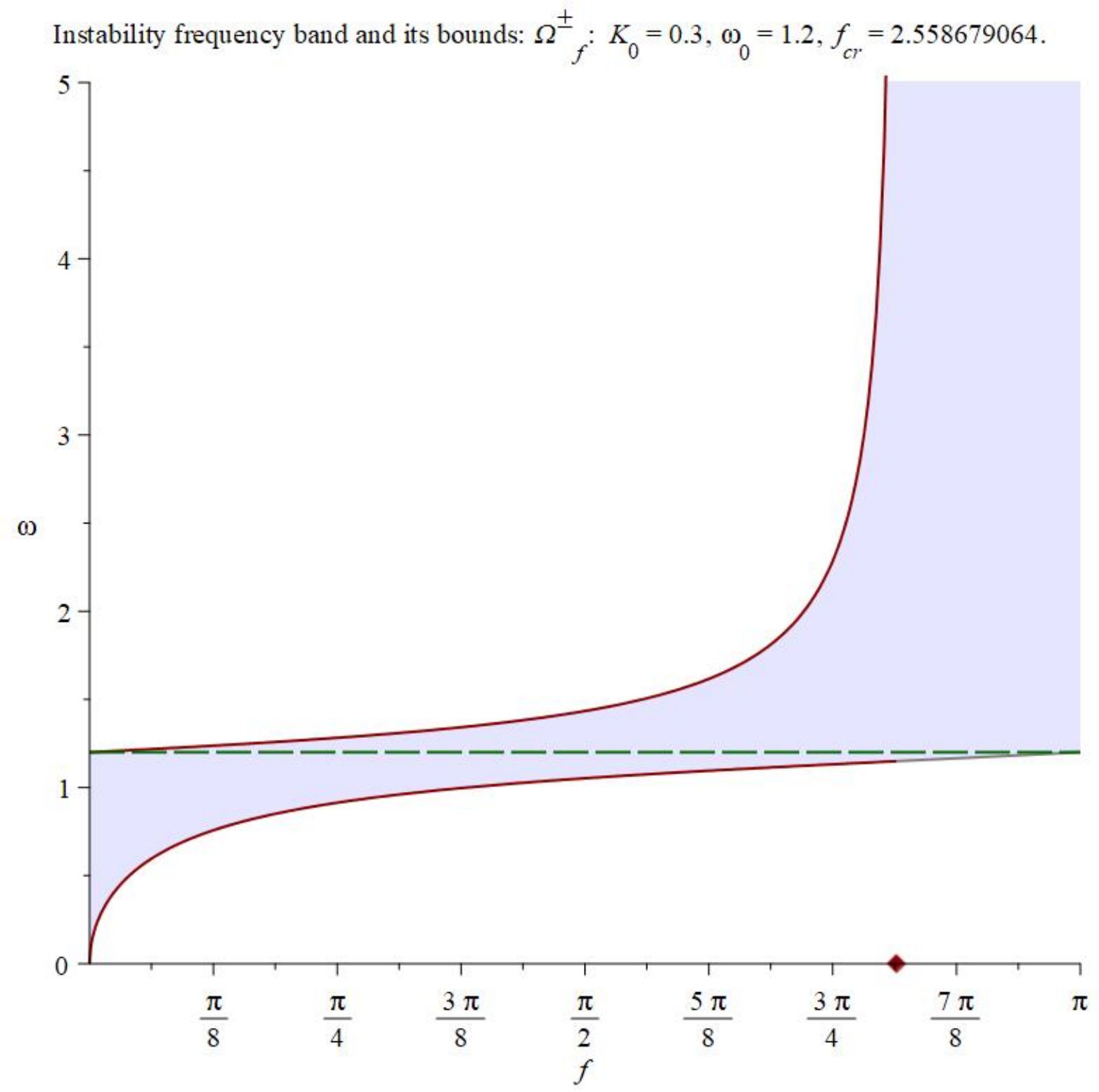}
\par\end{centering}
\centering{}\caption{\label{fig:Om-plusmin} Plots of two functions $\omega=\varOmega_{f}^{\pm}$:
$\varOmega_{f}^{-}<\omega_{0}<\varOmega_{f}^{+}$ defined by equations
(\ref{eq:Ompom1a}) for the case when $K_{0}=0.3$, $\omega_{0}=1.2$,
and $f_{\mathrm{cr}}\protect\cong2.558679064$: $\varOmega_{f}^{+}$
for $0<f<f_{\mathrm{cr}}$, $\varOmega_{f}^{-}$ for $0<f<\pi$. The
horizontal and vertical axes represent respectively variables $f$
and $\omega$. The upper and lower solid (brown) curves represent
respectively functions $\varOmega_{f}^{+}$ and $\varOmega_{f}^{-}$
and the dashed (green) line represent $\omega=\omega_{0}=1.2$. The
diamond solid (brown) dot marks the value of $f_{\mathrm{cr}}$. The
shaded (light blue) region between the two curves for functions $\varOmega_{f}^{+}$
and $\varOmega_{f}^{-}$ identify points $\left(f,\omega\right)$
of instability, that is points for which $\left|b_{f}\left(\omega\right)\right|>1$
and consequently the corresponding Floquet multipliers $s$ satisfy
$\left|s_{+}\right|>1>\left|s_{-}\right|$ (see Theorem \ref{thm:floqmultbf}).
For the points of instability the relevant Floquet modes either grow
or decay exponentially. The remaining points correspond to the case
when $\left|b_{f}\left(\omega\right)\right|\protect\leq1$. In the
later case the Floquet multipliers $s$ satisfy $\left|s_{\pm}\right|=1$
and the corresponding Floquet modes are bounded and oscillatory. The
points $\left(f,\varOmega_{f}^{+}\right)$ for $0<f<f_{\mathrm{cr}}$
correspond to $b_{f}=1$ and $\left(f,\varOmega_{f}^{-}\right)$ for
$0<f<\pi$ correspond to $b_{f}=-1$. Points $\left(f,\omega_{0}\pm0\right)$
for $0<f<\pi$ laying on the dashed (green) line correspond respectively
to $b_{f}=\pm\infty$.}
\end{figure}

Combining statements of Theorems \ref{thm:Omegapm}, \ref{thm:bfinf}
and \ref{thm:floqmultbf}, particularly equations (\ref{eq:Ompom3a}),
(\ref{eq:Ompom3a}) and relations (\ref{eq:bfomf2b}), (\ref{eq:bfspm1b}),
we obtain the following statement.
\begin{thm}[Floquet multiplier and instability]
\label{thm:floqinst} Let Floquet multiplier $s_{+}=s_{+}\left(f,\omega\right)$
be defined by equations \ref{eq:bfspm1b}. Then function $\left|s_{+}\left(f,\omega\right)\right|>1$
if and only if $\varOmega_{f}^{-}<\omega<\varOmega_{f}^{+}$ for $0<f<\pi$
. The later relations describe all unstable states of the MCK.

For any $0<f<\pi$ function $\left|s_{+}\left(f,\omega\right)\right|$
is monotonically increasing for $\varOmega_{f}^{-}<\omega<\omega_{0}$
and monotonically decreasing for $\omega_{0}<\omega<\varOmega_{f}^{+}$
and $\lim_{\omega\rightarrow\omega_{0}}\left|s_{+}\left(f,\omega\right)\right|=+\infty$.
In addition to that, the following lower bound holds:
\begin{equation}
\left|s_{+}\left(f,\omega\right)\right|>b_{f}^{\infty}>1,\quad f_{\mathrm{cr}}<f<\pi,\quad\omega>\omega_{0},\label{eq:hpmK2e}
\end{equation}
where high-frequency limit instability $b_{f}^{\infty}$ of the instability
parameter is defined by equations (\ref{eq:bfKinf1}).
\end{thm}

\section{Gain and its dependence of the frequency\label{sec:gain-freq}}

Based on the prior analysis we introduce the MCK gain $G$ in $\mathrm{dB}$
per one period as a the rate of the exponential growth of the MCK
eigenmodes associated with Floquet multipliers $s_{\pm}$ defined
by equations (\ref{eq:bfKpsi1bb}) (see Theorem \ref{thm:floqmultbf}).
More precisely the definition is as follows.
\begin{defn}[MCK gain per one period]
\label{def:gain} Let $s_{\pm}$ be the MCK Floquet multipliers satisfying
by equations (\ref{eq:bfspm1b})-(\ref{eq:bfspm1d}). Then the corresponding
to them gain $G$ in $\mathrm{dB}$ per one period is defined by
\begin{gather}
G=G\left(f,\omega,K_{0}\right)=\left\{ \begin{array}{ccc}
20\left|\log\left(\left|s_{+}\right|\right)\right|=20\left|\log\left(\left|\left|b_{f}\right|+\sqrt{b_{f}^{2}-1}\right|\right)\right| & \text{if} & \left|b_{f}\right|>1\\
0 & \text{if} & \left|b_{f}\right|\leq1
\end{array}\right.,\label{eq:sSbf1c}\\
b_{f}=b_{f}\left(\omega\right)=K\left(\omega\right)\sin\left(f\right)-\cos\left(f\right),\quad K\left(\omega\right)=K_{0}\frac{\omega^{2}}{\omega^{2}-\omega_{0}^{2}},\quad K_{0}=\frac{b^{2}\beta_{0}}{2f}=\frac{b^{2}g_{\mathrm{B}}}{c_{0}}.\nonumber 
\end{gather}
\end{defn}

Fig. \ref{fig:mck-gain-om1} shows the frequency dependence of the
gain $G$ per one period which is consistent with statements of Theorem
\ref{thm:floqinst} including that $G\left(f,\omega,K_{0}\right)$
is a monotonically increasing and decreasing function of $\omega$
on respective intervals $\left(\varOmega_{f}^{-},\omega_{0}\right)$
and $\left(\omega_{0},\varOmega_{f}^{+}\right)$.
\begin{figure}[h]
\begin{centering}
\hspace{-0.2cm}\includegraphics[scale=0.38]{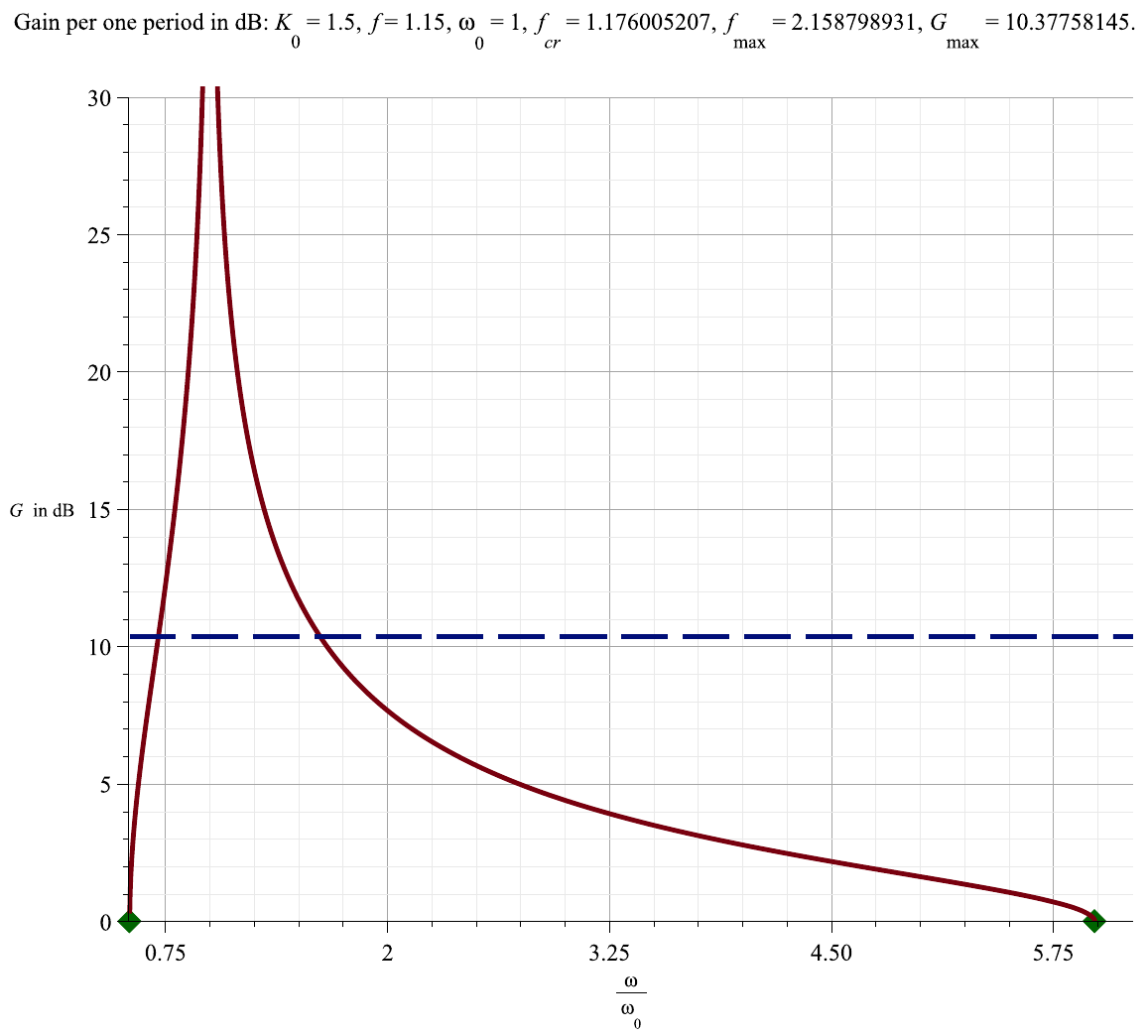}\hspace{0.1cm}\includegraphics[scale=0.38]{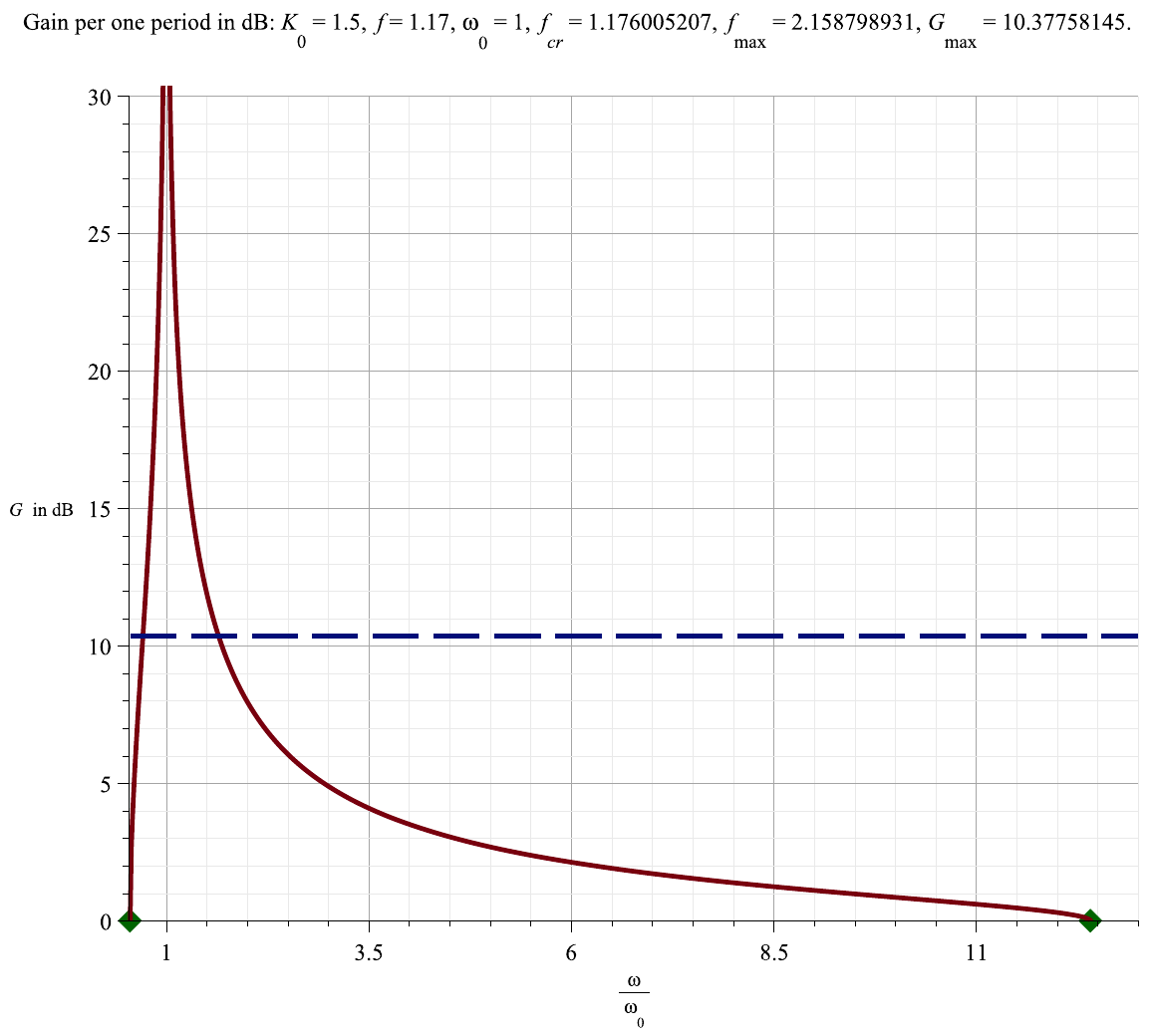}
\par\end{centering}
\centering{}(a)\hspace{6cm}(b)\caption{\label{fig:mck-gain-om1} Plots of gain $G$ as a function of frequency
$\omega$ for $K_{0}=1$, $\omega_{0}=1$ for which $f_{\mathrm{cr}}\protect\cong1.176005207$,
$f_{\mathrm{max}}\protect\cong2.158798931$, $G_{\mathrm{max}}=10.3775845$
and for different values of $f$: (a) $f=1.15<f_{\mathrm{cr}}\protect\cong1.176005207$;
(b) $f=1.17<f_{\mathrm{cr}}\protect\cong1.176005207$. In all plots
the horizontal and vertical axes represent respectively frequency
$\omega$ and gain $G$ in $\mathrm{dB}$. The solid (brown) curves
represent gain $G$ as a function of frequency $\omega$, the dashed
(blue) line $G=G_{\mathrm{max}}$ represent the maximal value $G_{\mathrm{max}}$
of $G$ in the high frequency limit. The diamond solid (green) dots
mark the values of $\varOmega_{f}^{-}$ and $\varOmega_{f}^{+}$ which
are the frequency boundaries of the instability.}
\end{figure}
In view of Definition \ref{def:gain} a state of the MCK is unstable
and has positive gain $G>0$ if and only if $\left|b_{f}\left(\omega\right)\right|>1$.

\subsection{Maximal values of the gain\label{subsec:gain-hf}}

In Sections \ref{sec:floqmul} and \ref{sec:instfreq} we carried
out detailed studies of the MCK instability including its dependence
on frequency $\omega$ and dimensionless period $f$, see Theorems
\ref{thm:floqmultbf}, \ref{thm:Omegapm} and \ref{thm:floqinst}.
In particular Theorem \ref{thm:floqinst} implies the following sharp
lower bound holds for the gain $G$ per one period (see Fig. \ref{fig:mck-gain1}):
\begin{gather}
G=G\left(f,\omega,K_{0}\right)>G^{\infty}\left(f\right)=\lim_{\omega\rightarrow+\infty}20\left|\log\left(\left|s_{+}\right|\right)\right|=20\left|\log\left(\left|b_{f}^{\infty}+\sqrt{\left(b_{f}^{\infty}\right)^{2}-1}\right|\right)\right|>0,\label{eq:sSbf1d}\\
b_{f}^{\infty}=K_{0}\sin\left(f\right)-\cos\left(f\right)>1,\quad f_{\mathrm{cr}}<f<\pi,\quad\omega>\omega_{0}.\nonumber 
\end{gather}
Consequently, the maximal value of gain $G$ is attained when $b_{f}^{\infty}$
gets its maximal value for $f=f_{\mathrm{max}}$ such that $f_{\mathrm{cr}}<f_{\mathrm{max}}<\pi$.
To find $f_{\mathrm{max}}$ and the corresponding $G_{\mathrm{max}}$
we us representation (\ref{eq:bfomf1g}) for $b_{f}^{\infty}$ which
we copy here for the reader's convenience
\begin{equation}
b_{f}^{\infty}=-\sqrt{1+K_{0}^{2}}\cos\left(f+\arctan\left(K_{0}\right)\right),\quad f_{\mathrm{cr}}<f<\pi,\quad\left|\arctan\left(*\right)\right|<\frac{\pi}{2}.\label{eq:sSbf1e}
\end{equation}
An elementary trigonometric analysis of above representation shows
$b_{f}^{\infty}$ attains its maximal value of $\sqrt{1+K_{0}^{2}}$
at $f_{\mathrm{max}}$ which as follows

\begin{equation}
f_{\mathrm{max}}=\pi-\arctan\left(K_{0}\right),\quad b_{f_{\mathrm{max}}}^{\infty}=\sqrt{K_{0}^{2}+1}.\label{eq:sSbf2a}
\end{equation}
Consequently, the desired maximum value $G_{\mathrm{max}}$ of $G$
in view of relations (\ref{eq:sSbf1d}) is
\begin{equation}
G_{\mathrm{max}}=20\left|\log\left(\left|K_{0}+\sqrt{K_{0}^{2}+1}\right|\right)\right|=20\frac{\ln\left(K_{0}+\sqrt{K_{0}^{2}+1}\right)}{\ln\left(10\right)}.\label{eq:sSbf2b}
\end{equation}
Fig. \ref{fig:mck-gain1} shows the dependence of the gain $G$ on
frequency $\omega$ and its asymptotic behavior as $\omega\rightarrow+\infty$.
\begin{figure}[h]
\begin{centering}
\includegraphics[scale=0.35]{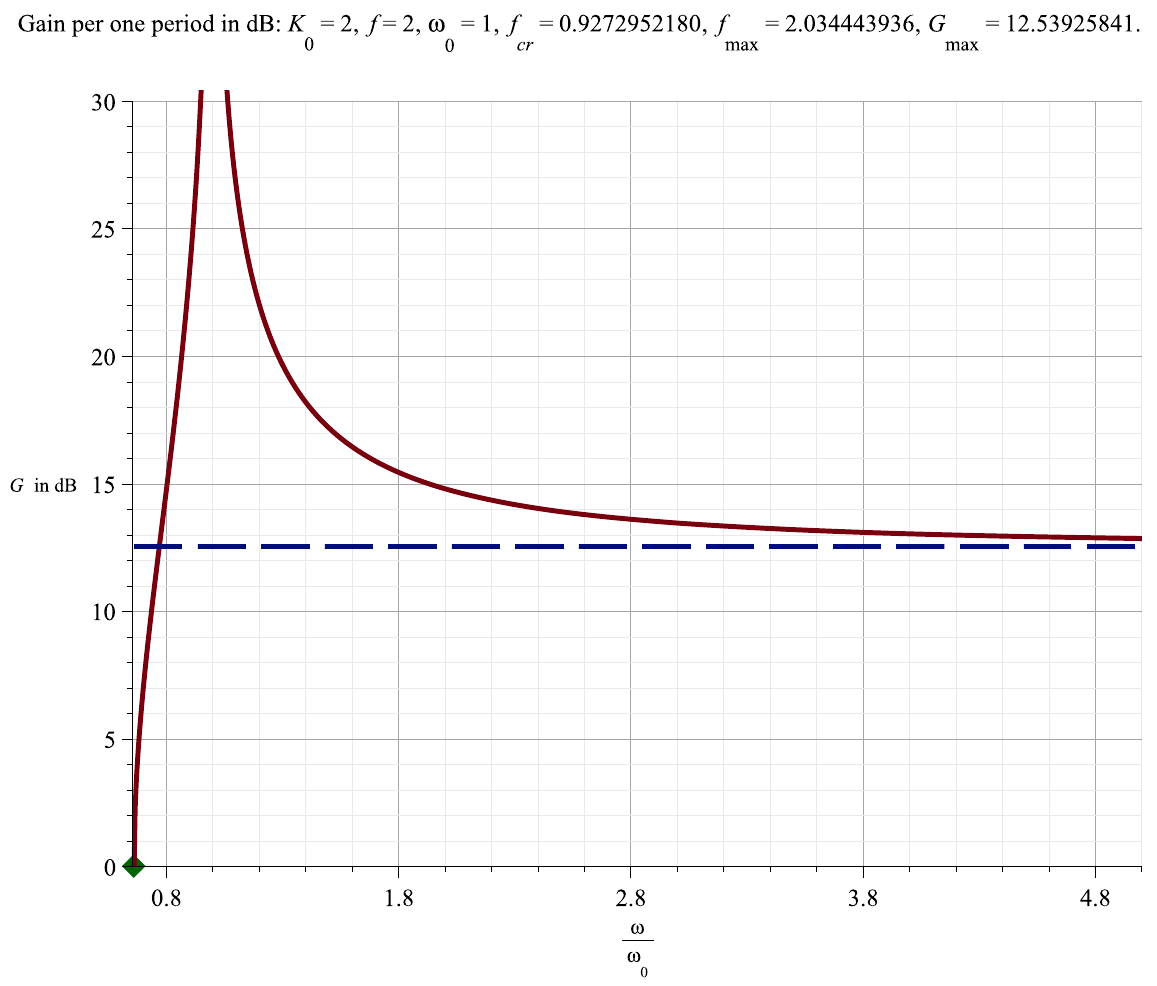}\hspace{1cm}\includegraphics[scale=0.36]{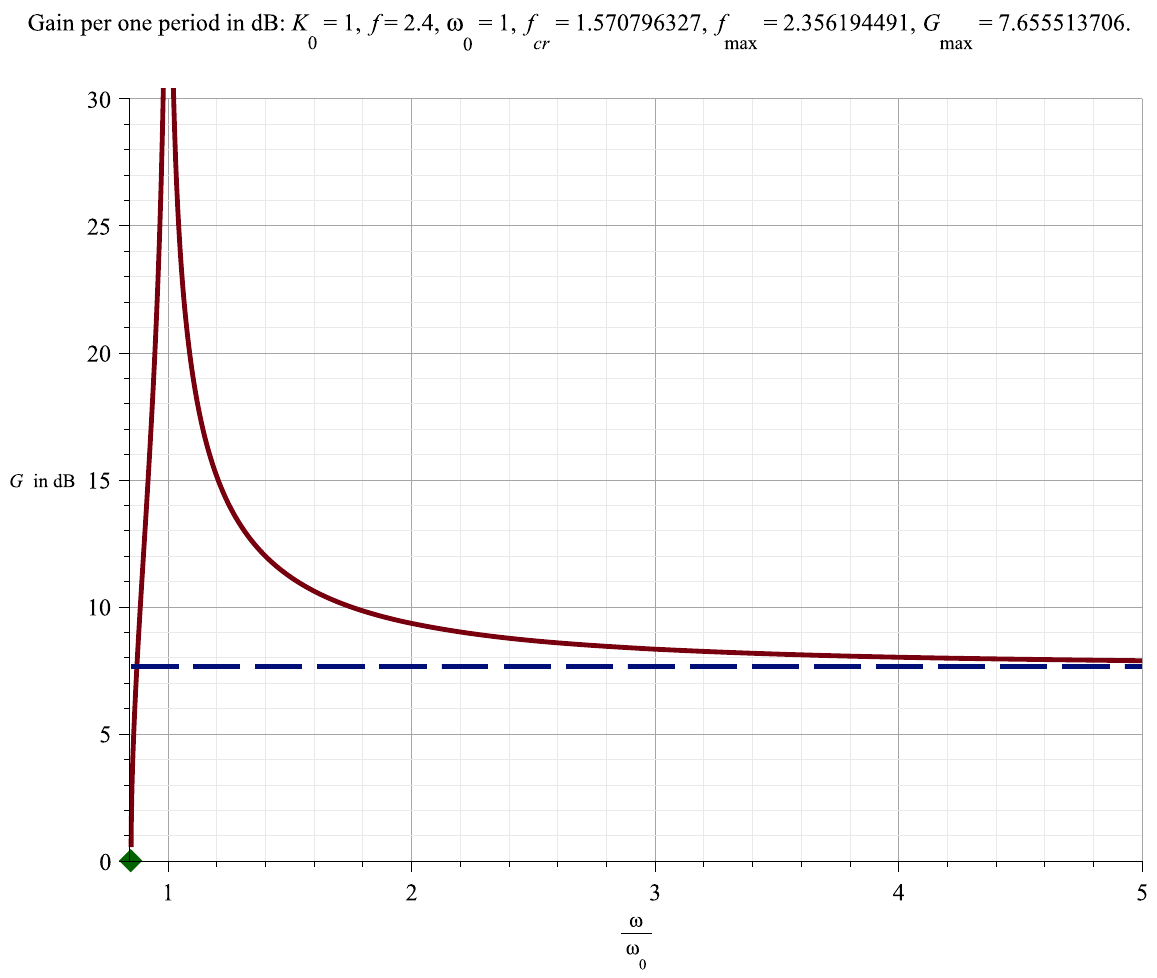}
\par\end{centering}
\centering{}\hspace{1cm}(a)\hspace{7.5cm}(b)\caption{\label{fig:mck-gain1} Plots of gain $G$ as a function of frequency
$\omega$ defined by equations (\ref{eq:sSbf1c}) for $\omega_{0}=1$
and (a) $K_{0}=2$, $f=2>f_{\mathrm{cr}}$ with $f_{\mathrm{cr}}\protect\cong0.9272952180$,
$f_{\mathrm{max}}\protect\cong2.034443936$ and $G_{\mathrm{max}}\protect\cong12.53925841$;
(b) $K_{0}=1$, $f=2.4>f_{\mathrm{cr}}$ with $f_{\mathrm{cr}}\protect\cong1.570796327$,
$f_{\mathrm{max}}\protect\cong2.356194491$ and $G_{\mathrm{max}}\protect\cong7.655513706$.
In all plots the horizontal and vertical axes represent respectively
frequency $\omega$ and gain $G$ in $\mathrm{dB}$. The solid (brown)
curves represent gain $G$ as a function of frequency $\omega$, the
dashed (blue) line $G=G_{\mathrm{max}}$ represent the maximal value
$G_{\mathrm{max}}$ of $G$ in the high frequency limit (see Section
\ref{subsec:gain-hf}). The diamond solid (green) dots mark the values
of $\varOmega_{f}^{-}$ which is the lower frequency boundary of the
instability interval.}
\end{figure}
In particular, we have the following asymptotic formulas for $f_{\mathrm{max}}$
and $G_{\mathrm{max}}$:
\begin{equation}
f_{\mathrm{max}}=\pi-K_{0}+\frac{K_{0}^{3}}{3}+O\left(K_{0}^{5}\right),\quad K_{0}\rightarrow0,\label{eq:sSbf2ca}
\end{equation}
\begin{equation}
f_{\mathrm{max}}=\frac{\pi}{2}+\frac{1}{K_{0}}-\frac{3}{K_{0}^{3}}+O\left(K_{0}^{5}\right),\quad K_{0}\rightarrow+\infty,\label{eq:sSbf2da}
\end{equation}
\begin{equation}
G_{\mathrm{max}}=\frac{20K_{0}}{\ln\left(10\right)}-\frac{10K_{0}^{3}}{3\ln\left(10\right)}+O\left(K_{0}^{5}\right)\cong8.685889638K_{0},\quad K_{0}\rightarrow+0\label{eq:sSbf2c}
\end{equation}
\begin{equation}
G_{\mathrm{max}}=\frac{20\ln\left(2K_{0}\right)}{\ln\left(10\right)}+\frac{5}{\ln\left(10\right)K_{0}^{2}}+O\left(\frac{1}{K_{0}^{4}}\right),\quad K_{0}\rightarrow+\infty.\label{eq:sSbf2d}
\end{equation}
Fig. \ref{fig:mck-gain1} shows the dependence of the gain $G_{\mathrm{max}}$
on $K_{0}$.
\begin{figure}[h]
\begin{centering}
\includegraphics[scale=0.35]{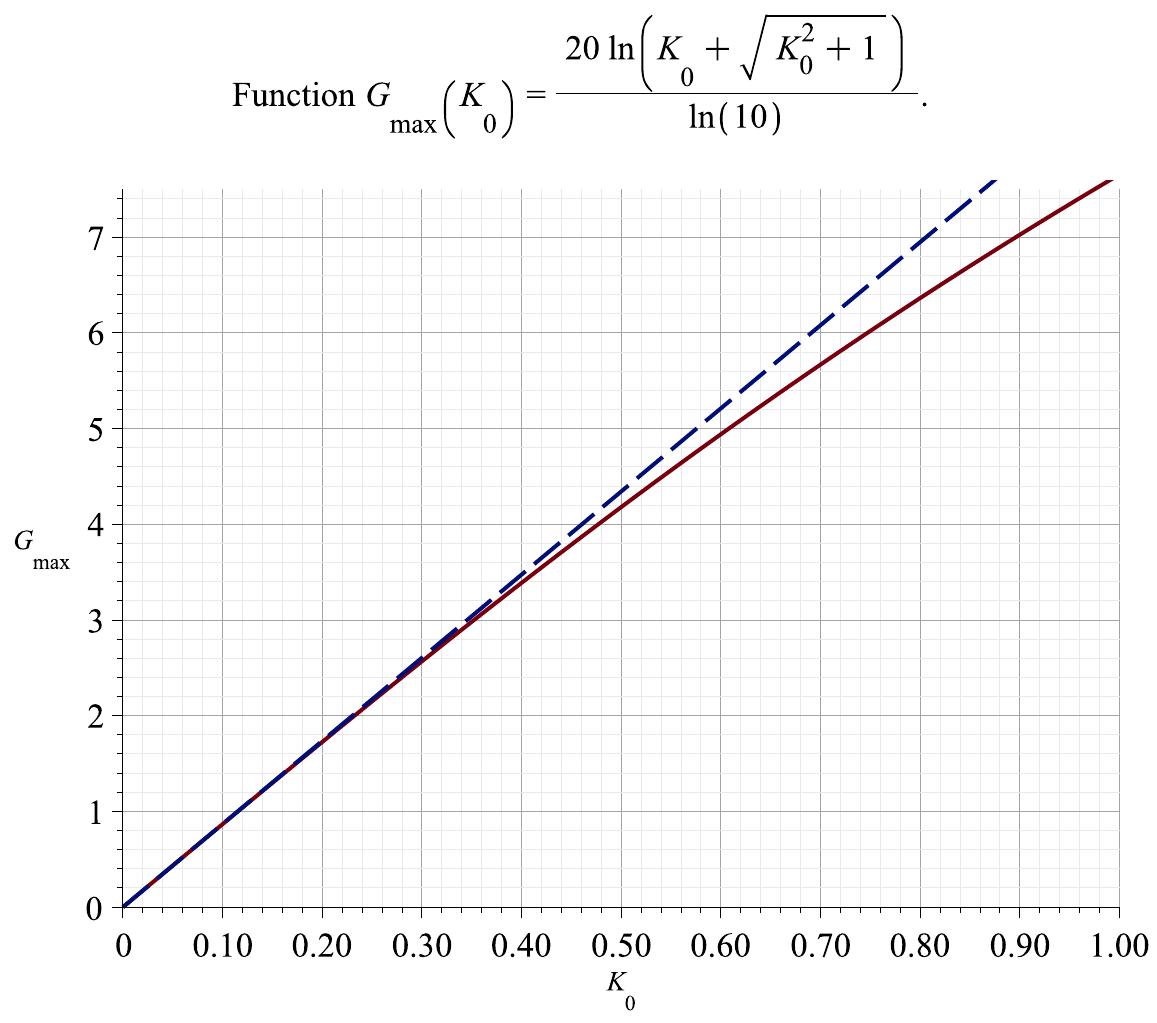}\hspace{1cm}\includegraphics[scale=0.35]{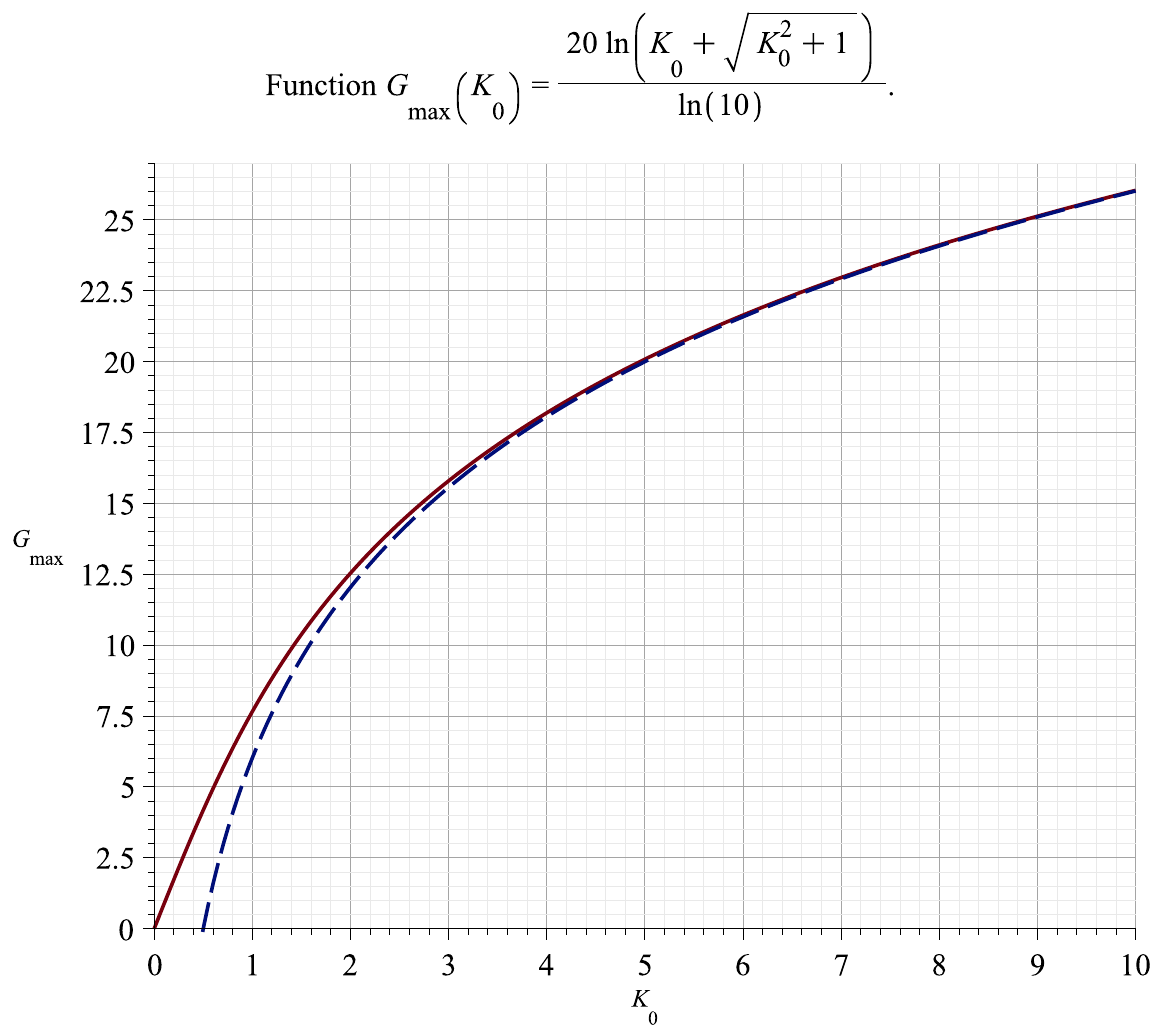}
\par\end{centering}
\centering{}(a)\hspace{7cm}(b)\caption{\label{fig:mck-gain-max} Plots of optimal gain $G_{\mathrm{max}}=20\log\left(K_{0}+\sqrt{K_{0}^{2}+1}\right)$
as solid (red) curves and and its approximations by leading terms
in asymptotic expansions (\ref{eq:sSbf2c}) and (\ref{eq:sSbf2c})
as dashed (blue) curves: (a) $0\protect\leq K_{0}\protect\leq1$;
(b) $0\protect\leq K_{0}\protect\leq10$.}
\end{figure}

It is also an elementary exercise in trigonometry to verify that the
following identity holds
\begin{equation}
f_{\mathrm{max}}=\frac{f_{\mathrm{cr}}+\pi}{2},\text{ implying }\frac{\pi}{2}<f_{\mathrm{max}}<\pi.\label{eq:sSbf2e}
\end{equation}
To find out how $f_{\mathrm{max}}-f_{\mathrm{cr}}$ depends on $K_{0}$
we use equations (\ref{eq:bfomf1d}) and (\ref{eq:sSbf2a}) that yield
\begin{equation}
f_{\mathrm{max}}-f_{\mathrm{cr}}=\pi-\arctan\left(K_{0}\right)-2\arctan\left(\frac{1}{K_{0}}\right).\label{eq:sSbf3a}
\end{equation}
Equation (\ref{eq:sSbf3a}) in turn readily implies the following
relationships illustrated by Fig. \ref{fig:mck-fcr-fop}:
\begin{equation}
\frac{d}{dK_{0}}\left(f_{\mathrm{max}}-f_{\mathrm{cr}}\right)=\frac{1}{1+K_{0}^{2}}>0;\quad\lim_{K_{0}\rightarrow+0}\left(f_{\mathrm{max}}-f_{\mathrm{cr}}\right)=0;\quad f_{\mathrm{cr}}<f_{\mathrm{max}}<\pi.\label{eq:sSbf3b}
\end{equation}
\begin{figure}[h]
\begin{centering}
\includegraphics[scale=0.6]{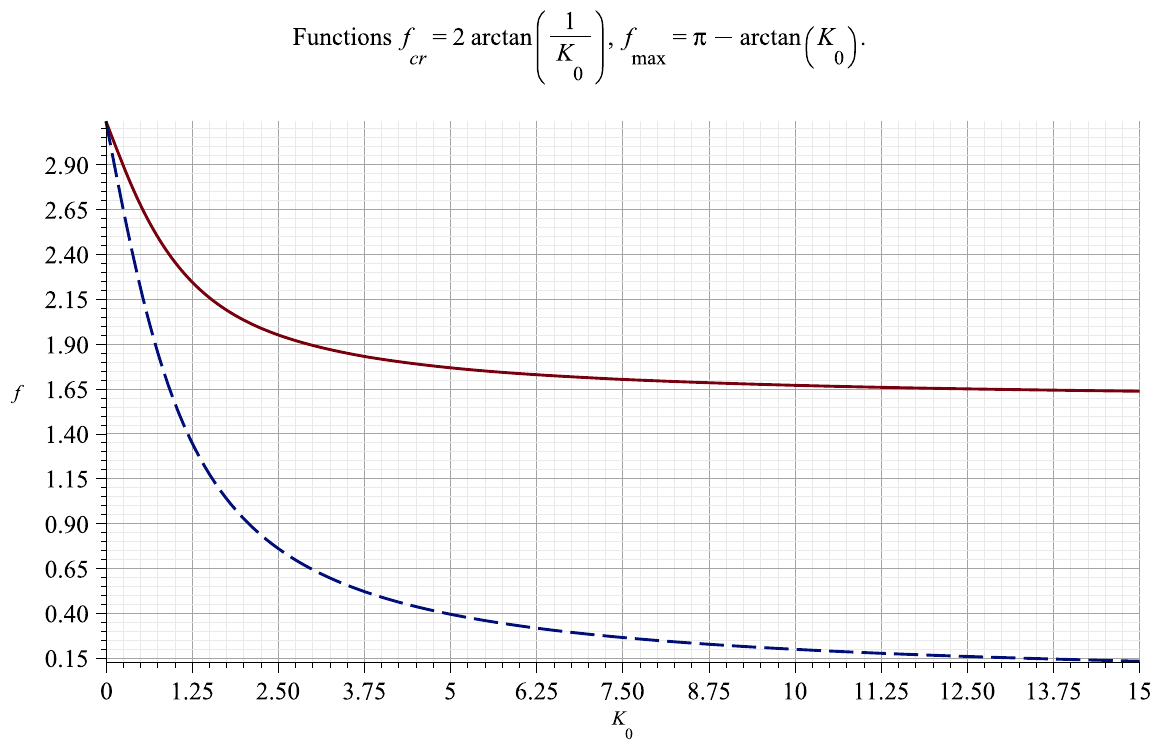}
\par\end{centering}
\centering{}\caption{\label{fig:mck-fcr-fop} Plots of $f_{\mathrm{max}}=\pi-\arctan\left(K_{0}\right)$
as solid (brown) curve and $f_{\mathrm{cr}}=2\arctan\left(\frac{1}{K_{0}}\right)$
as dashed (blue) curve. The horizontal and vertical axes represent
respectively $K_{0}$ and $f$. Asymptotic formulas (\ref{eq:bfomf2f}),
(\ref{eq:bfomf2f}) and (\ref{eq:sSbf2ca}), (\ref{eq:sSbf2ca}) describe
respectively the behavior of $f_{\mathrm{cr}}$ and $f_{\mathrm{max}}$
for small and large $K_{0}$.}
\end{figure}
\begin{figure}[h]
\begin{centering}
\includegraphics[scale=0.35]{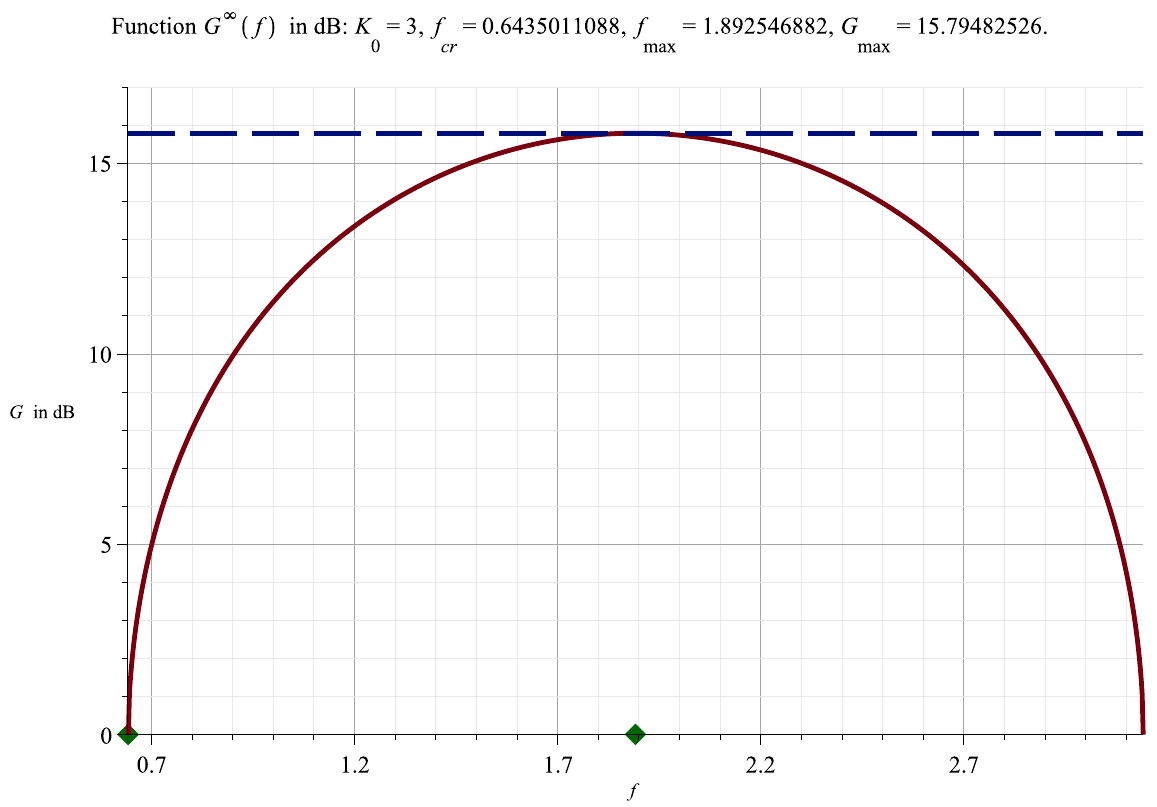}\hspace{1cm}\includegraphics[scale=0.36]{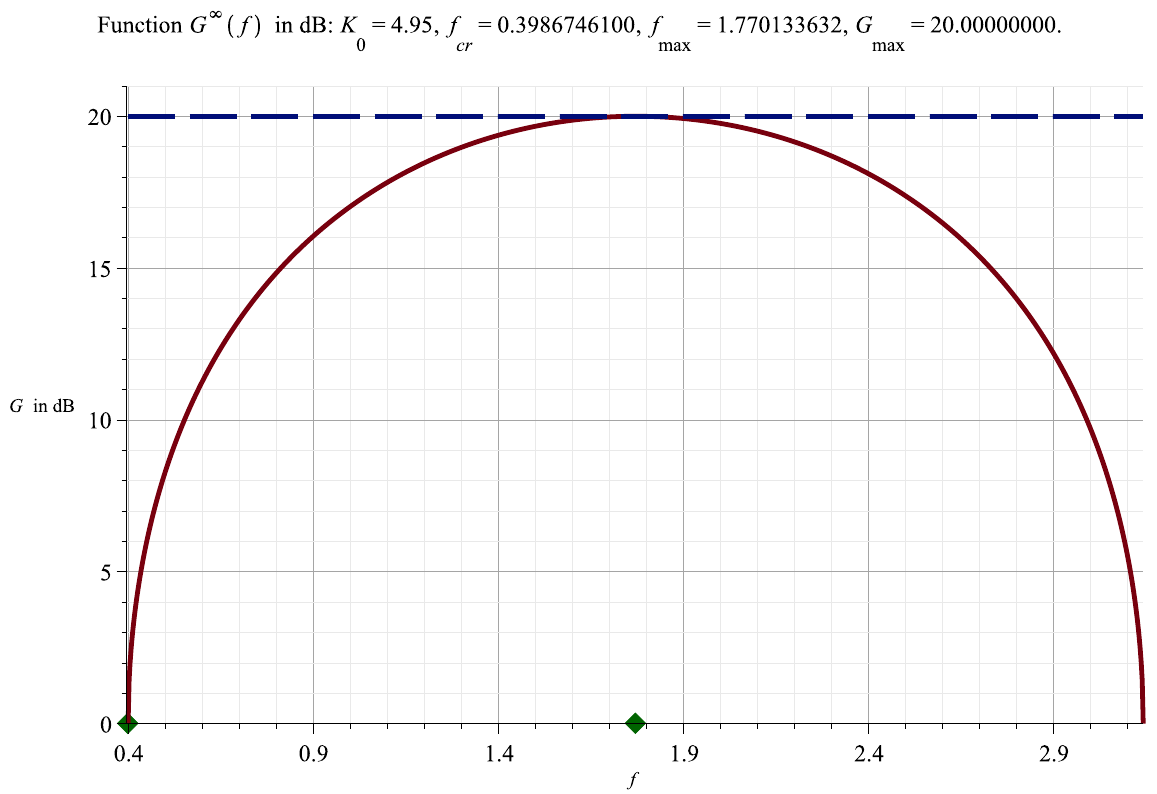}
\par\end{centering}
\centering{}\hspace{1cm}(a)\hspace{7.5cm}(b)\caption{\label{fig:mck-gain-inf} Plots of gain $G^{\infty}\left(f\right)$,
that is the gain in high frequency limit, as a function of $f$ defined
in relations (\ref{eq:sSbf1d}) for $f_{\mathrm{cr}}<f<\pi$ and (a)
$K_{0}=3$, with $f_{\mathrm{cr}}\protect\cong0.06435011088$, $f_{\mathrm{max}}\protect\cong1.892546882$
and $G_{\mathrm{max}}\protect\cong15.79482525$; (b) $K_{0}=4.95$,
with $f_{\mathrm{cr}}\protect\cong0.3986746100$, $f_{\mathrm{max}}\protect\cong1.770133632$
and $G_{\mathrm{max}}=20$. In all plots the horizontal and vertical
axes represent respectively frequency $f$ and gain $G$ in $\mathrm{dB}$.
The solid (brown) curves represent gain $G$ as a function of $f$.
The diamond solid (green) dots mark the values of $f_{\mathrm{cr}}$
and $f_{\mathrm{max}}$.}
\end{figure}

\section{Typical values of the MCK gain and its significant parameters\label{sec:typ}}

It would instructive to make an assessment of typical values the MCK
gain and its other parameters. There is the following rough empirical
formula that shows the dependence of the maximum power gain $G_{\mathrm{T}}\left(N\right)$
on the number $N$ of cavities in the klystron, \cite[7.7.1]{Tsim},
\cite[7.2.6]{Grigo}, \cite[16]{ValMid}:
\begin{equation}
G_{\mathrm{T}}\left(N\right)=15+20\left(N-2\right)\,\mathrm{dB}.\label{eq:GTmaxdB1a}
\end{equation}
Realistically achievable maximum amplification values though are smaller
and are of the order of $50\,\mathrm{dB}$ to $70\,\mathrm{dB}$.
The main limiting factors are noise and self-excitation of the klystron
because of parasitic feedback between cavities.

As to other universal values in the klystron theory the MCK theory
features a fundamental scale $z_{\mathrm{B}}$, called \emph{bunching
distance}, associated with one quarter a plasma oscillation cycle,
\cite[9.3.4]{BenSweScha}, \cite[9.2]{Gilm1}:
\begin{equation}
z_{\mathrm{B}}=\frac{\lambda_{\mathrm{rp}}}{4},\quad\lambda_{\mathrm{rp}}=\frac{2\pi\mathring{v}}{\omega_{\mathrm{rp}}}.\label{eq:GTmaxdB1b}
\end{equation}
The physical origin of the bunching distance scale $z_{\mathrm{B}}$
defined by equations (\ref{eq:GTmaxdB1b}) is evidently related to
the reduced plasma frequency $\lambda_{\mathrm{rp}}$ and it can be
explained as follows, \cite[9.2]{Gilm1} ($\lambda_{\mathrm{p}}=\frac{2\pi\mathring{v}}{\omega_{\mathrm{p}}}$):
\begin{quotation}
``At the axial position denoted by $\lambda_{\mathrm{p}}/4$ fast
electrons have been slowed to the dc velocity of the beam and slow
electrons have been accelerated to the dc beam velocity. At $\lambda_{\mathrm{p}}/4$,
all electrons have the same velocity. Also, at $\lambda_{\mathrm{p}}/4$
, the RF electron density and the RF current reach maximum values.
For small to medium RF signals, the RF current is nearly sinusoidal.
A very important characteristic of the bunching process with space
charge forces is that all electrons are either speeded up or slowed
down to the same velocity (the dc beam velocity) at the same axial
position ($\lambda_{\mathrm{p}}/4$). In addition, even if the amplitude
of the modulating field is changed so that initial electron velocities
are changed, the axial position of the bunch remains the same. This
result is extremely important to the klystron engineer because, unlike
the situation when space charge forces are ignored, the cavity location
for maximum RF beam current is not a function of signal level, of
gap width, or of frequency of operation.''
\end{quotation}
In other words the charge wave in the moving stream of electrons of
stationary (dc) velocity $\mathring{v}$ has the density that proportional
to a sinusoidal traveling wave, namely
\begin{equation}
\sin\left[\omega_{\mathrm{rp}}\left(\frac{z}{\mathring{v}}-t\right)\right]=\sin\left[k_{\mathrm{q}}z-\omega_{\mathrm{rp}}t\right],\quad k_{\mathrm{q}}=\frac{\omega_{\mathrm{rp}}}{\mathring{v}},\quad\lambda_{\mathrm{rp}}=\frac{2\pi}{k_{\mathrm{q}}}=\frac{2\pi\mathring{v}}{\omega_{\mathrm{rp}}},\label{eq:GTmaxdB1c}
\end{equation}
where $\lambda_{\mathrm{rp}}$ is the electron plasma wavelength as
in equations (\ref{eq:GTmaxdB1b}). In particular, equation (\ref{eq:GTmaxdB1c})
is consistent the formula (\ref{eq:GTmaxdB1b}) for bunching distance.

Comparing our formula (\ref{eq:sSbf2b}) for the maximal gain $G_{\mathrm{max}}$
for one period with Tsimring formula (\ref{eq:GTmaxdB1a}) and assuming
their consistency we readily arrive with the following equation for
a ``typical'' value $K_{0\mathrm{T}}$ for the instability coefficient
$K_{0}$:
\begin{equation}
G_{\mathrm{max}}=20\frac{\ln\left(K_{0}+\sqrt{K_{0}^{2}+1}\right)}{\ln\left(10\right)}=\lim_{N\rightarrow\infty}\frac{G_{\mathrm{T}}\left(N\right)}{N}=20.\label{eq:GTmaxdB1d}
\end{equation}
The solution $K_{0}=K_{0\mathrm{T}}$ to equation (\ref{eq:GTmaxdB1d})
is
\begin{equation}
K_{0\mathrm{T}}=4.95.\label{eq:GTmaxdB1e}
\end{equation}
yielding the corresponding values of $f_{\mathrm{cr}}$, $f_{\mathrm{max}}$
and $\quad G_{\mathrm{max}}:$
\begin{equation}
f_{\mathrm{crT}}\cong0.3986746100,\quad f_{\mathrm{maxT}}\cong1.770133632,\quad G_{\mathrm{maxT}}=20.\label{eq:GTmaxdB1f}
\end{equation}
Based on typical values of the MCK parameters as in equations (\ref{eq:GTmaxdB1e})
and (\ref{eq:GTmaxdB1f}) we generate the gain $G$ as a function
of frequency $\omega$ defined by equations (\ref{eq:sSbf1c}) shown
in Figs. \ref{fig:mck-gainT1} and \ref{fig:mck-gainT2}.
\begin{figure}[h]
\begin{centering}
\includegraphics[scale=0.35]{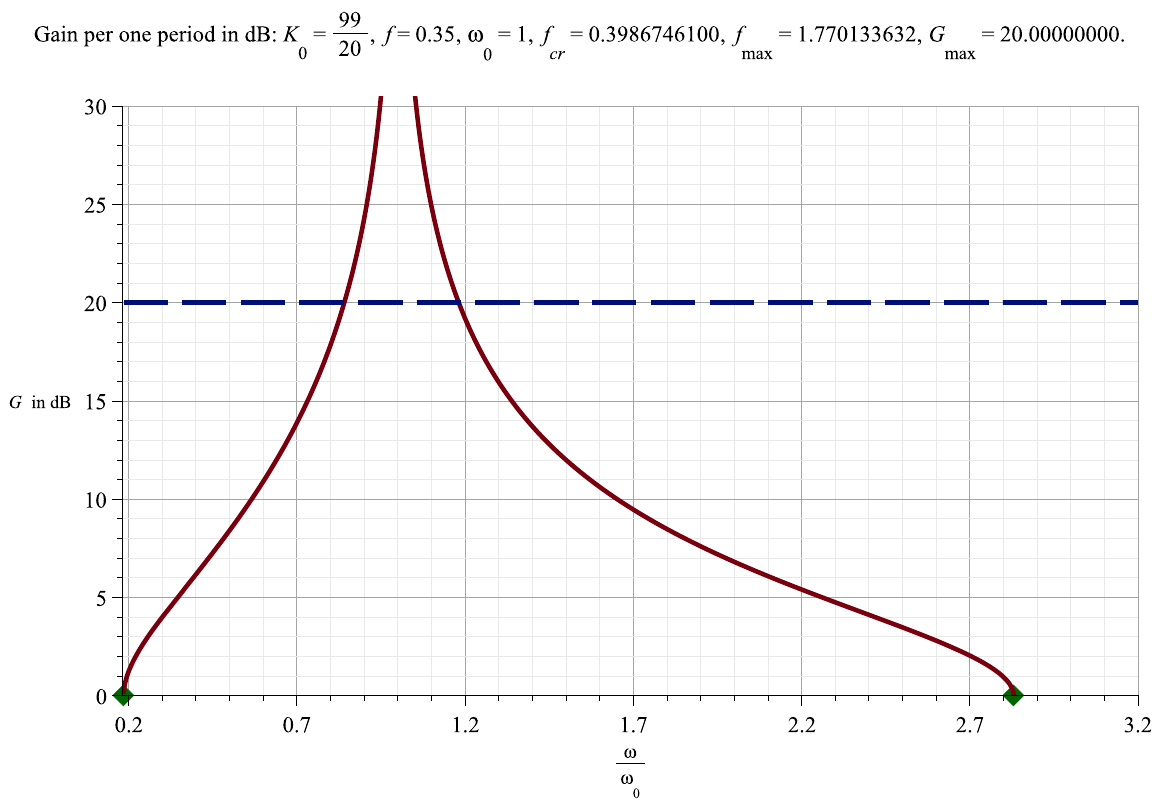}\hspace{1cm}\includegraphics[scale=0.36]{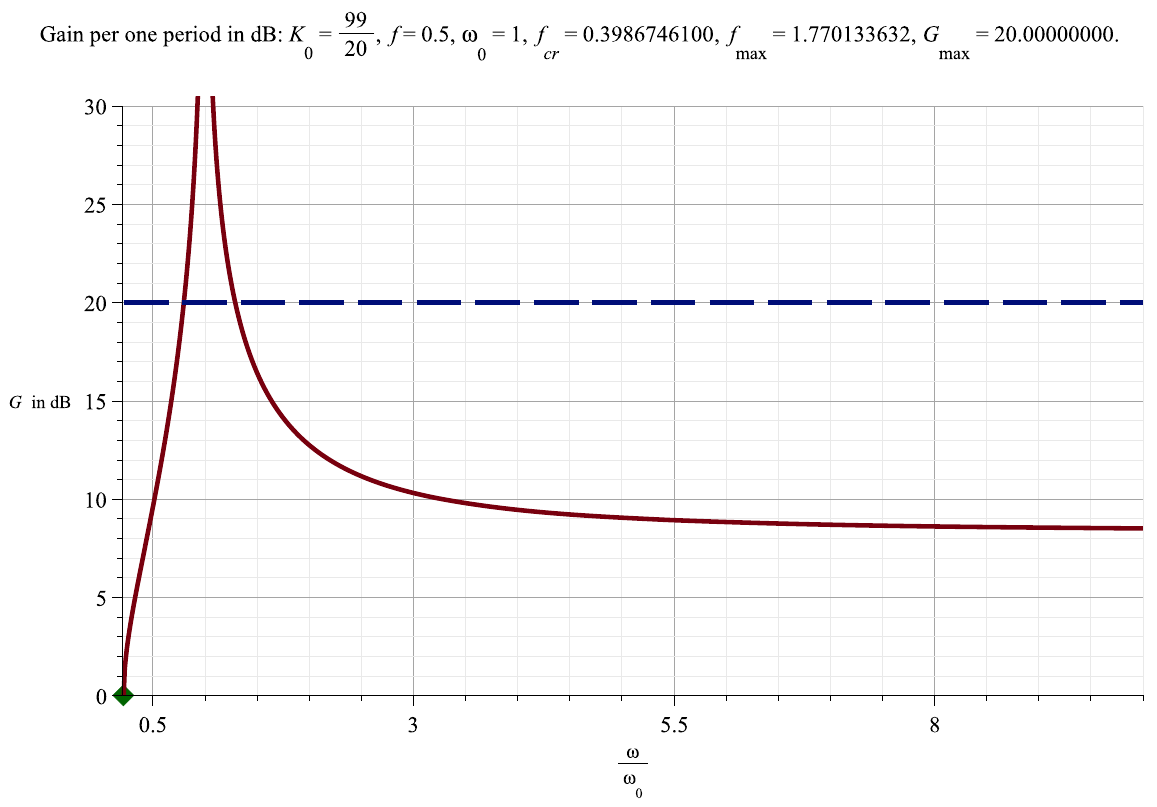}
\par\end{centering}
\centering{}\hspace{1cm}(a)\hspace{7.5cm}(b)\caption{\label{fig:mck-gainT1} Plots of gain $G$ as a function of frequency
$\omega$ defined by equations (\ref{eq:sSbf1c}) for $\omega_{0}=1$,
$K_{0}=K_{0\mathrm{T}}=4.95$ and consequently $f_{\mathrm{cr}}\protect\cong0.3986746100$,
$f_{\mathrm{max}}\protect\cong1.770133632$, $G_{\mathrm{max}}=20$
and: (a) $f=0.35<f_{\mathrm{cr}}$; (b) $f=0.5>f_{\mathrm{cr}}$.
In all plots the horizontal and vertical axes represent respectively
frequency $\omega$ and gain $G$ in $\mathrm{dB}$. The solid (brown)
curves represent gain $G$ as a function of frequency $\omega$, the
dashed (blue) line $G=G_{\mathrm{max}}$ represent the maximal value
$G_{\mathrm{max}}$ of $G$ in the high frequency limit (see Section
\ref{subsec:gain-hf}). The diamond solid (green) dots mark the values
of $\varOmega_{f}^{-}$ and $\varOmega_{f}^{+}$ which are the frequency
boundaries of the instability.}
\end{figure}
\begin{figure}[h]
\begin{centering}
\includegraphics[scale=0.5]{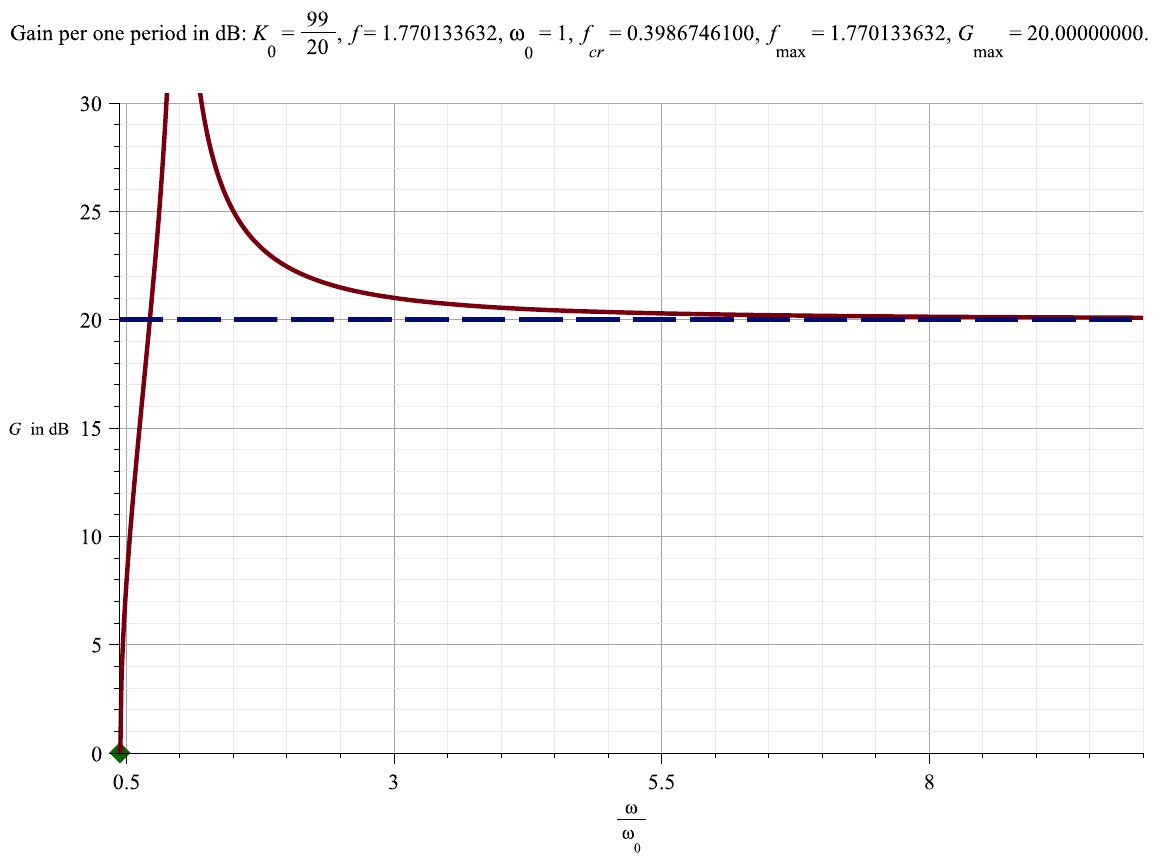}
\par\end{centering}
\centering{}\caption{\label{fig:mck-gainT2} Plot of gain $G$ as a function of frequency
$\omega$ defined by equations (\ref{eq:sSbf1c}) for $\omega_{0}=1$,
$K_{0}=K_{0\mathrm{T}}=4.95$ and consequently $f_{\mathrm{cr}}\protect\cong0.3986746100$,
$f_{\mathrm{max}}\protect\cong1.770133632$, $G_{\mathrm{max}}=20$
and $f=f_{\mathrm{max}}\protect\cong1.7701336320$. The horizontal
and vertical axes represent respectively frequency $\omega$ and gain
$G$ in $\mathrm{dB}$. The solid (brown) curve represents gain $G$
as a function of frequency $\omega$, the dashed (blue) line $G=G_{\mathrm{max}}$
represent the maximal value $G_{\mathrm{max}}$ of $G$ in the high
frequency limit (see Section \ref{subsec:gain-hf}). The diamond solid
(green) dot mark the value of $\varOmega_{f}^{-}$ which is the frequency
boundaries of the instability.}
\end{figure}

Note that according to equation (\ref{eq:Lagdim1ek}) the normalized
period $f=\frac{2\pi a}{\lambda_{\mathrm{rp}}}$. Combining that with
equations (\ref{eq:bfomf1d}), (\ref{eq:sSbf2a}) and (\ref{eq:sSbf2e})
for $f_{\mathrm{cr}}$ and $f_{\mathrm{max}}$ we obtain the following
expression for the corresponding values $a_{\mathrm{cr}}$ and $a_{\mathrm{max}}$
of the MCK period $a$:
\begin{gather}
a_{\mathrm{cr}}=\frac{f_{\mathrm{cr}}}{2\pi}\lambda_{\mathrm{rp}}=\frac{\arctan\left(\frac{1}{K_{0}}\right)}{\pi}\lambda_{\mathrm{rp}},\quad a_{\mathrm{max}}=\frac{f_{\mathrm{max}}}{2\pi}\lambda_{\mathrm{rp}}=\frac{\pi-\arctan\left(K_{0}\right)}{2\pi}\lambda_{\mathrm{rp}}=\frac{a_{\mathrm{cr}}}{2}+\frac{\lambda_{\mathrm{rp}}}{4},\label{eq:GTmaxdB2a}\\
\text{ where }K_{0}=\frac{b^{2}g_{\mathrm{B}}}{c_{0}},\quad g_{\mathrm{B}}=\frac{\sigma_{\mathrm{B}}}{4\lambda_{\mathrm{rp}}},\quad\left|\arctan\left(*\right)\right|<\frac{\pi}{2}.\nonumber 
\end{gather}
The MCK period $a=a_{\mathrm{cr}}$ signifies the onset of MCK instability
for all frequencies $\omega>\omega_{0}$, that is for any $a_{\mathrm{cr}}<a<\frac{\lambda_{\mathrm{rp}}}{2}$
the MCK system is unstable for all frequencies $\omega>\omega_{0}$.
The MCK period $a=a_{\mathrm{max}}$ is the one at which the MCK system
attains its maximal gain for all frequencies $\omega>\omega_{0}$,
see Figs. \ref{fig:mck-gainT1} (b) and \ref{fig:mck-gainT2}.

Equations (\ref{eq:GTmaxdB2a}) the following limit relations (see
Fig. \ref{fig:mck-fcr-fop}):
\begin{equation}
\lim_{K_{0}\rightarrow\infty}a_{\mathrm{cr}}=0,\quad\lim_{K_{0}\rightarrow0}a_{\mathrm{cr}}=\frac{\lambda_{\mathrm{rp}}}{2},\label{eq:GTmaxdB2b}
\end{equation}
\begin{equation}
\lim_{K_{0}\rightarrow\infty}a_{\mathrm{max}}=\frac{\lambda_{\mathrm{rp}}}{4},\quad\lim_{K_{0}\rightarrow0}a_{\mathrm{max}}=\frac{\lambda_{\mathrm{rp}}}{2}.\label{eq:GTmaxdB2c}
\end{equation}
Note the first limit relation in equations (\ref{eq:GTmaxdB2c}) yields
the well known in the klystron theory bunching distance $z_{\mathrm{B}}=\frac{\lambda_{\mathrm{rp}}}{4}$
in equations (\ref{eq:GTmaxdB1b}) as the limit of $a_{\mathrm{max}}$
for large values of $K_{0}$. 

Expressions (\ref{eq:GTmaxdB2a}) for $a_{\mathrm{cr}}$ and $a_{\mathrm{max}}$
imply also the following inequalities:
\begin{equation}
a_{\mathrm{cr}}<a_{\mathrm{max}};\quad0<a_{\mathrm{cr}}<\frac{\lambda_{\mathrm{rp}}}{2},\quad\frac{\lambda_{\mathrm{rp}}}{4}<a_{\mathrm{max}}=\frac{a_{\mathrm{cr}}}{2}+\frac{\lambda_{\mathrm{rp}}}{4}<\frac{\lambda_{\mathrm{rp}}}{2}.\label{eq:GTmaxdB2d}
\end{equation}

\section{Dispersion relations\label{sec:disp}}

We start with an observation that in view of the relation $s=\exp\left\{ \mathrm{i}k\right\} $
between the Floquet multiplier $s$ and the wave number $k$ (see
Section \ref{sec:floquet} and Remark \ref{rem:disprel}) \emph{the
characteristic equation (\ref{eq:bfKpsi1ba}) can be viewed as an
expression of the dispersion relations between the frequency $\omega$
and the wavenumber $k$ and we will refer to it as the MCK dispersion
relations or just the dispersion relations. }Dispersion relation (\ref{eq:bfKpsi1ba})
can be readily recast as
\begin{equation}
S+2b_{f}+S^{-1}=0,\quad S=s\exp\left\{ -\mathrm{i}{\it \omega}\right\} =\exp\left\{ \mathrm{i}\left(k-{\it \omega}\right)\right\} ,\quad s=\exp\left\{ \mathrm{i}k\right\} ,\label{eq:dispSbf1a}
\end{equation}
or, equivalently, as
\begin{gather}
\cos\left(k-{\it \omega}\right)+b_{f}\left(\omega\right)=0,\quad b_{f}\left(\omega\right)=K\left(\omega\right)\sin\left(f\right)-\cos\left(f\right),\label{eq:dispSbf1b}\\
K\left(\omega\right)=K_{0}\frac{\omega^{2}}{\omega^{2}-\omega_{0}^{2}},\quad K_{0}=\frac{b^{2}\beta_{0}}{2f}=\frac{b^{2}g_{\mathrm{B}}}{c_{0}}.\nonumber 
\end{gather}
Equation (\ref{eq:dispSbf1b}) in turn is equivalent to
\begin{gather}
k_{\pm}\left(\omega\right)=\omega\pm\arccos\left(-b_{f}\left(\omega\right)\right)+2\pi m,\quad b_{f}\left(\omega\right)=K\left(\omega\right)\sin\left(f\right)-\cos\left(f\right),\quad m\in\mathbb{Z},\label{eq:dispSbf1c}\\
K\left(\omega\right)=K_{0}\frac{\omega^{2}}{\omega^{2}-\omega_{0}^{2}},\quad K_{0}=\frac{b^{2}\beta_{0}}{2f}=\frac{b^{2}g_{\mathrm{B}}}{c_{0}}.\nonumber 
\end{gather}

When constructing the MCK dispersion relations we follow to the general
approach reviewed in Section \ref{sec:floquet} (see Remark \ref{rem:disprel})
for finding the dispersion relations for periodic systems. Using Theorem
\ref{thm:floqmultbf} we obtain the following statement relating the
frequency $\omega$ to the wavenumber $k=k_{\pm}\left(\omega\right)$.
\begin{thm}[MCK dispersion relations]
\label{thm:mckdis} Let $s_{\pm}$ be the MCK Floquet multipliers
and let $k_{\pm}\left(\omega\right)$ be the corresponding complex-valued
wave numbers satisfying
\begin{equation}
s_{\pm}=s_{\pm}\left(\omega\right)=\exp\left\{ \mathrm{i}k_{\pm}\left(\omega\right)\right\} ,\label{eq:spmpo1d}
\end{equation}
Then statements of Theorem \ref{thm:floqmultbf} imply the following
representation for $k_{\pm}\left(\omega\right)$:
\begin{equation}
k_{\pm}\left(\omega\right)=\left\{ \begin{array}{rcr}
-\frac{1+\mathrm{sign}\,\left\{ b_{f}\left(\omega\right)\right\} }{2}\pi+\omega+2\pi m\pm\mathrm{i}\ln\left[\left(\left|b_{f}\left(\omega\right)\right|+\sqrt{b_{f}^{2}\left(\omega\right)-1}\right)\right] & \text{if } & b_{f}^{2}>1\\
-\frac{1+\mathrm{sign}\,\left\{ b_{f}\left(\omega\right)\right\} }{2}\pi+\omega+2\pi m\pm\arccos\left(\left|b_{f}\left(\omega\right)\right|\right) & \text{if } & b_{f}^{2}\leq1
\end{array}\right.,\quad m\in\mathbb{Z},\label{eq:spmpo1e}
\end{equation}
where $0<f<\pi$ and 
\begin{equation}
b_{f}\left(\omega\right)=K_{0}\frac{\omega^{2}}{\omega^{2}-\omega_{0}^{2}}\sin\left(f\right)-\cos\left(f\right),\quad K_{0}=\frac{b^{2}\beta_{0}}{2}=\frac{b^{2}g_{\mathrm{B}}}{c_{0}},\quad g_{\mathrm{B}}=\frac{\sigma_{\mathrm{B}}}{4\lambda_{\mathrm{rp}}}.\label{eq:spmpol1ea}
\end{equation}
Requirement for $\Re\left\{ k_{\pm}\left(\omega\right)\right\} $
to be in the first (main) Brillouin zone $\left(-\pi,\pi\right]$
effectively selects the band number $m$ that depend on $\omega$
as follows. For any given $\omega>0$ and $0<f<\pi$ the band number
$m\in\mathbb{Z}$ is determined by the requirement to satisfy the
following inequalities:
\begin{equation}
\begin{array}{rcr}
-\pi<-\frac{1+\mathrm{sign}\,\left\{ b_{f}\left(\omega\right)\right\} }{2}\pi+\omega+2\pi m\leq\pi, & \text{if } & b_{f}^{2}\left(\omega\right)>1\\
-\pi<-\frac{1+\mathrm{sign}\,\left\{ b_{f}\left(\omega\right)\right\} }{2}\pi\pm\arccos\left(-b_{f}\left(\omega\right)\right)+\omega+2\pi m\leq\pi, & \text{if } & b_{f}^{2}\left(\omega\right)<1
\end{array}.\label{eq:spmpo1h}
\end{equation}
\emph{The equations (\ref{eq:spmpo1e}) for the complex-valued wave
numbers $k_{\pm}\left(\omega\right)$ represent the dispersion relations
of the MCK.} 
\end{thm}

\begin{rem}[real part of the wave number]
\label{rem:real-kpi} Note that according to expression\emph{ (\ref{eq:spmpo1e})
}in Theorem \ref{thm:mckdis}\emph{ }and relations (\ref{eq:Ompom1c})
in Theorem \ref{thm:Omegapm}\emph{ }(see also Figure \ref{fig:Om-plusmin})
we have\emph{
\begin{equation}
\Re\left\{ k_{\pm}\left(\omega\right)\right\} =\pi+\omega+2\pi m,\quad\omega_{0}<\omega<\varOmega_{f}^{+}.\label{eq:spmpo1j}
\end{equation}
}Figures \ref{fig:mck-disp1}, \ref{fig:mck-disp2} and \ref{fig:mck-disp3}
illustrate graphically equation\emph{ (\ref{eq:spmpo1j}) }by perfect
straight lines parallel to $\Re\left\{ k\right\} =\omega$ in the
shadowed area.
\end{rem}

There is yet another form of the dispersion relation (\ref{eq:dispSbf1b})
and (\ref{eq:dispSbf1c}) which is the high-frequency form:
\begin{gather}
D_{\mathrm{K}}\left(\omega,k\right)=D_{\mathrm{K}}^{\left(0\right)}\left(\omega,k\right)+\frac{K_{0}\omega_{0}^{2}}{\omega^{2}-\omega_{0}^{2}}=0,\label{eq:DKomk1aK}\\
D_{\mathrm{K}}^{\left(0\right)}\left(\omega,k\right)=\cos\left({\it \omega}-k\right)+b_{f}^{\infty},\quad b_{f}^{\infty}=K_{0}\sin\left(f\right)-\cos\left(f\right).\label{eq:DKomk1bK}
\end{gather}
This form readily yields the following high-frequency approximation
to the MCK dispersion relations
\begin{equation}
\cos\left({\it \omega}-k\right)+b_{f}^{\infty}=0,\quad b_{f}^{\infty}=K_{0}\sin\left(f\right)-\cos\left(f\right),\quad\left|b_{f}^{\infty}\right|\leq1,\label{eq:DKomk1cK}
\end{equation}
or, equivalently
\begin{equation}
\omega=k\pm\arccos\left(-b_{f}^{\infty}\right)+2\pi m,\quad b_{f}^{\infty}=K_{0}\sin\left(f\right)-\cos\left(f\right),\quad\left|b_{f}^{\infty}\right|\leq1,\quad m\in\mathbb{Z},\label{eq:DKom1dK}
\end{equation}
where inequality $\left|b_{f}^{\infty}\right|\leq1$ is necessary
and sufficient for the existence of real-valued $\omega$ and $k$
satisfying the dispersion relation.

\subsection{Plotting the dispersion relations}

The conventional dispersion relations are defined as the relations
between real-valued frequency $\omega$ and real-valued wavenumber
$k$ associated with the relevant eigenmodes. In the case of interest
$k$ can be complex-valued and to represent all system modes geometrically
we follow to \cite[7]{FigTWTbk}. First, we parametrize every mode
of the system uniquely by the pair $\left(k\left(\omega\right),\omega\right)$
where $\omega$ is its frequency and $k\left(\omega\right)$ is its
wavenumber. If $k\left(\omega\right)$ is degenerate, it is counted
a number of times according to its multiplicity. In view of the importance
to us of the mode instability, that is, when $\Im\left\{ k\left(\omega\right)\right\} \neq0$,
we partition all the system modes represented by pairs $\left(k\left(\omega\right),\omega\right)$
into two distinct classes \textendash{} oscillatory modes and unstable
ones \textendash{} based on whether the wavenumber $k\left(\omega\right)$
is real- or complex-valued with $\Im\left\{ k\left(\omega\right)\right\} \neq0$.
We refer to a mode (eigenmode) of the system as an \emph{oscillatory
mode} if its wavenumber $k\left(\omega\right)$ is real-valued. We
associate with such an oscillatory mode point $\left(k\left(\omega\right),\omega\right)$
in the $k\omega$-plane with $k$ being the horizontal axis and $\omega$
being the vertical one. Similarly, we refer to a mode (eigenmode)
of the system as a \emph{(convective) unstable mode} if its wavenumber
$k$ is complex-valued with a nonzero imaginary part, that is, $\Im\left\{ k\left(\omega\right)\right\} \neq0$.
We associate with such an unstable mode point $\left(\Re\left\{ k\left(\omega\right)\right\} ,\omega\right)$
in the $k\omega$-plane. Since we consider here only \emph{convective
unstable modes}, we refer to them shortly as \emph{unstable modes}.
Notice that every point $\left(\Re\left\{ k\left(\omega\right)\right\} ,\omega\right)$
is in fact associated with two complex conjugate system modes with
$\pm\Im\left\{ k\left(\omega\right)\right\} $.

Based on the above discussion, we represent the set of all oscillatory
and unstable modes of the system geometrically by the set of the corresponding
modal points $\left(k\left(\omega\right),\omega\right)$ and $\left(\Re\left\{ k\left(\omega\right)\right\} ,\omega\right)$
in the $k\omega$-plane. We name this set the \emph{dispersion-instability
graph}. To distinguish graphically points $\left(k\left(\omega\right),\omega\right)$
associated oscillatory modes when $k\left(\omega\right)$ is real-valued
from points $\left(\Re\left\{ k\left(\omega\right)\right\} ,\omega\right)$
associated unstable modes when $k\left(\omega\right)$ is complex-valued
with $\Im\left\{ k\left(\omega\right)\right\} \neq0$ we mark by a
shadow the region occupied by points with $\Im\left\{ k\left(\omega\right)\right\} \neq0$.
We remind once again that every point $\left(\omega,\Re\left\{ k\left(\omega\right)\right\} \right)$
with $\Im\left\{ k\left(\omega\right)\right\} \neq0$ represents exactly
two complex conjugate unstable modes associated with $\pm\Im\left\{ k\left(\omega\right)\right\} $. 

Figures \ref{fig:mck-disp1}, \ref{fig:mck-disp2} and \ref{fig:mck-disp3}
illustrate graphically the dispersion relations $k_{\pm}\left(\omega\right)$
described by equations (\ref{eq:spmpo1e}).
\begin{figure}[h]
\begin{centering}
\includegraphics[scale=0.18]{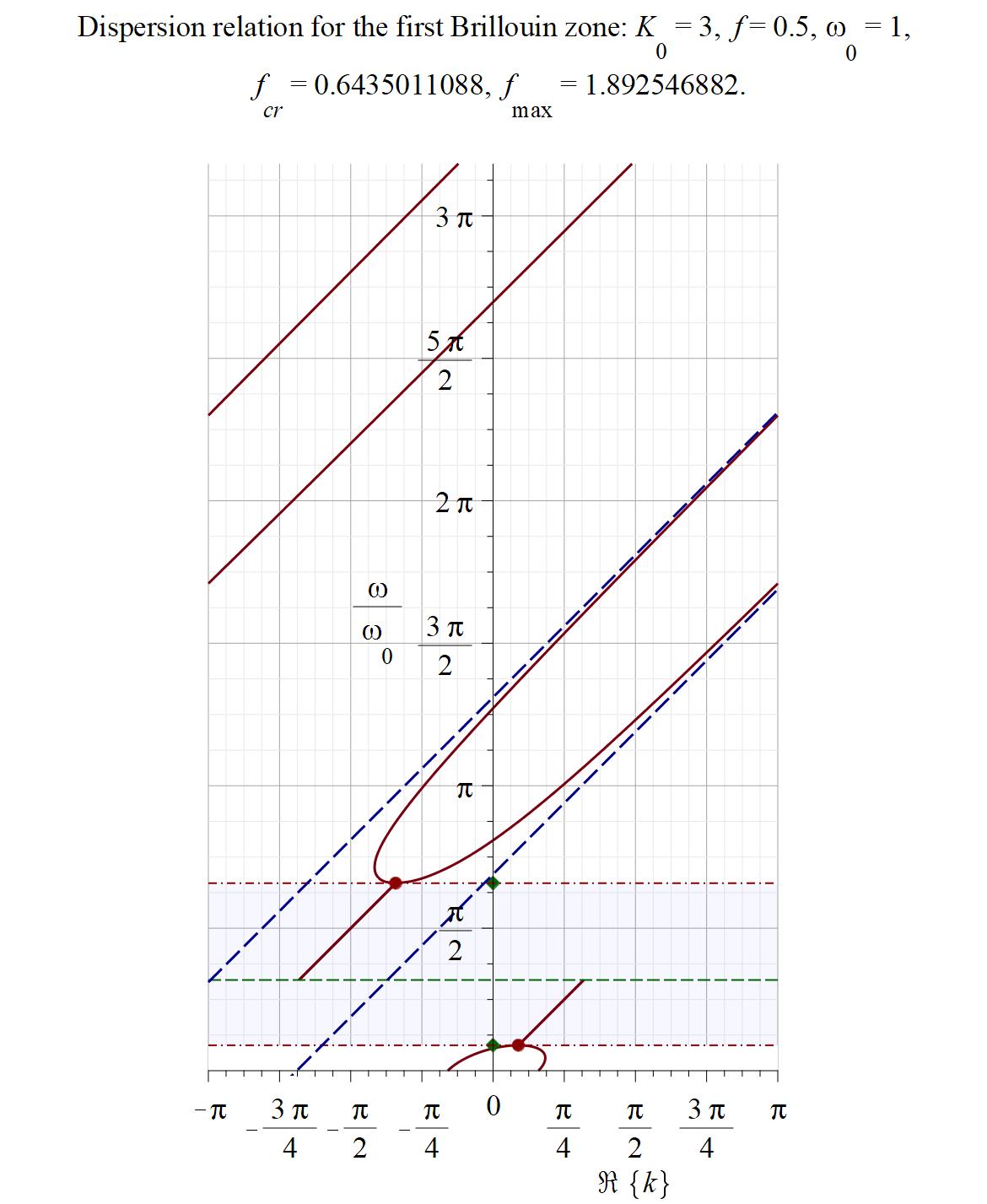}\hspace{0.1cm}\includegraphics[scale=0.18]{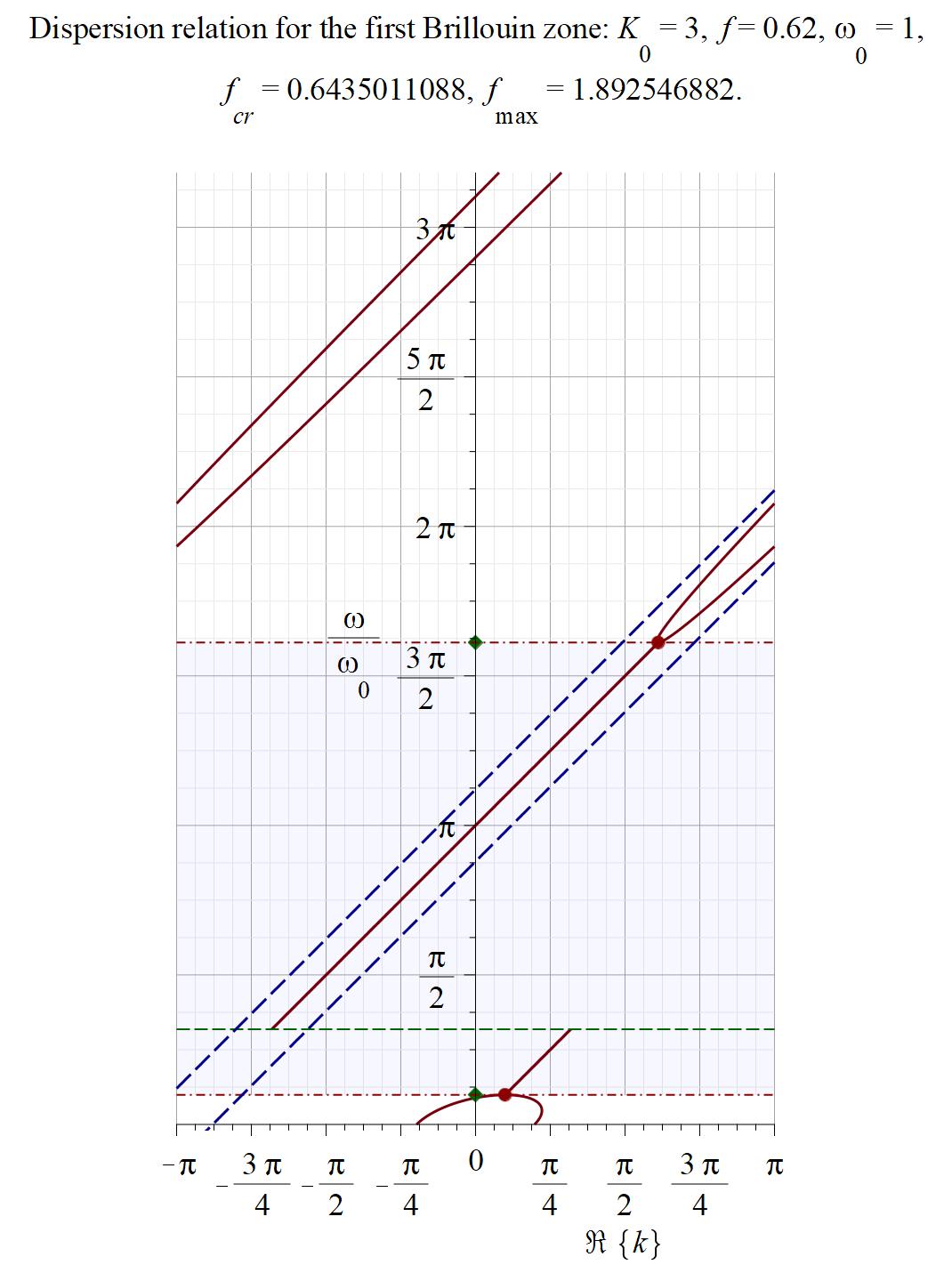}
\par\end{centering}
\centering{}(a)\hspace{6.5cm}(b)\caption{\label{fig:mck-disp1} The MCK dispersion-instability plots (solid
brown curves and lines) over main Brillouin zone $\left[-\pi,\pi\right]$
for $K_{0}=3$, $\omega_{0}=1$ for which $f_{\mathrm{cr}}\protect\cong0.6435011088$,
$f_{\mathrm{max}}\protect\cong1.892546882$: (a) $f=0.5<f_{\mathrm{cr}}\protect\cong0.6435011088$;
(b) $f=0.62<f_{\mathrm{cr}}\protect\cong0.6435011088$. In all plots
the horizontal and vertical axes represent respectively $\Re\left\{ k\right\} $
and $\frac{\omega}{\omega_{0}}$. Two solid (green) diamond dots identify
the values of $\varOmega_{f}^{-}$ and $\varOmega_{f}^{+}$ which
are the frequency boundaries of the instability. Two solid (brown)
disk dots identify points of the transition from the instability to
the stability which are also EPD points. Two (brown) dash-dot lines
$\omega=\varOmega_{f}^{\pm}$ identify the frequency boundaries of
the instability and the shaded (light blue) region between the lines
identify points $\left(\Re\left\{ k\right\} ,\omega\right)$ of instability.
Dashed (green) line $\omega=\omega_{0}$ identifies the resonance
frequency $\omega_{0}$. Note the plots have jump-discontinuity along
the dashed (green) line, namely $\Re\left\{ k_{\pm}\left(\omega\right)\right\} $
jumps by $\pi$ according to equations (\ref{eq:spmpo1e})\emph{ }as
the frequency $\omega$ passes through the resonance frequency $\omega_{0}$
and the sign of $b_{f}\left(\omega\right)$ changes. The shadowed
area marks points $\left(\Re\left\{ k\right\} ,\omega\right)$ associated
with the instability. The dashed (blue) straight lines lines correspond
to the high frequency approximation defined by equations (\ref{eq:DKom1dK}).}
\end{figure}
\begin{figure}[h]
\begin{centering}
\includegraphics[scale=0.2]{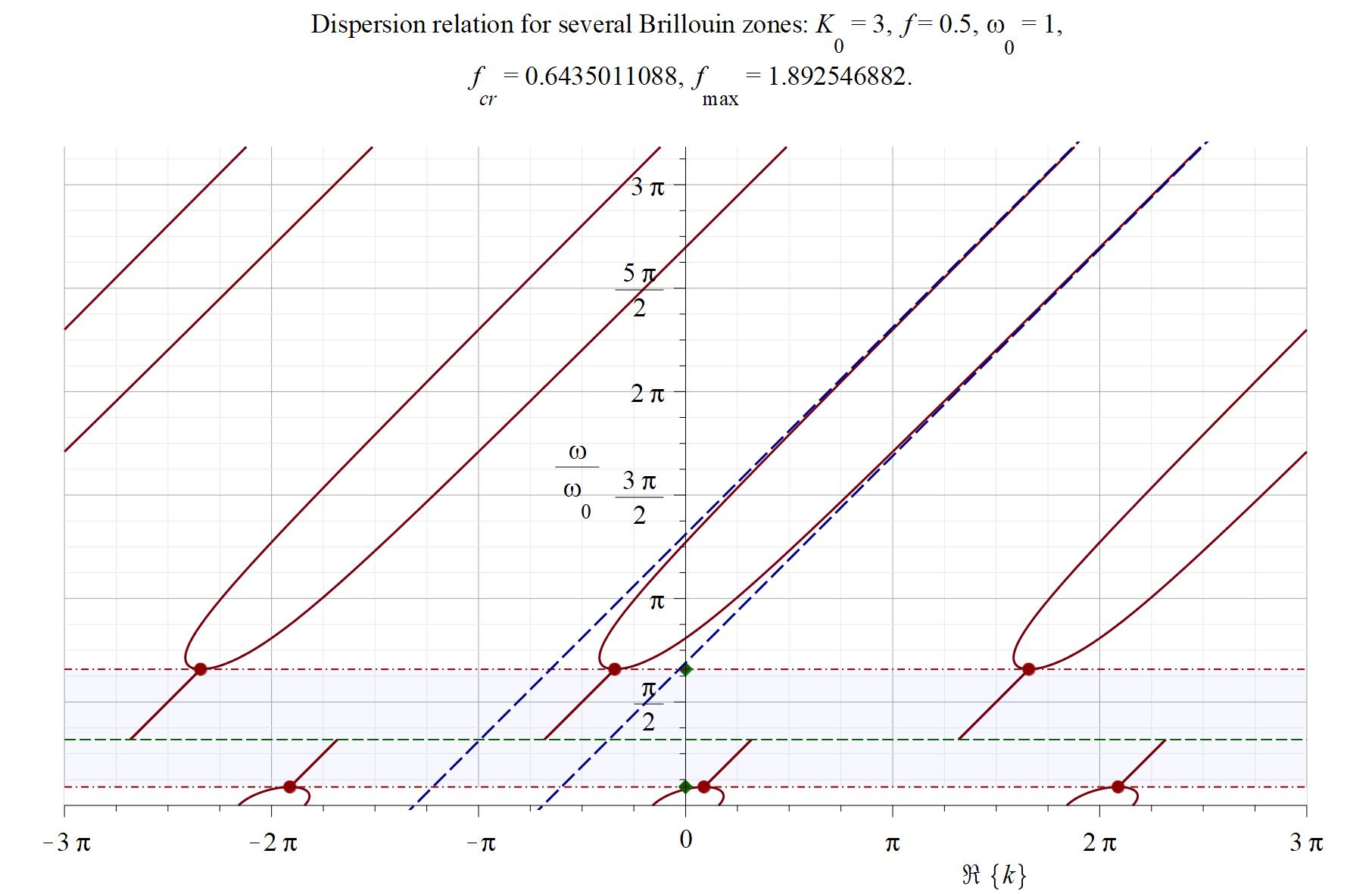}
\par\end{centering}
\centering{}\caption{\label{fig:mck-disp2} The MCK dispersion-instability plot (solid
brown curves and lines) over 3 Brillouin zones $3\left[-\pi,\pi\right]$
for $K_{0}=3$, $\omega_{0}=1$ for which $f_{\mathrm{cr}}\protect\cong0.6435011088$,
$f_{\mathrm{max}}\protect\cong1.892546882$ and $f=0.5<f_{\mathrm{cr}}\protect\cong0.6435011088$.
The horizontal and vertical axes represent respectively $\Re\left\{ k\right\} $
and $\frac{\omega}{\omega_{0}}$. Two solid (green) diamond dots identify
the values of $\varOmega_{f}^{-}$ and $\varOmega_{f}^{+}$ which
are the frequency boundaries of the instability. Solid (brown) disk
dots identify points of the transition from the instability to the
stability which are also EPD points. Two (brown) dash-dot lines $\omega=\varOmega_{f}^{\pm}$
identify the frequency boundaries of the instability and the shaded
(light blue) region between the lines identify points $\left(\Re\left\{ k\right\} ,\omega\right)$
of instability. Dashed (green) line $\omega=\omega_{0}$ identifies
the resonance frequency $\omega_{0}$. Note the plot has a jump-discontinuity
along the dashed (green) line, namely $\Re\left\{ k_{\pm}\left(\omega\right)\right\} $
jumps by $\pi$ according to equations (\ref{eq:spmpo1e})\emph{ }as
the frequency $\omega$ passes through the resonance frequency $\omega_{0}$
and the sign of $b_{f}\left(\omega\right)$ changes. The shadowed
area marks points $\left(\Re\left\{ k\right\} ,\omega\right)$ associated
with the instability. The dashed (blue) straight lines lines correspond
to the high frequency approximation defined by equations (\ref{eq:DKom1dK}).}
\end{figure}
\begin{figure}[h]
\begin{centering}
\includegraphics[scale=0.2]{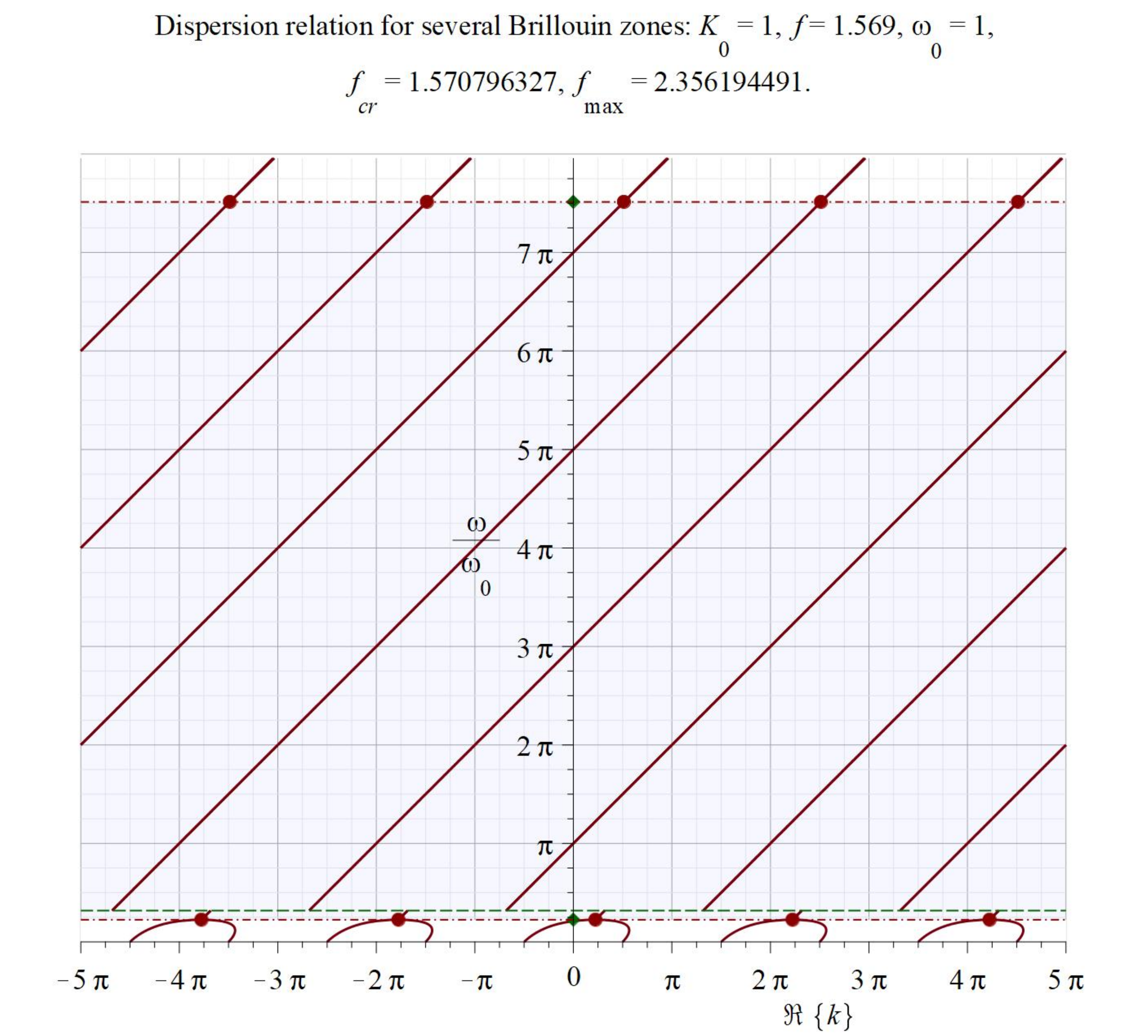}
\par\end{centering}
\centering{}\caption{\label{fig:mck-disp3} The MCK dispersion-instability plot (solid
brown curves and lines) over 3 Brillouin zones $3\left[-\pi,\pi\right]$
for $K_{0}=1$, $\omega_{0}=1$ for which $f_{\mathrm{cr}}\protect\cong1.570796327$,
$f_{\mathrm{max}}\protect\cong1.892546882$ and $f=1.569\protect\cong f_{\mathrm{cr}}\protect\cong1.570796327$.
The horizontal and vertical axes represent respectively $\Re\left\{ k\right\} $
and $\frac{\omega}{\omega_{0}}$. Two solid (green) diamond dots identify
the values of $\varOmega_{f}^{-}$ and $\varOmega_{f}^{+}$ which
are the frequency boundaries of the instability. Solid (brown) disk
dots identify points of the transition from the instability to the
stability which are also EPD points. Two (brown) dash-dot lines $\omega=\varOmega_{f}^{\pm}$
identify the frequency boundaries of the instability and the shaded
(light blue) region between the lines identify points $\left(\Re\left\{ k\right\} ,\omega\right)$
of instability. Dashed (green) line $\omega=\omega_{0}$ identifies
the resonance frequency $\omega_{0}$. Note the plot has a jump-discontinuity
along the dashed (green) line, namely $\Re\left\{ k_{\pm}\left(\omega\right)\right\} $
jumps by $\pi$ according to equations (\ref{eq:spmpo1e})\emph{ }as
the frequency $\omega$ passes through the resonance frequency $\omega_{0}$
and the sign of $b_{f}\left(\omega\right)$ changes. The shadowed
area marks points $\left(\Re\left\{ k\right\} ,\omega\right)$ associated
with the instability.}
\end{figure}

\section{Exceptional points of degeneracy\label{sec:epd}}

The concept of an exceptional point of degeneracy (EPD), \cite[II.1]{Kato},
refers to a system evolution matrix degeneracy when not only some
eigenvalues of the matrix coincide but the corresponding eigenvectors
coincide also. An important class of applications of EPDs is sensing,
\cite{CheN}. \cite{PeLiXu}, \cite{Wie}, \cite{Wie1}, \cite{KNAC},
\cite{OGC}. In our prior work in \cite{FigSynbJ} and \cite{FigPert}
we advanced and studied simple circuits exhibiting EPDs and their
applications to sensing. Our studies of traveling wave tubes (TWT)
in \cite[4, 7, 13, 14, 54, 55]{FigTWTbk} demonstrate that TWTs always
have EPDs. In \cite{FigtwtEPD} we developed applications of EPDs
to sensing based on TWTs. For more applications of EPDs to TWTs see
\cite{OTC}, \cite{OVFC}, \cite{OVFC1}, \cite{VOFC}.

In this section we study EPDs in the MCK system using the properties
of the MCK Floquet multipliers established in Sections \ref{sec:floqmul}
and \ref{sec:instfreq}. In particular, it follows from equation (\ref{eq:bfKpsi1bb})
that the degeneracy of the Floquet multipliers
\begin{equation}
s_{\pm}=e^{\mathrm{i}\omega}S_{\pm},\quad S_{\pm}=-b_{f}\pm\sqrt{b_{f}^{2}-1},\quad b_{f}=b_{f}\left(\omega\right),\label{eq:epdTpm1as}
\end{equation}
occurs if and only if $b_{f}\left(\omega\right)=\pm1$ where $b_{f}\left(\omega\right)$
is defined by equations (\ref{eq:spmpol1ea}). In this case the degenerate
Floquet multiplier is a single number for each of the values of $b_{f}$
which is
\begin{equation}
s=-e^{\mathrm{i}\omega}b_{f}\left(\omega\right)=\mp e^{\mathrm{i}\omega},\quad b_{f}\left(\omega\right)=K_{0}\frac{\omega^{2}}{\omega^{2}-\omega_{0}^{2}}\sin\left(f\right)-\cos\left(f\right)=\pm1.\label{eq:epdTpm1aa}
\end{equation}
The Floquet multiplier $s$ in equation (\ref{eq:epdTpm1aa}) is associated
with monodromy matrix $\mathscr{T}$ defined by equations (\ref{eq:abcha3dkL})
that for $b_{f}=\pm1$ takes the following form
\begin{equation}
\mathscr{T}_{\pm}=e^{\mathrm{i}\omega}\left[\begin{array}{rr}
\cos\left(f\right)-{\it \mathrm{i}\omega}\frac{\sin\left(f\right)}{f} & \frac{\sin\left(f\right)}{f}\\
\frac{\sin\left(f\right)}{f}\omega^{2}+2\mathrm{i}\omega\left({\it \cos\left(f\right)}\pm1\right)-\frac{f\left(1\pm\cos\left(f\right)\right)^{2}}{\sin\left(f\right)} & \mathrm{i}\omega\frac{\sin\left(f\right)}{f}-\cos\left(f\right)\mp2
\end{array}\right],\quad b_{f}=\pm1.\label{eq:epdTpm1a}
\end{equation}
Note that for the each value $b_{f}=1$ and $b_{f}=-1$ each of the
corresponding monodromy matrices $\mathscr{T}_{+}$ and $\mathscr{T}_{-}$
has a single and hence degenerate eigenvalue respectively $s=-e^{\mathrm{i}\omega}$
and $s=e^{\mathrm{i}\omega}$ and the index $\pm$ for matrix $\mathscr{T}$
corresponds to the sign of $b_{f}=\pm1$.

Using elementary identity $\sin^{2}\left(f\right)=1-\cos^{2}\left(f\right)$
we can recast representation (\ref{eq:epdTpm1a}) as
\begin{equation}
\mathscr{T}_{\pm}=e^{\mathrm{i}\omega}\left[\begin{array}{rr}
\cos\left(f\right)-{\it \mathrm{i}\omega}\frac{\sin\left(f\right)}{f} & \frac{\sin\left(f\right)}{f}\\
\frac{\sin\left(f\right)}{f}\omega^{2}+2\mathrm{i}\omega\left({\it \cos\left(f\right)}\pm1\right)+\frac{f\sin\left(f\right)\left(\cos\left(f\right)\pm1\right)}{{\it \cos\left(f\right)}\mp1} & \mathrm{i}\omega\frac{\sin\left(f\right)}{f}-\cos\left(f\right)\mp2
\end{array}\right],\quad b_{f}=\pm1.\label{eq:epdTpm1ab}
\end{equation}
The spectral analysis of the monodromy matrices $\mathscr{T}_{\pm}$
shows that their Jordan canonical forms $\mathscr{J}_{\pm}$ are
\begin{equation}
\mathscr{J}_{\pm}=\left[\begin{array}{rr}
\mp e^{\mathrm{i}\omega} & 1\\
0 & \mp e^{\mathrm{i}\omega}
\end{array}\right],\quad\mathscr{T}_{\pm}=\mathscr{Z}_{\pm}\mathscr{J}_{\pm}\mathscr{Z}_{\pm}^{-1},\label{eq:epdTpm1b}
\end{equation}
where matrices $\mathscr{Z}_{\pm}$ are defined by
\begin{equation}
\mathscr{Z}_{\pm}=\left[\begin{array}{rr}
e^{\mathrm{i}\omega}\left[\cos\left(f\right)\pm1-{\it \mathrm{i}\omega}\frac{\sin\left(f\right)}{f}\right] & 1\\
e^{\mathrm{i}\omega}\left[\frac{\sin\left(f\right)}{f}\omega^{2}+2\mathrm{i}\omega\left({\it \cos\left(f\right)}\pm1\right)+\frac{f\sin\left(f\right)\left(\cos\left(f\right)\pm1\right)}{{\it \cos\left(f\right)}\mp1}\right] & 0
\end{array}\right].\label{eq:epdTpm1c}
\end{equation}
Note that the first column $c_{\pm}$ of the corresponding matrix
$\mathscr{Z}_{\pm}$ is the single eigenvector of $\mathscr{T}_{\pm}$
and its second column $r_{\pm}$ is the relevant root vector of $\mathscr{T}_{\pm}$,
that is
\begin{equation}
c_{\pm}=\left[\begin{array}{r}
e^{\mathrm{i}\omega}\left[\cos\left(f\right)\pm1-{\it \mathrm{i}\omega}\frac{\sin\left(f\right)}{f}\right]\\
e^{\mathrm{i}\omega}\left[\frac{\sin\left(f\right)}{f}\omega^{2}+2\mathrm{i}\omega\left({\it \cos\left(f\right)}\pm1\right)+\frac{f\sin\left(f\right)\left(\cos\left(f\right)\pm1\right)}{{\it \cos\left(f\right)}\mp1}\right]
\end{array}\right],\quad r_{\pm}=\left[\begin{array}{r}
1\\
0
\end{array}\right],\label{eq:epdTpm1d}
\end{equation}
\begin{equation}
c_{\pm}=\left(\mathscr{T}_{\pm}\pm e^{\mathrm{i}\omega}\mathbb{I}\right)r_{\pm},\quad\left(\mathscr{T}_{\pm}\pm e^{\mathrm{i}\omega}\mathbb{I}\right)c_{\pm}=0.\label{eq:epdTpm1e}
\end{equation}
\begin{figure}[h]
\begin{centering}
\includegraphics[scale=0.2]{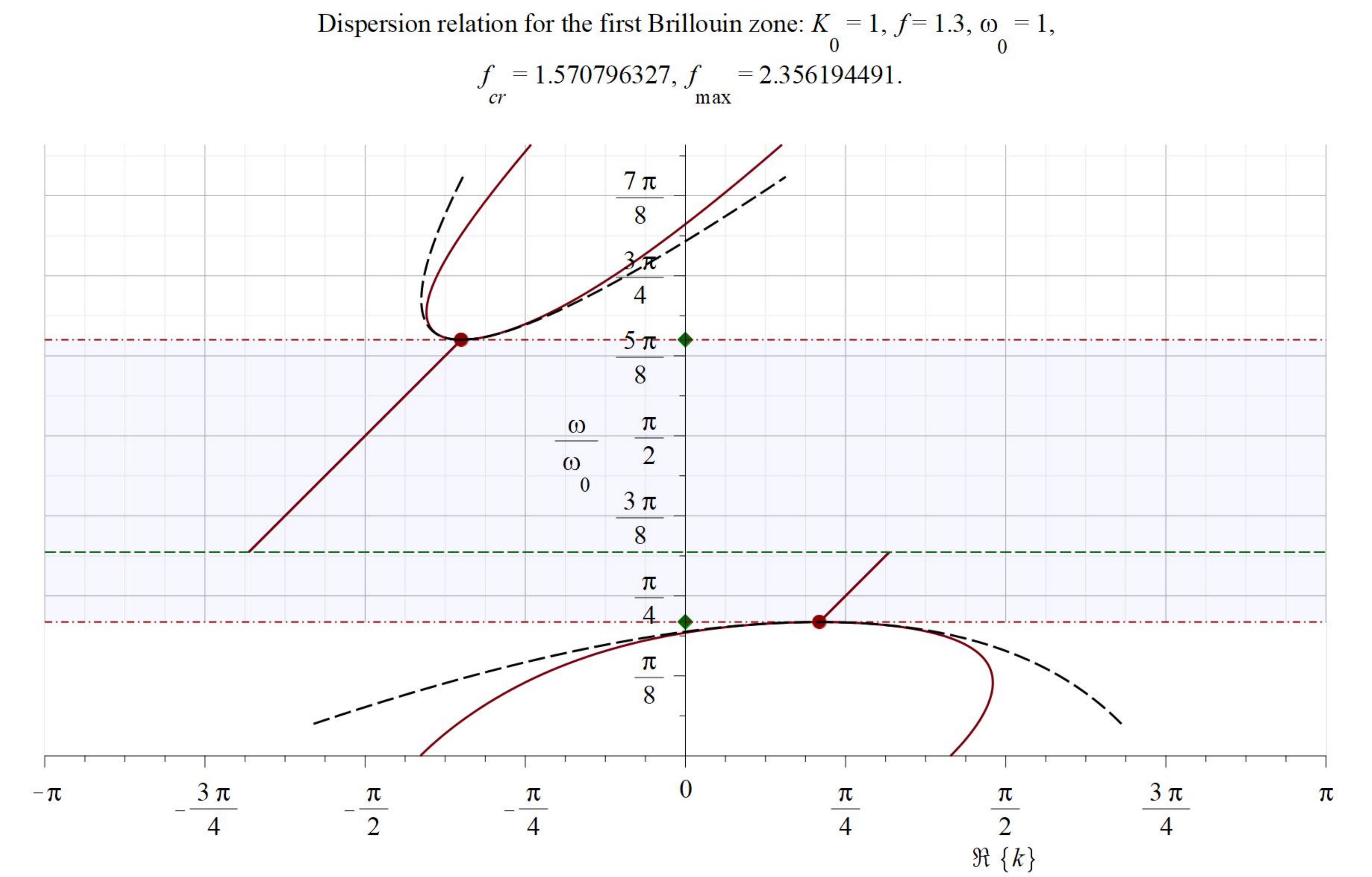}
\par\end{centering}
\centering{}\caption{\label{fig:mck-disp-epd} The dispersion-instability plot (solid brown
curves and lines) for the MCK with $K_{0}=1$, $\omega_{0}=1$ $f=1.3$
for which $f_{\mathrm{cr}}\protect\cong1.57096327$, $f_{\mathrm{max}}\protect\cong2.356194491$.
The horizontal and vertical axes represent respectively $\Re\left\{ k\right\} $
and $\frac{\omega}{\omega_{0}}$. Two solid (green) diamond dots identify
the values of $\varOmega_{f}^{-}$ and $\varOmega_{f}^{+}$ which
are the frequency boundaries of the instability. Solid (brown) disk
dots identify points of the transition from the instability to the
stability which are also EPD points. Two (brown) dash-dot lines $\omega=\varOmega_{f}^{\pm}$
identify the frequency boundaries of the instability and the shaded
(light blue) region between the lines identify points $\left(\Re\left\{ k\right\} ,\omega\right)$
of instability. Dashed (green) line $\omega=\omega_{0}$ identifies
the resonance frequency $\omega_{0}$. The two dashed (black) curves
represent the approximations to the dispersion relations described
by equations (\ref{eq:epdOm3a}) and (\ref{eq:epdOm3c}). Note the
plot has a jump-discontinuity along the dashed (green) line, namely
$\Re\left\{ k_{\pm}\left(\omega\right)\right\} $ jumps by $\pi$
according to equations (\ref{eq:spmpo1e})\emph{ }as the frequency
$\omega$ passes through the resonance frequency $\omega_{0}$ and
the sign of $b_{f}\left(\omega\right)$ changes.}
\end{figure}

According to the above analysis all EPD points of the MCK can be found
by solving equations $b_{f}\left(\omega\right)=\pm1$ for $\omega$.
Then based on Theorem \ref{thm:floqmultbf}, particularly equation
(\ref{eq:bfspm1d})), and Theorem \ref{thm:Omegapm}, particularly
relations (\ref{eq:Ompom1a})-(\ref{eq:Ompom1e}), we obtain the following
statement on EPDs of the MCK.
\begin{thm}[EPD points, their frequencies and wavenumbers]
\label{thm:epdom} If $0<f<f_{\mathrm{cr}}$ then there are exactly
two EPD points with the corresponding frequencies $\varOmega_{f}^{\pm}$
satisfying
\begin{equation}
b_{f}\left(\varOmega_{f}^{\pm}\right)=\pm1;\quad\varOmega_{f}^{+}=\omega_{0}\sqrt{\frac{1}{1-K_{0}\tan\left(\frac{f}{2}\right)}}>\omega_{0},\quad\varOmega_{f}^{-}=\omega_{0}\sqrt{\frac{\tan\left(\frac{f}{2}\right)}{\tan\left(\frac{f}{2}\right)+K_{0}}}<\omega_{0}.\label{eq:epdOm1a}
\end{equation}
where $b_{f}\left(\omega\right)$ is defined by equations (\ref{eq:spmpol1ea}).
If $f_{\mathrm{cr}}\leq f<\pi$ then there is exactly one EPD points
with the corresponding frequency $\varOmega_{f}^{-}$ satisfying
\begin{equation}
b_{f}\left(\varOmega_{f}^{-}\right)=-1;\quad\varOmega_{f}^{-}=\omega_{0}\sqrt{\frac{\tan\left(\frac{f}{2}\right)}{\tan\left(\frac{f}{2}\right)+K_{0}}}<\omega_{0}.\label{eq:epdOm1b}
\end{equation}
The expressions of the corresponding wavenumbers $k\left(\omega\right)$
with $\omega=\varOmega_{f}^{\pm}$ are provided by relations (\ref{eq:spmpo1e})
and (\ref{eq:spmpol1ea}) in Theorem \ref{thm:mckdis}.

The monodromy matrix $\mathscr{T}$ and its Floquet multipliers $s$
at the EPD points satisfy equations (\ref{eq:epdTpm1aa}), (\ref{eq:epdTpm1a})
and (\ref{eq:epdTpm1ab}) with the corresponding Jordan form $\mathscr{J}$
of $\mathscr{T}$ satisfying equations (\ref{eq:epdTpm1b}) and (\ref{eq:epdTpm1c}).
Note that matrix $\mathscr{J}$ is the Jordan block of dimension $2$
as expected for EPD points.
\end{thm}

We would like to derive asymptotic formulas for wave numbers $k_{\pm}\left(\omega\right)$
defined by equations (\ref{eq:spmpo1e}) when frequency $\omega$
is in a vicinity of EPD frequencies $\varOmega_{f}^{\pm}$. In order
to do that we set $\omega=\varOmega_{f}^{\pm}+\delta$ assuming that
$\delta$ is small and introduce the power series expansion for $b_{f}\left(\varOmega_{f}^{\pm}+\delta\right)$,
that is
\begin{equation}
b_{f}\left(\varOmega_{f}^{\pm}+\delta\right)=\pm1+\sum_{n=1}^{\infty}b_{n}^{\pm}\delta^{n},\quad\delta=\omega-\varOmega_{f}^{\pm}\rightarrow0,\quad b_{n}^{\pm}=\left.\partial_{\omega}\left[b_{f}\left(\omega\right)\right]\right|_{\omega=\varOmega_{f}^{\pm}},\;n=1,2,\ldots.\label{eq:epdOm1c}
\end{equation}
Using expression (\ref{eq:spmpol1ea}) for $b_{f}\left(\omega\right)$
and expressions (\ref{eq:epdOm1a}) for $\varOmega_{f}^{\pm}$ we
obtain the following representations for $b_{1}^{\pm}$:
\begin{equation}
b_{1}^{+}=-\frac{2\omega_{0}^{2}\sin(f)}{K_{0}\tan^{2}\left(\frac{f}{2}\right)\left(\varOmega_{f}^{+}\right)^{3}}<0,\quad\varOmega_{f}^{+}=\omega_{0}\sqrt{\frac{1}{1-K_{0}\tan\left(\frac{f}{2}\right)}}>\omega_{0},\quad0<f<f_{\mathrm{cr}},\label{eq:epdOm1d}
\end{equation}
\begin{equation}
b_{1}^{-}=-\frac{2\omega_{0}^{2}\sin(f)\tan^{2}\left(\frac{f}{2}\right)}{K_{0}\left(\varOmega_{f}^{-}\right)^{3}}<0,\quad\varOmega_{f}^{-}=\omega_{0}\sqrt{\frac{\tan\left(\frac{f}{2}\right)}{\tan\left(\frac{f}{2}\right)+K_{0}}}<\omega_{0}<0,\quad0<f<\pi.\label{eq:epdOm1e}
\end{equation}
Then based on equation (\ref{eq:epdTpm1aa}) for $S$, that is $S=-b_{f}\pm\sqrt{b_{f}^{2}-1},$
and relations (\ref{eq:epdOm1c})-(\ref{eq:epdOm1e}) we obtain 
\begin{equation}
S\left(\varOmega_{f}^{\pm}+\delta\right)=\mp1+\sqrt{\pm2b_{1}^{\pm}\delta}+b_{1}^{\pm}\delta-\frac{\left[\left(b_{1}^{\pm}\right)^{2}\pm2b_{2}^{\pm}\right]}{2\sqrt{\pm2b_{1}^{\pm}}}\delta^{\frac{3}{2}}+O\left(\delta^{2}\right),\quad\delta=\omega-\varOmega_{f}^{\pm}\rightarrow0.\label{eq:epdOm2a}
\end{equation}
Using equations (\ref{eq:spmpo1e}) for $k_{\pm}\left(\omega\right)$
in the case of the primary Brillouin zone with $m=0$ we get
\begin{equation}
k_{\mathrm{sign}\,\left(\sigma\right)}\left(\varOmega_{f}^{\pm}\right)=\varOmega_{f}^{\pm}-\frac{1\mp1}{2}\pi,\quad\sigma=\pm1,\quad0<f<\pi,\quad\sigma=\pm1,\label{eq:epdOm2b}
\end{equation}
\begin{equation}
k_{\mathrm{sign}\,\left(\sigma\right)}\left(\varOmega_{f}^{\pm}+\delta\right)=\varOmega_{f}^{\pm}-\frac{1\pm1}{2}\pi\mp\mathrm{i}\sigma\sqrt{\pm2b_{1}^{\pm}\delta}+\delta+O\left(\left|\delta\right|^{\frac{3}{2}}\right),\quad\delta=\omega-\varOmega_{f}^{\pm}\rightarrow0.\label{eq:epdOm2c}
\end{equation}
As to the real and imaginary parts of $k_{\mathrm{sign}\,\left(\sigma\right)}\left(\varOmega_{f}^{\pm}+\delta\right)$
equations (\ref{eq:spmpo1e}) and (\ref{eq:epdOm2c}) imply
\begin{equation}
\Re\left\{ k_{\mathrm{sign}\,\left(\sigma\right)}\left(\varOmega_{f}^{+}+\delta\right)\right\} =\left\{ \begin{array}{rcr}
\varOmega_{f}^{+}-\pi+\sigma\sqrt{2\left|b_{1}^{+}\right|}\sqrt{\delta}+\delta+O\left(\left|\delta\right|^{\frac{3}{2}}\right) & \text{if } & \delta>0\\
\varOmega_{f}^{+}-\pi+\delta & \text{if } & \delta<0
\end{array}\right.,\label{eq:epdOm3a}
\end{equation}
\begin{equation}
\Im\left\{ k_{\mathrm{sign}\,\left(\sigma\right)}\left(\varOmega_{f}^{+}+\delta\right)\right\} =\left\{ \begin{array}{rcr}
0 & \text{if } & \delta>0\\
\sigma\sqrt{2\left|b_{1}^{+}\right|}\sqrt{-\delta}+O\left(\left|\delta\right|^{\frac{3}{2}}\right) & \text{if } & \delta<0
\end{array}\right.,\label{eq:epdOm3b}
\end{equation}
\begin{equation}
\Re\left\{ k_{\mathrm{sign}\,\left(\sigma\right)}\left(\varOmega_{f}^{-}+\delta\right)\right\} =\left\{ \begin{array}{rcr}
\varOmega_{f}^{-}+\delta & \text{if } & \delta>0\\
\varOmega_{f}^{-}+\sigma\sqrt{2\left|b_{1}^{-}\right|}\sqrt{-\delta}+\delta+O\left(\left|\delta\right|^{\frac{3}{2}}\right) & \text{if } & \delta<0
\end{array}\right.,\label{eq:epdOm3c}
\end{equation}
\begin{equation}
\Im\left\{ k_{\mathrm{sign}\,\left(\sigma\right)}\left(\varOmega_{f}^{-}+\delta\right)\right\} =\left\{ \begin{array}{rcr}
\sigma\sqrt{2\left|b_{1}^{-}\right|}\sqrt{\delta}+O\left(\left|\delta\right|^{\frac{3}{2}}\right) & \text{if } & \delta>0\\
0 & \text{if } & \delta<0
\end{array}\right.,\label{eq:epdOm3d}
\end{equation}
\begin{figure}[h]
\begin{centering}
\includegraphics[scale=0.5]{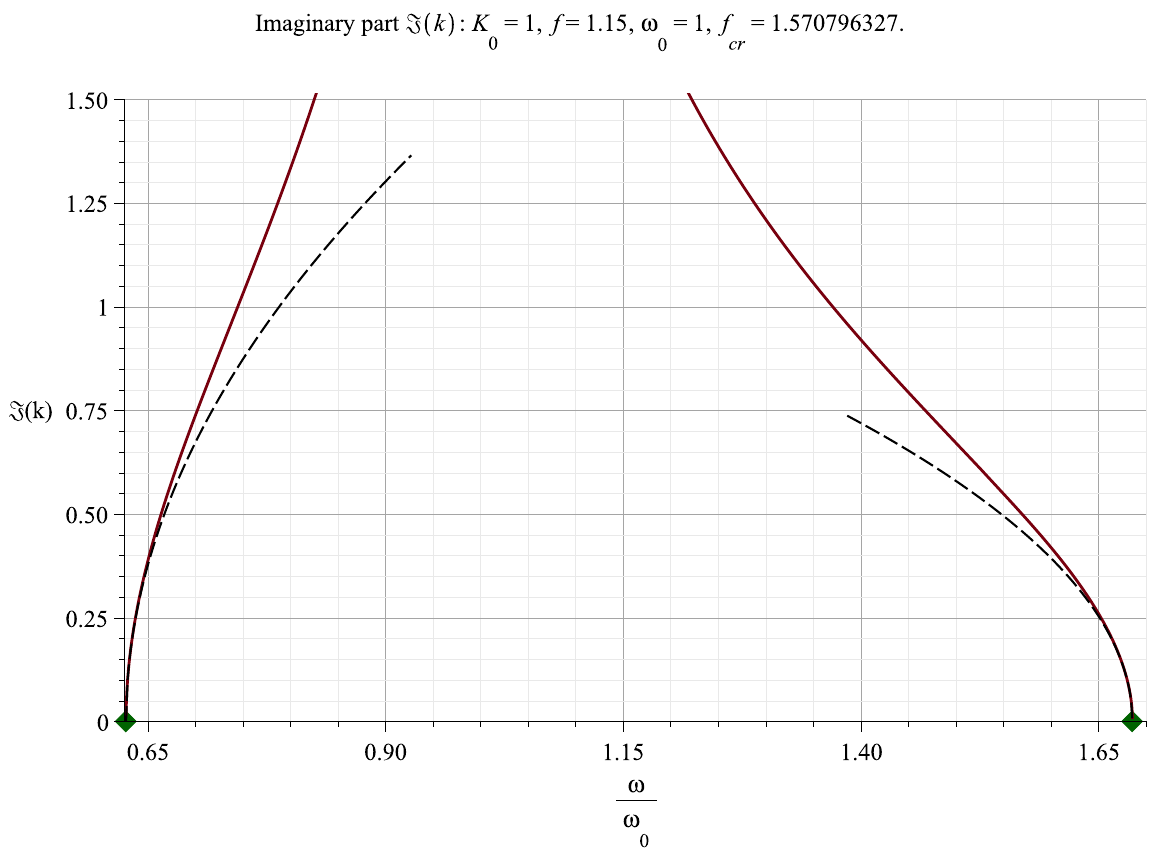}
\par\end{centering}
\centering{}\caption{\label{fig:mck-gain-om1-1} Plot of $\left|\Im\left\{ k_{\pm}\left(\omega\right)\right\} \right|$
as a function of frequency $\omega$ for $K_{0}=1$, $\omega_{0}=1$
and $f=1.15<f_{\mathrm{cr}}\protect\cong1.176005207$. The horizontal
and vertical axes represent respectively frequency $\omega$ and $\Im\left\{ k\right\} $.
The solid (brown) curves represent function $\left|\Im\left\{ k_{\pm}\left(\omega\right)\right\} \right|$.
The diamond solid (green) dots mark the values of $\varOmega_{f}^{-}$
and $\varOmega_{f}^{+}$ which are the frequencies of the MCK EPDs
and also are frequency boundaries of the instability. The two dashed
(black) curves represent the approximations of $\left|\Im\left\{ k_{\pm}\left(\omega\right)\right\} \right|$
described by equations (\ref{eq:epdOm3b}) and (\ref{eq:epdOm3d}). }
\end{figure}

\section{Lagrangian variational framework\label{sec:lagvar}}

We construct here the Lagrangian variational framework for our model
of the MCK. According to Assumption \ref{ass:mckmod} the model integrates
into it quantities associated with continuum of real numbers on one
hand and features associated with discrete points on the another hand.
The continuum features are represented by Lagrangian densities $\mathcal{L}_{\mathrm{B}}$
in equations (\ref{eq:aZep1ck}) whereas discrete features are represented
by Lagrangian $\mathcal{L}_{\mathrm{CB}}$ in equations (\ref{eq:aZep1dk})
with energies concentrated in a set of discrete points $a\mathbb{Z}$.
One possibility for constructing the desired Lagrangian variational
framework is to apply the general approach developed in \cite{FigRey2}
when the ``rigidity'' condition holds. Another possibility is to
directly construct the Lagrangian variational framework using some
ideas from \cite{FigRey2} and that is what we actually pursue here.

Following to the standard procedures of the Least Action principle
\cite[II.3]{ArnMech}, \cite[3]{GantM}, \cite[7]{GelFom}, \cite[8.6]{GoldM}
we start with setting up the action integral $S$ based on the Lagrangian
$\mathcal{L}$ defined by equations (\ref{eq:aZep1Lk}), (\ref{eq:aZep1ck})
and (\ref{eq:aZep1dk}). Using notations (\ref{eq:aZep1Q}) and (\ref{eq:eZep1xk})
we define the action integral $S$ as follows:
\begin{gather}
S\left(\left\{ x\right\} \right)=\int_{t_{0}}^{t_{1}}\mathrm{d}t\int_{z_{1}}^{z_{2}}\mathcal{L}\left(\left\{ x\right\} \right)\,\mathrm{d}z=S_{\mathrm{B}}\left(\left\{ q\right\} \right)+S_{\mathrm{CB}}\left(x\right),\quad t_{0}<t_{1},\quad z_{0}<z_{1},\label{eq:Sact1ak}
\end{gather}
where
\begin{equation}
S_{\mathrm{B}}\left(\left\{ q\right\} \right)=\int_{t_{0}}^{t_{1}}\mathrm{d}t\int_{z_{1}}^{z_{2}}\mathcal{L}_{\mathrm{B}}\left(\left\{ q\right\} \right)\,\mathrm{d}z=\int_{t_{0}}^{t_{1}}\mathrm{d}t\int_{z_{1}}^{z_{2}}\left[\frac{1}{2\beta}\left(\partial_{t}q+\mathring{v}\partial_{z}q\right)^{2}-\frac{2\pi}{\sigma_{\mathrm{B}}}q^{2}\right]\,\mathrm{d}z,\label{eq:Sact1ck}
\end{equation}
\begin{gather}
S_{\mathrm{CB}}\left(\left\{ Q\right\} \right)=\int_{t_{0}}^{t_{1}}\mathrm{d}t\int_{z_{1}}^{z_{2}}\mathcal{L}_{\mathrm{CB}}\left(Q,q\right)\,\mathrm{d}z=\label{eq:Sact1dk}\\
=\sum_{z_{1}<a\ell<z_{2}}\int_{t_{0}}^{t_{1}}\mathrm{d}t\left[\frac{l_{0}}{2}\left(\partial_{t}Q\left(a\ell\right)\right)^{2}-\frac{1}{2c_{0}}\left(Q\left(a\ell\right)+bq\left(a\ell\right)\right)\right]^{2}.\nonumber 
\end{gather}
To make expressions of the action integrals less cluttered we suppress
notationally their dependence on intervals $\left(z_{0},z_{1}\right)$
and $\left(t_{0},t_{1}\right)$ that can be chosen arbitrarily. We
consider then variation $\delta S$ of action $S$ assuming that variation
$\delta q$ of charge $q=q\left(z,t\right)$ vanishes outside intervals
$\left(z_{0},z_{1}\right)$ and $\left(t_{0},t_{1}\right)$, that
is 
\begin{equation}
\delta q\left(z,t\right)=0,\quad\left(z,t\right)\notin\left(z_{0},z_{1}\right)\times\left(t_{0},t_{1}\right),\label{eq:Sact1ek}
\end{equation}
implying, in particular, that $\delta q$ vanishes on the boundary
of the rectangle $\left(z_{0},z_{1}\right)\times\left(t_{0},t_{1}\right)$,
that is
\begin{equation}
\delta q\left(z,t\right)=0,\text{if }z=z_{0},z_{1}\text{ or if }t=t_{0},t_{1}.\label{eq:Sact1fk}
\end{equation}
We refer to variations $\delta Q$ and $\delta q$ satisfying equations
(\ref{eq:Sact1ek}) and hence (\ref{eq:Sact1ek}) for a rectangle
$\left(z_{0},z_{1}\right)\times\left(t_{0},t_{1}\right)$ as\emph{
admissible}.

Following to the least action principle we introduce the functional
differential $\delta S$ of the action by the following formula \cite[7(35)]{GelFom}
\begin{equation}
\delta S=\lim_{\varepsilon\rightarrow0}\frac{S\left(\left\{ x+\varepsilon\delta x\right\} \right)-S\left(\left\{ x\right\} \right)}{\varepsilon}.\label{eq:Sact2ak}
\end{equation}
Then the system configurations $x=x\left(z,t\right)$ that actually
can occur must satisfy
\begin{equation}
\delta S=\lim_{\varepsilon\rightarrow0}\frac{S\left(\left\{ x+\varepsilon\delta x\right\} \right)-S\left(\left\{ x\right\} \right)}{\varepsilon}=0\text{ for all admissible variations.}\label{eq:Sact2bk}
\end{equation}
Let us choose now any $z$ outside lattice $a\mathbb{Z}$. Then there
always exist a sufficiently small $\xi>0$ and an integer $\ell_{0}$
such that
\begin{equation}
a\ell_{0}<z_{0}=z-\xi<z<z_{1}=z+\xi<a\left(\ell_{0}+1\right).\label{eq:Sact2ck}
\end{equation}
If we apply now the variational principle (\ref{eq:Sact2bk}) for
all admissible variations $\delta Q$ and $\delta q$ such that space
interval $\left(z_{0},z_{1}\right)$ is compliant with inequalities
(\ref{eq:Sact2ck}) we readily find that
\begin{equation}
\delta S=\delta S_{\mathrm{B}}=0,\label{eq:Sact2dk}
\end{equation}
where $S_{\mathrm{B}}$ is defined by expression (\ref{eq:Sact1ck}).
Using equations (\ref{eq:Sact1fk}) and carrying out in the standard
way the integration by parts transformations we arrive at
\begin{gather}
\delta S_{\mathrm{B}}=-\int_{t_{0}}^{t_{1}}\mathrm{d}t\int_{z_{1}}^{z_{2}}\left[\frac{1}{\beta}\left(\partial_{t}+\mathring{v}\partial_{z}\right)^{2}q+\frac{4\pi}{\sigma_{\mathrm{B}}}q\right]\delta q\,\mathrm{d}z.\label{eq:Sact2fk}
\end{gather}
Combining equations (\ref{eq:Sact2dk}) and (\ref{eq:Sact2fk}) we
arrive at the following EL equations

\begin{gather}
\frac{1}{\beta}\left(\partial_{t}+\mathring{v}\partial_{z}\right)^{2}q+\frac{4\pi}{\sigma_{\mathrm{B}}}q=0,\quad z\neq a\ell,\quad\ell\in\mathbb{Z}.\label{eq:Sact2gk}
\end{gather}

Consider now the case when $z=a\ell_{0}$ for an integer $\ell_{0}$
and select space interval $\left(z_{0},z_{1}\right)$ as follows
\begin{equation}
a\left(\ell_{0}-1\right)<z_{0}=a\left(\ell_{0}-\frac{1}{2}\right)<z=a\ell_{0}<z_{1}=a\left(\ell_{0}+\frac{1}{2}\right)<a\left(\ell_{0}+1\right).\label{eq:Sact2hk}
\end{equation}
Notice that in this case both actions $S_{\mathrm{B}}$ and $S_{\mathrm{CB}}$
contribute to the variation $\delta S$. In particular, as consequence
of the presence of delta functions $\delta\left(z-a\ell\right)$ in
the expression of the Lagrangian $\mathcal{L}_{\mathrm{CB}}$ defined
by equation (\ref{eq:aZep1dk}) the space derivatives $\partial_{z}q$
can have jumps at $z=a\ell_{0}$ as it was already acknowledged by
Assumption \ref{ass:jumpcon}. Based on this circumstance we proceed
as follows: (i) we split the integral with respect to the space variable
$z$ into two integrals:

\begin{equation}
\int_{z_{0}}^{z_{1}}=\int_{a\left(\ell_{0}-\frac{1}{2}\right)}^{a\ell_{0}}+\int_{a\ell_{0}}^{a\left(\ell_{0}+\frac{1}{2}\right)};\label{eq:Sact3ak}
\end{equation}
(ii) we carry out the integration by parts for each of the two integrals
in the right-hand side of equation (\ref{eq:Sact3ak}); (iii) we use
already established EL equations (\ref{eq:Sact2gk}) to simplify the
integral expressions. When that is all done we arrive at the following:

\begin{equation}
\delta S_{\mathrm{B}}=-\int_{t_{0}}^{t_{1}}\frac{\mathring{v}^{2}}{\beta}\left[\partial_{z}q\right]\left(a\ell_{0},t\right)\delta q\left(a\ell_{0},t\right)\,\mathrm{d}t,\label{eq:Sact3bk}
\end{equation}
where jumps $\left[\partial_{z}q\right]\left(a\ell\right)$ are defined
by equation (\ref{eq:klimpm1b}), and
\begin{gather}
\delta S_{\mathrm{CB}}=-\int_{t_{0}}^{t_{1}}\left\{ l_{0}\partial_{t}^{2}Q\left(a\ell\right)+\frac{1}{c_{0}}\left[Q\left(a\ell_{0},t\right)+bq\left(a\ell_{0},t\right)\right]\right\} \delta Q\left(a\ell_{0},t\right)\,\mathrm{d}t-\label{eq:Sact3ck}\\
-\int_{t_{0}}^{t_{1}}\left\{ \frac{b}{c_{0}}\left[Q\left(a\ell_{0},t\right)+bq\left(a\ell_{0},t\right)\right]\right\} \delta q\left(a\ell_{0},t\right)\,\mathrm{d}t.\nonumber 
\end{gather}
Using the variational principle (\ref{eq:Sact2bk}), that is
\begin{equation}
\delta S_{\mathrm{B}}+\delta S_{\mathrm{CB}}=0,\label{eq:Sact3dk}
\end{equation}
and the fact that variations $\delta Q$ and $\delta q$ can be chosen
arbitrarily we arrive at the following
\begin{gather}
l_{0}\partial_{t}^{2}Q\left(a\ell\right)+\frac{1}{c_{0}}\left[Q\left(a\ell_{0},t\right)+bq\left(a\ell_{0},t\right)\right]=0,\label{eq:Sact3ek}\\
\frac{\mathring{v}^{2}}{\beta}\left[\partial_{z}q\right]\left(a\ell_{0},t\right)=-\frac{b}{c_{0}}\left[Q\left(a\ell_{0}\right)+bq\left(a\ell_{0},t\right)\right],\nonumber 
\end{gather}
where jumps $\left[\partial_{z}q\right]\left(a\ell\right)$ are defined
by equation (\ref{eq:klimpm1b}). We remind also that as consequence
of continuity of $q$ we also have 
\begin{equation}
\left[q\right]\left(a\ell_{0},t\right)=0.\label{eq:Sact3gk}
\end{equation}
Hence equations (\ref{eq:Sact3ek}) and (\ref{eq:Sact3gk}) can be
viewed as the EL equations at point $a\ell_{0}$. 

Equations (\ref{eq:Sact3ek}) at interaction point $a\ell_{0}$ are
perfectly consistent with boundary conditions (2.12) of the general
treatment in \cite{FigRey2}, which are
\begin{align}
-\frac{\partial L_{\mathrm{D}}}{\partial\partial_{1}\psi_{\mathrm{D}}^{\ell}}(b_{1},t)+\frac{\partial L_{\mathrm{B}}}{\partial\psi_{\mathrm{B}}^{\ell}(b_{1},t)}-\partial_{0}\left(\frac{\partial L_{\mathrm{B}}}{\partial\partial_{0}\psi_{\mathrm{B}}^{\ell}(b_{1},t)}\right) & =0,\quad\partial_{0}=\partial_{t},\quad\partial_{1}=\partial_{z};\label{eq:Sact3fk}\\
\frac{\partial L_{\mathrm{D}}}{\partial\partial_{1}\psi_{\mathrm{D}}^{\ell}}(b_{2},t)+\frac{\partial L_{\mathrm{B}}}{\partial\psi_{\mathrm{B}}^{\ell}(b_{2},t)}-\partial_{0}\left(\frac{\partial L_{\mathrm{B}}}{\partial\left(\partial_{0}\psi_{\mathrm{B}}^{\ell}(b_{2},t)\right)}\right) & =0.\nonumber 
\end{align}
where (i) $b_{1}=a\ell_{0}-0$ and $b_{2}=a\ell_{0}+0$; (ii) $L_{\mathrm{D}}$
corresponds to $\mathcal{L}_{\mathrm{B}}+\mathcal{L}_{\mathrm{CB}}$;
(iii) $L_{\mathrm{B}}$ corresponds to $\mathcal{L}_{\mathrm{CB}}$;
(iv) fields $\psi_{\mathrm{D}}^{\ell}$ correspond to charges $Q$
and $q$; (v) boundary fields $\psi_{\mathrm{B}}^{\ell}$ correspond
to $Q\left(a\ell_{0},t\right)$ and $q\left(a\ell_{0},t\right)$.
We remind the reader that boundary conditions (2.12) in \cite{FigRey2}
is an implementation of the ``rigidity'' requirement which is appropriate
for Lagrangian $\mathcal{L}_{\mathrm{CB}}$ defined by equation (\ref{eq:aZep1dk}).
If fact, the signs of the terms containing $L_{\mathrm{D}}$ in equations
(\ref{eq:Sact3fk}) are altered compare to original equations (2.12)
in \cite{FigRey2} to correct an unfortunate typo there.

\emph{Thus equations (\ref{eq:Sact2gk}), (\ref{eq:Sact3ek}) and
(\ref{eq:Sact3gk}) form a complete set of the EL equations}.

\textbf{\vspace{0.2cm}
}

\textbf{ACKNOWLEDGMENT:} This research was supported by AFOSR MURI
Grant FA9550-20-1-0409 administered through the University of New
Mexico. The author is grateful to E. Schamiloglu for sharing his deep
and vast knowledge of high power microwave devices and inspiring discussions.\textbf{\vspace{0.2cm}
}

\textbf{NOMENCLATURE:} 
\begin{itemize}
\item EL stands for the Euler-Lagrange (equations)
\item HF stands for high-frequency
\item MCK stands for multi-cavity klystron
\item $\mathbb{C}$ is a set of complex number.
\item $\bar{s}$ is complex-conjugate to complex number $s$
\item $\mathbb{C}^{n}$ is a set of $n$ dimensional column vectors with
complex complex-valued entries.
\item $\mathbb{C}^{n\times m}$ is a set of $n\times m$ matrices with complex-valued
entries.
\item $\mathbb{R}^{n\times m}$ is a set of $n\times m$ matrices with real-valued
entries.
\item $\mathrm{diag}\,\left(A_{1},A_{2},\ldots,A_{r}\right)$ is bock diagonal
matrix with indicated blocks.
\item $\mathrm{dim}\,\left(W\right)$ is the dimension of the vector space
$W$.
\item $\mathrm{ker}\,\left(A\right)$ is the kernel of matrix $A$, that
is the vector space of vector $x$ such that $Ax=0$.
\item $\det\left\{ A\right\} $ is the determinant of matrix $A$.
\item $\sigma$$\left\{ A\right\} $ is the spectrum of matrix $A$.
\item $\chi_{A}\left(s\right)=\det\left\{ s\mathbb{I}_{\nu}-A\right\} $
is the characteristic polynomial of a $\nu\times\nu$ matrix $A$.
\item $\mathbb{I}_{\nu}$ is $\nu\times\nu$ identity matrix.
\item $M^{\mathrm{T}}$ is a matrix transposed to matrix $M$.
\item EL stands for the Euler-Lagrange (equations).
\item ODE stands for ordinary differential equations.
\end{itemize}
\renewcommand{\sectionname}{Appendix}
\counterwithin{section}{part}
\renewcommand{\thesection}{\Alph{section}}
\setcounter{section}{0}
\renewcommand{\theequation}{\Alph{section}.\arabic{equation}}

\section{Fourier transform\label{sec:four}}

Our preferred form of the Fourier transforms as in \cite[7.2, 7.5]{Foll},
\cite[20.2]{ArfWeb}:
\begin{gather}
f\left(t\right)=\int_{-\infty}^{\infty}\hat{f}\left(\omega\right)\mathrm{e}^{-\mathrm{i}\omega t}\,\mathrm{d}\omega,\quad\hat{f}\left(\omega\right)=\frac{1}{2\pi}\int_{-\infty}^{\infty}f\left(t\right)e^{\mathrm{i}\omega t}\,\mathrm{d}t,\label{eq:fourier1a}\\
f\left(z,t\right)=\int_{-\infty}^{\infty}\hat{f}\left(k,\omega\right)\mathrm{e}^{-\mathrm{i}\left(\omega t-kz\right)}\,\mathrm{d}k\mathrm{d}\omega,\label{eq:fourier1b}\\
\hat{f}\left(k,\omega\right)=\frac{1}{\left(2\pi\right)^{2}}\int_{-\infty}^{\infty}f\left(z,t\right)e^{\mathrm{i}\left(\omega t-kz\right)}\,dz\mathrm{d}t.\nonumber 
\end{gather}
This preference was motivated by the fact that the so-defined Fourier
transform of the convolution of two functions has its simplest form.
Namely, the convolution $f\ast g$ of two functions $f$ and $g$
is defined by \cite[7.2, 7.5]{Foll},
\begin{gather}
\left[f\ast g\right]\left(t\right)=\left[g\ast f\right]\left(t\right)=\int_{-\infty}^{\infty}f\left(t-t^{\prime}\right)g\left(t^{\prime}\right)\,\mathrm{d}t^{\prime},\label{eq:fourier2a}\\
\left[f\ast g\right]\left(z,t\right)=\left[g\ast f\right]\left(z,t\right)=\int_{-\infty}^{\infty}f\left(z-z^{\prime},t-t^{\prime}\right)g\left(z^{\prime},t^{\prime}\right)\,\mathrm{d}z^{\prime}\mathrm{d}t^{\prime}.\label{eq:fourier2b}
\end{gather}
 Then its Fourier transform as defined by equations (\ref{eq:fourier1a})
and (\ref{eq:fourier1b}) satisfies the following properties:
\begin{gather}
\widehat{f\ast g}\left(\omega\right)=\hat{f}\left(\omega\right)\hat{g}\left(\omega\right),\label{eq:fourier3a}\\
\widehat{f\ast g}\left(k,\omega\right)=\hat{f}\left(k,\omega\right)\hat{g}\left(k,\omega\right).\label{eq:fourier3b}
\end{gather}

\section{Jordan canonical form\label{sec:jord-form}}

We provide here very concise review of Jordan canonical forms following
mostly to \cite[III.4]{Hale}, \cite[3.1,3.2]{HorJohn}. As to a demonstration
of how Jordan block arises in the case of a single $n$-th order differential
equation we refer to \cite[25.4]{ArnODE}.

Let $A$ be an $n\times n$ matrix and $\lambda$ be its eigenvalue,
and let $r\left(\lambda\right)$ be the least integer $k$ such that
$\mathcal{N}\left[\left(A-\lambda\mathbb{I}\right)^{k}\right]=\mathcal{N}\left[\left(A-\lambda\mathbb{I}\right)^{k+1}\right]$,
where $\mathcal{N}\left[C\right]$ is a null space of a matrix $C$.
Then we refer to $M_{\lambda}=\mathcal{N}\left[\left(A-\lambda\mathbb{I}\right)^{r\left(\lambda\right)}\right]$
is the \emph{generalized eigenspace} of matrix $A$ corresponding
to eigenvalue $\lambda$. Then the following statements hold, \cite[III.4]{Hale}.
\begin{prop}[generalized eigenspaces]
\label{prop:gen-eig} Let $A$ be an $n\times n$ matrix and $\lambda_{1},\ldots,\lambda_{p}$
be its distinct eigenvalues. Then generalized eigenspaces $M_{\lambda_{1}},\ldots,M_{\lambda_{p}}$
are linearly independent, invariant under the matrix $A$ and
\begin{equation}
\mathbb{C}^{n}=M_{\lambda_{1}}\oplus\ldots\oplus M_{\lambda_{p}}.\label{eq:Mgeneig1a}
\end{equation}
Consequently, any vector $x_{0}$ in $\mathbb{C}^{n}$can be represented
uniquely as
\begin{equation}
x_{0}=\sum_{j=1}^{p}x_{0,j},\quad x_{0,j}\in M_{\lambda_{j}},\label{eq:Mgeneig1b}
\end{equation}
and
\begin{equation}
\exp\left\{ At\right\} x_{0}=\sum_{j=1}^{p}e^{\lambda_{j}t}p_{j}\left(t\right),\label{eq:Mgeneig1c}
\end{equation}
where column-vector polynomials $p_{j}\left(t\right)$ satisfy
\begin{gather}
p_{j}\left(t\right)=\sum_{k=0}^{r\left(\lambda_{j}\right)-1}\left(A-\lambda_{j}\mathbb{I}\right)^{k}\frac{t^{k}}{k!}x_{0,j},\quad x_{0,j}\in M_{\lambda_{j}},\quad1\leq j\leq p.\label{eq:Mgeneig1d}
\end{gather}
\end{prop}

For a complex number $\lambda$ a Jordan block $J_{r}\left(\lambda\right)$
of size $r\geq1$ is a $r\times r$ upper triangular matrix of the
form
\begin{gather}
J_{r}\left(\lambda\right)=\lambda\mathbb{I}_{r}+K_{r}=\left[\begin{array}{ccccc}
\lambda & 1 & \cdots & 0 & 0\\
0 & \lambda & 1 & \cdots & 0\\
0 & 0 & \ddots & \cdots & \vdots\\
\vdots & \vdots & \ddots & \lambda & 1\\
0 & 0 & \cdots & 0 & \lambda
\end{array}\right],\quad J_{1}\left(\lambda\right)=\left[\lambda\right],\quad J_{2}\left(\lambda\right)=\left[\begin{array}{cc}
\lambda & 1\\
0 & \lambda
\end{array}\right],\label{eq:Jork1a}
\end{gather}
\begin{equation}
K_{r}=J_{r}\left(0\right)=\left[\begin{array}{ccccc}
0 & 1 & \cdots & 0 & 0\\
0 & 0 & 1 & \cdots & 0\\
0 & 0 & \ddots & \cdots & \vdots\\
\vdots & \vdots & \ddots & 0 & 1\\
0 & 0 & \cdots & 0 & 0
\end{array}\right].\label{eq:Jork1k}
\end{equation}
The special Jordan block $K_{r}=J_{r}\left(0\right)$ defined by equation
(\ref{eq:Jork1k}) is an nilpotent matrix that satisfies the following
identities

\begin{gather}
K_{r}^{2}=\left[\begin{array}{ccccc}
0 & 0 & 1 & \cdots & 0\\
0 & 0 & 0 & \cdots & \vdots\\
0 & 0 & \ddots & \cdots & 1\\
\vdots & \vdots & \ddots & 0 & 0\\
0 & 0 & \cdots & 0 & 0
\end{array}\right],\cdots,\;K_{r}^{r-1}=\left[\begin{array}{ccccc}
0 & 0 & \cdots & 0 & 1\\
0 & 0 & 0 & \cdots & 0\\
0 & 0 & \ddots & \cdots & \vdots\\
\vdots & \vdots & \ddots & 0 & 0\\
0 & 0 & \cdots & 0 & 0
\end{array}\right],\quad K_{r}^{r}=0.\label{eq:Jord1a}
\end{gather}
A general Jordan $n\times n$ matrix $J$ is defined as a direct sum
of Jordan blocks, that is

\begin{equation}
J=\left[\begin{array}{ccccc}
J_{n_{1}}\left(\lambda_{1}\right) & 0 & \cdots & 0 & 0\\
0 & J_{n_{2}}\left(\lambda_{2}\right) & 0 & \cdots & 0\\
0 & 0 & \ddots & \cdots & \vdots\\
\vdots & \vdots & \ddots & J_{n_{q-1}}\left(\lambda_{n_{q}-1}\right) & 0\\
0 & 0 & \cdots & 0 & J_{n_{q}}\left(\lambda_{n_{q}}\right)
\end{array}\right],\quad n_{1}+n_{2}+\cdots n_{q}=n,\label{eq:Jork1b}
\end{equation}
where $\lambda_{j}$ need not be distinct. Any square matrix $A$
is similar to a Jordan matrix as in equation (\ref{eq:Jork1b}) which
is called \emph{Jordan canonical form} of $A$. Namely, the following
statement holds, \cite[3.1]{HorJohn}.
\begin{prop}[Jordan canonical form]
\label{prop:jor-can} Let $A$ be an $n\times n$ matrix. Then there
exists a non-singular $n\times n$ matrix $Q$ such that the following
block-diagonal representation holds
\begin{equation}
Q^{-1}AQ=J\label{eq:QAQC1a}
\end{equation}
where $J$ is the Jordan matrix defined by equation (\ref{eq:Jork1b})
and $\lambda_{j}$, $1\leq j\leq q$ are not necessarily different
eigenvalues of matrix $A$. Representation (\ref{eq:QAQC1a}) is known
as the \emph{Jordan canonical form} of matrix $A$, and matrices $J_{j}$
are called \emph{Jordan blocks}. The columns of the $n\times n$ matrix
$Q$ constitute the \emph{Jordan basis} providing for the Jordan canonical
form (\ref{eq:QAQC1a}) of matrix $A$.
\end{prop}

A function $f\left(J_{r}\left(s\right)\right)$ of a Jordan block
$J_{r}\left(s\right)$ is represented by the following equation \cite[7.9]{MeyCD},
\cite[10.5]{BernM}

\begin{gather}
f\left(J_{r}\left(s\right)\right)=\left[\begin{array}{ccccc}
f\left(s\right) & \partial f\left(s\right) & \frac{\partial^{2}f\left(s\right)}{2} & \cdots & \frac{\partial^{r-1}f\left(s\right)}{\left(r-1\right)!}\\
0 & f\left(s\right) & \partial f\left(s\right) & \cdots & \frac{\partial^{r-2}f\left(s\right)}{\left(r-2\right)!}\\
0 & 0 & \ddots & \cdots & \vdots\\
\vdots & \vdots & \ddots & f\left(s\right) & \partial f\left(s\right)\\
0 & 0 & \cdots & 0 & f\left(s\right)
\end{array}\right].\label{eq:JJordf1a}
\end{gather}
Notice that any function $f\left(J_{r}\left(s\right)\right)$ of the
Jordan block $J_{r}\left(s\right)$ is evidently an upper triangular
Toeplitz matrix.

There are two particular cases of formula (\ref{eq:JJordf1a}) which
can be also derived straightforwardly using equations (\ref{eq:Jord1a}):
\begin{gather}
\exp\left\{ K_{r}t\right\} =\sum_{k=0}^{r-1}\frac{t^{k}}{k!}K_{r}^{k}=\left[\begin{array}{ccccc}
1 & t & \frac{t^{2}}{2!} & \cdots & \frac{t^{r-1}}{\left(r-1\right)!}\\
0 & 1 & t & \cdots & \frac{t^{r-2}}{\left(r-2\right)!}\\
0 & 0 & \ddots & \cdots & \vdots\\
\vdots & \vdots & \ddots & 1 & t\\
0 & 0 & \cdots & 0 & 1
\end{array}\right],\label{eq:Jord1c}
\end{gather}
 
\begin{gather}
\left[J_{r}\left(s\right)\right]^{-1}=\sum_{k=0}^{r-1}s^{-k-1}\left(-K_{r}\right)^{k}=\left[\begin{array}{ccccc}
\frac{1}{s} & -\frac{1}{s^{2}} & \frac{1}{s^{3}} & \cdots & \frac{\left(-1\right)^{r-1}}{s^{r}}\\
0 & \frac{1}{s} & -\frac{1}{s^{2}} & \cdots & \frac{\left(-1\right)^{r-2}}{s^{r-1}}\\
0 & 0 & \ddots & \cdots & \vdots\\
\vdots & \vdots & \ddots & \frac{1}{s} & -\frac{1}{s^{2}}\\
0 & 0 & \cdots & 0 & \frac{1}{s}
\end{array}\right].\label{eq:JJordf1b}
\end{gather}

\section{Companion matrix and cyclicity condition\label{sec:co-mat}}

The companion matrix $C\left(a\right)$ for monic polynomial
\begin{equation}
a\left(s\right)=s^{\nu}+\sum_{1\leq k\leq\nu}a_{\nu-k}s^{\nu-k}\label{eq:compas1a}
\end{equation}
where coefficients $a_{k}$ are complex numbers is defined by \cite[5.2]{BernM}
\begin{equation}
C\left(a\right)=\left[\begin{array}{ccccc}
0 & 1 & \cdots & 0 & 0\\
0 & 0 & 1 & \cdots & 0\\
0 & 0 & 0 & \cdots & \vdots\\
\vdots & \vdots & \ddots & 0 & 1\\
-a_{0} & -a_{1} & \cdots & -a_{\nu-2} & -a_{\nu-1}
\end{array}\right].\label{eq:compas1b}
\end{equation}
Notice that
\begin{equation}
\det\left\{ C\left(a\right)\right\} =\left(-1\right)^{\nu}a_{0}.\label{eq:compas1c}
\end{equation}

An eigenvalue is called \emph{cyclic (nonderogatory)} if its geometric
multiplicity is 1. A square matrix is called \emph{cyclic (nonderogatory)}
if all its eigenvalues are cyclic \cite[5.5]{BernM}. The following
statement provides different equivalent descriptions of a cyclic matrix
\cite[5.5]{BernM}.
\begin{prop}[criteria for a matrix to be cyclic]
\label{prop:cyc1} Let $A\in\mathbb{C}^{n\times n}$ be $n\times n$
matrix with complex-valued entries. Let $\mathrm{spec}\,\left(A\right)=\left\{ \zeta_{1},\zeta_{2},\ldots,\zeta_{r}\right\} $
be the set of all distinct eigenvalues and $k_{j}=\mathrm{ind}{}_{A}\,\left(\zeta_{j}\right)$
is the largest size of Jordan block associated with $\zeta_{j}$.
Then the minimal polynomial $\mu_{A}\left(s\right)$ of the matrix
$A$, that is a monic polynomial of the smallest degree such that
$\mu_{A}\left(A\right)=0$, satisfies
\begin{equation}
\mu_{A}\left(s\right)=\prod_{j=1}^{r}\left(s-\zeta_{j}\right)^{k_{j}}.\label{eq:compas1d}
\end{equation}
Furthermore, and following statements are equivalent:
\end{prop}

\begin{enumerate}
\item $\mu_{A}\left(s\right)=\chi_{A}\left(s\right)=\det\left\{ s\mathbb{I}-A\right\} $.
\item $A$ is cyclic.
\item For every $\zeta_{j}$ the Jordan form of $A$ contains exactly one
block associated with $\zeta_{j}$.
\item $A$ is similar to the companion matrix $C\left(\chi_{A}\right)$.
\end{enumerate}
\begin{prop}[companion matrix factorization]
\label{prop:cyc2} Let $a\left(s\right)$ be a monic polynomial having
degree $\nu$ and $C\left(a\right)$ is its $\nu\times\nu$ companion
matrix. Then, there exist unimodular $\nu\times\nu$ matrices $S_{1}\left(s\right)$
and $S_{2}\left(s\right)$, that is $\det\left\{ S_{m}\right\} =\pm1$,
$m=1,2$, such that
\begin{equation}
s\mathbb{I}_{\nu}-C\left(a\right)=S_{1}\left(s\right)\left[\begin{array}{lr}
\mathbb{I}_{\nu-1} & 0_{\left(\nu-1\right)\times1}\\
0_{1\times\left(\nu-1\right)} & a\left(s\right)
\end{array}\right]S_{2}\left(s\right).\label{eq:compas1e}
\end{equation}
Consequently, $C\left(a\right)$ is cyclic and
\begin{equation}
\chi_{C\left(a\right)}\left(s\right)=\mu_{C\left(a\right)}\left(s\right)=a\left(s\right).\label{eq:compas1f}
\end{equation}
\end{prop}

The following statement summarizes important information on the Jordan
form of the companion matrix and the generalized Vandermonde matrix,
\cite[5.16]{BernM}, \cite[2.11]{LanTsi}, \cite[7.9]{MeyCD}.
\begin{prop}[Jordan form of the companion matrix]
\label{prop:cycJ} Let $C\left(a\right)$ be an $n\times n$ a companion
matrix of the monic polynomial $a\left(s\right)$ defined by equation
(\ref{eq:compas1a}). Suppose that the set of distinct roots of polynomial
$a\left(s\right)$ is $\left\{ \zeta_{1},\zeta_{2},\ldots,\zeta_{r}\right\} $
and $\left\{ n_{1},n_{2},\ldots,n_{r}\right\} $ is the corresponding
set of the root multiplicities such that
\begin{equation}
n_{1}+n_{2}+\cdots+n_{r}=n.\label{eq:compas2a}
\end{equation}
Then 
\begin{equation}
C\left(a\right)=RJR^{-1},\label{eq:compas2b}
\end{equation}
where
\begin{equation}
J=\mathrm{diag}\,\left\{ J_{n_{1}}\left(\zeta_{1}\right),J_{n_{2}}\left(\zeta_{2}\right),\ldots,J_{n_{r}}\left(\zeta_{r}\right)\right\} \label{eq:compas2c}
\end{equation}
is the the Jordan form of companion matrix $C\left(a\right)$ and
$n\times n$ matrix $R$ is the so-called generalized Vandermonde
matrix defined by
\begin{equation}
R=\left[R_{1}|R_{2}|\cdots|R_{r}\right],\label{eq:compas2d}
\end{equation}
 where $R_{j}$ is $n\times n_{j}$ matrix of the form
\begin{equation}
R_{j}=\left[\begin{array}{rrcr}
1 & 0 & \cdots & 0\\
\zeta_{j} & 1 & \cdots & 0\\
\vdots & \vdots & \ddots & \vdots\\
\zeta_{j}^{n-2} & \binom{n-2}{1}\,\zeta_{j}^{n-3} & \cdots & \binom{n-2}{n_{j}-1}\,\zeta_{j}^{n-n_{j}-1}\\
\zeta_{j}^{n-1} & \binom{n-1}{1}\,\zeta_{j}^{n-2} & \cdots & \binom{n-1}{n_{j}-1}\,\zeta_{j}^{n-n_{j}}
\end{array}\right].\label{eq:compas2f}
\end{equation}
As a consequence of representation (\ref{eq:compas2c}) $C\left(a\right)$
is a cyclic matrix.
\end{prop}

As to the structure of matrix $R_{j}$ in equation (\ref{eq:compas2f}),
if we denote by $Y\left(\zeta_{j}\right)$ its first column then it
can be expressed as follows \cite[2.11]{LanTsi}:
\begin{equation}
R_{j}=\left[Y^{\left(0\right)}|Y^{\left(1\right)}|\cdots|Y^{\left(n_{j}-1\right)}\right],\quad Y^{\left(m\right)}=\frac{1}{m!}\partial_{s_{j}}^{m}Y\left(\zeta_{j}\right),\quad0\leq m\leq n_{j}-1.\label{eq:compas3a}
\end{equation}
In the case when all eigenvalues of a cyclic matrix are distinct then
the generalized Vandermonde matrix turns into the standard Vandermonde
matrix
\begin{equation}
V=\left[\begin{array}{rrcr}
1 & 1 & \cdots & 1\\
\zeta_{1} & \zeta_{2} & \cdots & \zeta_{n}\\
\vdots & \vdots & \ddots & \vdots\\
\zeta_{1}^{n-2} & \zeta_{2}^{n-2} & \cdots & \zeta_{n}^{n-2}\\
\zeta_{1}^{n-1} & \zeta_{2}^{n-1} & \cdots & \zeta_{n}^{n-1}
\end{array}\right].\label{eq:compas3c}
\end{equation}

\section{Matrix polynomials\label{sec:matpol}}

An important incentive for considering matrix polynomials is that
they are relevant to the spectral theory of the differential equations
of the order higher than 1, particularly the Euler-Lagrange equations
which are the second-order differential equations in time. We provide
here selected elements of the theory of matrix polynomials following
mostly to \cite[II.7, II.8]{GoLaRo}, \cite[9]{Baum}. General matrix
polynomial eigenvalue problem reads
\begin{equation}
A\left(s\right)x=0,\quad A\left(s\right)=\sum_{j=0}^{\nu}A_{j}s^{j},\quad x\neq0,\label{eq:Aux1a}
\end{equation}
where $s$ is complex number, $A_{k}$ are constant $m\times m$ matrices
and $x\in\mathbb{C}^{m}$ is $m$-dimensional column-vector. We refer
to problem (\ref{eq:Aux1a}) of funding complex-valued $s$ and non-zero
vector $x\in\mathbb{C}^{m}$ as polynomial eigenvalue problem. 

If a pair of a complex $s$ and non-zero vector $x$ solves problem
(\ref{eq:Aux1a}) we refer to $s$ as an \emph{eigenvalue} or as a\emph{
characteristic value} and to $x$ as the corresponding to $s$ \emph{eigenvector}.
Evidently the characteristic values of problem (\ref{eq:Aux1a}) can
be found from polynomial \emph{characteristic equation}
\begin{equation}
\det\left\{ A\left(s\right)\right\} =0.\label{eq:Aux1b}
\end{equation}
We refer to matrix polynomial $A\left(s\right)$ as \emph{regular}
if $\det\left\{ A\left(s\right)\right\} $ is not identically zero.
We denote by $m\left(s_{0}\right)$ the \emph{multiplicity} (called
also \emph{algebraic multiplicity}) of eigenvalue $s_{0}$ as a root
of polynomial $\det\left\{ A\left(s\right)\right\} $. In contrast,
the \emph{geometric multiplicity} of eigenvalue $s_{0}$ is defined
as $\dim\left\{ \ker\left\{ A\left(s_{0}\right)\right\} \right\} $,
where $\ker\left\{ A\right\} $ defined for any square matrix $A$
stands for the subspace of solutions $x$ to equation $Ax=0$. Evidently,
the geometric multiplicity of eigenvalue does not exceed its algebraic
one, see Corollary \ref{cor:dim-ker}. 

It turns out that the matrix polynomial eigenvalue problem (\ref{eq:Aux1a})
can be always recast as the standard ``linear'' eigenvalue problem,
namely
\begin{equation}
\left(s\mathsf{B}-\mathsf{A}\right)\mathsf{x}=0,\label{eq:Aux1c}
\end{equation}
where $m\nu\times m\nu$ matrices $\mathsf{A}$ and $\mathsf{B}$
are defined by
\begin{gather}
\mathsf{B}=\left[\begin{array}{ccccc}
\mathbb{I} & 0 & \cdots & 0 & 0\\
0 & \mathbb{I} & 0 & \cdots & 0\\
0 & 0 & \ddots & \cdots & \vdots\\
\vdots & \vdots & \ddots & \mathbb{I} & 0\\
0 & 0 & \cdots & 0 & A_{\nu}
\end{array}\right],\quad\mathsf{A}=\left[\begin{array}{ccccc}
0 & \mathbb{I} & \cdots & 0 & 0\\
0 & 0 & \mathbb{I} & \cdots & 0\\
0 & 0 & 0 & \cdots & \vdots\\
\vdots & \vdots & \ddots & 0 & \mathbb{I}\\
-A_{0} & -A_{1} & \cdots & -A_{\nu-2} & -A_{\nu-1}
\end{array}\right],\label{eq:CBA1b}
\end{gather}
with $\mathbb{I}$ being $m\times m$ identity matrix. Matrix $\mathsf{A}$,
particularly in monic case, is often referred to as \emph{companion
matrix}. In the case of \emph{monic polynomial} $A\left(\lambda\right)$,
when $A_{\nu}=\mathbb{I}$ is $m\times m$ identity matrix, matrix
$\mathsf{B}=\mathsf{I}$ is $m\nu\times m\nu$ identity matrix. The
reduction of original polynomial problem (\ref{eq:Aux1a}) to an equivalent
linear problem (\ref{eq:Aux1c}) is called \emph{linearization}.

The linearization is not unique, and one way to accomplish is by introducing
the so-called known ``\emph{companion polynomia}l'' which is $m\nu\times m\nu$
matrix
\begin{gather}
\mathsf{C}_{A}\left(s\right)=s\mathsf{B}-\mathsf{A}=\left[\begin{array}{ccccc}
s\mathbb{I} & -\mathbb{I} & \cdots & 0 & 0\\
0 & s\mathbb{I} & -\mathbb{I} & \cdots & 0\\
0 & 0 & \ddots & \cdots & \vdots\\
\vdots & \vdots & \vdots & s\mathbb{I} & -\mathbb{I}\\
A_{0} & A_{1} & \cdots & A_{\nu-2} & sA_{\nu}+A_{\nu-1}
\end{array}\right].\label{eq:CBA1a}
\end{gather}
Notice that in the case of the EL equations the linearization can
be accomplished by the relevant Hamilton equations.

To demonstrate the equivalency between the eigenvalue problems for
$m\nu\times m\nu$ companion polynomial $\mathsf{C}_{A}\left(s\right)$
and the original $m\times m$ matrix polynomial $A\left(s\right)$
we introduce two $m\nu\times m\nu$ matrix polynomials $\mathsf{E}\left(s\right)$
and $\mathsf{F}\left(s\right)$. Namely,
\begin{gather}
\mathsf{E}\left(s\right)=\left[\begin{array}{ccccc}
E_{1}\left(s\right) & E_{2}\left(s\right) & \cdots & E_{\nu-1}\left(s\right) & \mathbb{I}\\
-\mathbb{I} & 0 & 0 & \cdots & 0\\
0 & -\mathbb{I} & \ddots & \cdots & \vdots\\
\vdots & \vdots & \ddots & 0 & 0\\
0 & 0 & \cdots & -\mathbb{I} & 0
\end{array}\right],\label{eq:CBA1c}\\
\det\left\{ \mathsf{E}\left(s\right)\right\} =1,\nonumber 
\end{gather}
where $m\times m$ matrix polynomials $E_{j}\left(s\right)$ are defined
by the following recursive formulas
\begin{gather}
E_{\nu}\left(s\right)=A_{\nu},\quad E_{j-1}\left(s\right)=A_{j-1}+sE_{j}\left(s\right),\quad j=\nu,\ldots,2.\label{eq:CBA1d}
\end{gather}
Matrix polynomial $\mathsf{F}\left(s\right)$ is defined by
\begin{gather}
\mathsf{F}\left(s\right)=\left[\begin{array}{ccccc}
\mathbb{I} & 0 & \cdots & 0 & 0\\
-s\mathbb{I} & \mathbb{I} & 0 & \cdots & 0\\
0 & -s\mathbb{I} & \ddots & \cdots & \vdots\\
\vdots & \vdots & \ddots & \mathbb{I} & 0\\
0 & 0 & \cdots & -s\mathbb{I} & \mathbb{I}
\end{array}\right],\quad\det\left\{ \mathsf{F}\left(s\right)\right\} =1.\label{eq:CBA1e}
\end{gather}
Notice, that both matrix polynomials $\mathsf{E}\left(s\right)$ and
$\mathsf{F}\left(s\right)$ have constant determinants readily implying
that their inverses $\mathsf{E}^{-1}\left(s\right)$ and $\mathsf{F}^{-1}\left(s\right)$
are also matrix polynomials. Then it is straightforward to verify
that
\begin{gather}
\mathsf{E}\left(s\right)\mathsf{C}_{A}\left(s\right)\mathsf{F}^{-1}\left(s\right)=\mathsf{E}\left(s\right)\left(s\mathsf{B}-\mathsf{A}\right)\mathsf{F}^{-1}\left(s\right)=\left[\begin{array}{ccccc}
A\left(s\right) & 0 & \cdots & 0 & 0\\
0 & \mathbb{I} & 0 & \cdots & 0\\
0 & 0 & \ddots & \cdots & \vdots\\
\vdots & \vdots & \ddots & \mathbb{I} & 0\\
0 & 0 & \cdots & 0 & \mathbb{I}
\end{array}\right].\label{eq:CBA1f}
\end{gather}
The identity (\ref{eq:CBA1f}) where matrix polynomials $\mathsf{E}\left(s\right)$
and $\mathsf{F}\left(s\right)$ have constant determinants can be
viewed as the definition of equivalency between matrix polynomial
$A\left(s\right)$ and its companion polynomial $\mathsf{C}_{A}\left(s\right)$. 

Let us take a look at the eigenvalue problem for eigenvalue $s$ and
eigenvector $\mathsf{x}\in\mathbb{C}^{m\nu}$ associated with companion
polynomial $\mathsf{C}_{A}\left(s\right)$, that is
\begin{gather}
\left(s\mathsf{B}-\mathsf{A}\right)\mathsf{x}=0,\quad\mathsf{x}=\left[\begin{array}{c}
x_{0}\\
x_{1}\\
x_{2}\\
\vdots\\
x_{\nu-1}
\end{array}\right]\in\mathbb{C}^{m\nu},\quad x_{j}\in\mathbb{C}^{m},\quad0\leq j\leq\nu-1,\label{eq:CBAx1a}
\end{gather}
where
\begin{equation}
\left(s\mathsf{B}-\mathsf{A}\right)\mathsf{x}=\left[\begin{array}{c}
sx_{0}-x_{1}\\
sx_{1}-x_{2}\\
\vdots\\
sx_{\nu-2}-x_{\nu-1}\\
\sum_{j=0}^{\nu-2}A_{j}x_{j}+\left(sA_{\nu}+A_{\nu-1}\right)x_{\nu-1}
\end{array}\right].\label{eq:CBAx1b}
\end{equation}
With equations (\ref{eq:CBAx1a}) and (\ref{eq:CBAx1b}) in mind we
introduce the following vector polynomial
\begin{equation}
\mathsf{x}_{s}=\left[\begin{array}{c}
x_{0}\\
sx_{0}\\
\vdots\\
s^{\nu-2}x_{0}\\
s^{\nu-1}x_{0}
\end{array}\right],\quad x_{0}\in\mathbb{C}^{m}.\label{eq:CBAx1c}
\end{equation}
Not accidentally, the components of the vector $\mathsf{x}_{s}$ in
its representation (\ref{eq:CBAx1c}) are in evident relation with
the derivatives $\partial_{t}^{j}\left(x_{0}\mathrm{e}^{st}\right)=s^{j}x_{0}\mathrm{e}^{st}$.
That is just another sign of the intimate relations between the matrix
polynomial theory and the theory of systems of ordinary differential
equations, see Appendix \ref{sec:dif-jord}. 
\begin{thm}[eigenvectors]
\label{thm:matpol-eigvec} Let $A\left(s\right)$ as in equations
(\ref{eq:Aux1a}) be regular, that $\det\left\{ A\left(s\right)\right\} $
is not identically zero, and let $m\nu\times m\nu$ matrices $\mathsf{A}$
and $\mathsf{B}$ be defined by equations (\ref{eq:Aux1b}). Then
the following identities hold
\begin{equation}
\left(s\mathsf{B}-\mathsf{A}\right)\mathsf{x}_{s}=\left[\begin{array}{c}
0\\
0\\
\vdots\\
0\\
A\left(s\right)x_{0}
\end{array}\right],\;\mathsf{x}_{s}=\left[\begin{array}{c}
x_{0}\\
sx_{0}\\
\vdots\\
s^{\nu-2}x_{0}\\
s^{\nu-1}x_{0}
\end{array}\right],\label{eq:CBAx1d}
\end{equation}
\begin{gather}
\det\left\{ A\left(s\right)\right\} =\det\left\{ s\mathsf{B}-\mathsf{A}\right\} ,\quad\det\left\{ \mathsf{B}\right\} =\det\left\{ A_{\nu}\right\} ,\label{eq:CBAx1g}
\end{gather}
where $\det\left\{ A\left(s\right)\right\} =\det\left\{ s\mathsf{B}-\mathsf{A}\right\} $
is a polynomial of the degree $m\nu$ if $\det\left\{ \mathsf{B}\right\} =\det\left\{ A_{\nu}\right\} \neq0$.
There is one-to-one correspondence between solutions of equations
$A\left(s\right)x=0$ and $\left(s\mathsf{B}-\mathsf{A}\right)\mathsf{x}=0$.
Namely, a pair $s,\:\mathsf{x}$ solves eigenvalue problem $\left(s\mathsf{B}-\mathsf{A}\right)\mathsf{x}=0$
if and only if the following equalities hold
\begin{gather}
\mathsf{x}=\mathsf{x}_{s}=\left[\begin{array}{c}
x_{0}\\
sx_{0}\\
\vdots\\
s^{\nu-2}x_{0}\\
s^{\nu-1}x_{0}
\end{array}\right],\quad A\left(s\right)x_{0}=0,\quad x_{0}\neq0;\quad\det\left\{ A\left(s\right)\right\} =0.\label{eq:CBAx1e}
\end{gather}
\end{thm}

\begin{proof}
Polynomial vector identity (\ref{eq:CBAx1d}) readily follows from
equations (\ref{eq:CBAx1b}) and (\ref{eq:CBAx1c}). Identities (\ref{eq:CBAx1g})
for the determinants follow straightforwardly from equations (\ref{eq:CBAx1c}),
(\ref{eq:CBAx1e}) and (\ref{eq:CBA1f}). If $\det\left\{ \mathsf{B}\right\} =\det\left\{ A_{\nu}\right\} \neq0$
then the degree of the polynomial $\det\left\{ s\mathsf{B}-\mathsf{A}\right\} $
has to be $m\nu$ since $\mathsf{A}$ and $\mathsf{B}$ are $m\nu\times m\nu$
matrices.

Suppose that equations (\ref{eq:CBAx1e}) hold. Then combining them
with proven identity (\ref{eq:CBAx1d}) we get $\left(s\mathsf{B}-\mathsf{A}\right)\mathsf{x}_{s}=0$
proving that expressions (\ref{eq:CBAx1e}) define an eigenvalue $s$
and an eigenvector $\mathsf{x}=\mathsf{x}_{s}$.

Suppose now that $\left(s\mathsf{B}-\mathsf{A}\right)\mathsf{x}=0$
where $\mathsf{x}\neq0$. Combing that with equations (\ref{eq:CBAx1b})
we obtain
\begin{gather}
x_{1}=sx_{0},\quad x_{2}=sx_{1}=s^{2}x_{0},\cdots,\quad x_{\nu-1}=s^{\nu-1}x_{0},\label{eq:CBAx2a}
\end{gather}
implying that
\begin{equation}
\mathsf{x}=\mathsf{x}_{s}=\left[\begin{array}{c}
x_{0}\\
sx_{0}\\
\vdots\\
s^{\nu-2}x_{0}\\
s^{\nu-1}x_{0}
\end{array}\right],\quad x_{0}\neq0,\label{eq:CBAx2b}
\end{equation}
 and 
\begin{equation}
\sum_{j=0}^{\nu-2}A_{j}x_{j}+\left(sA_{\nu}+A_{\nu-1}\right)x_{\nu-1}=A\left(s\right)x_{0}.\label{eq:CBAx2c}
\end{equation}
Using equations (\ref{eq:CBAx2b}) and identity (\ref{eq:CBAx1d})
we obtain
\begin{equation}
0=\left(s\mathsf{B}-\mathsf{A}\right)\mathsf{x}=\left(s\mathsf{B}-\mathsf{A}\right)\mathsf{x}_{s}=\left[\begin{array}{c}
0\\
0\\
\vdots\\
0\\
A\left(s\right)x_{0}
\end{array}\right].\label{eq:CBAx2d}
\end{equation}
 Equations (\ref{eq:CBAx2d}) readily imply $A\left(s\right)x_{0}=0$
and $\det\left\{ A\left(s\right)\right\} =0$ since $x_{0}\neq0$.
That completes the proof.
\end{proof}
\begin{rem}[characteristic polynomial degree]
\label{rem:char-pol-deg} Notice that according to Theorem \ref{thm:matpol-eigvec}
the characteristic polynomial $\det\left\{ A\left(s\right)\right\} $
for $m\times m$ matrix polynomial $A\left(s\right)$ has the degree
$m\nu$, whereas in linear case $s\mathbb{I}-A_{0}$ for $m\times m$
identity matrix $\mathbb{I}$ and $m\times m$ matrix $A_{0}$ the
characteristic polynomial $\det\left\{ s\mathbb{I}-A_{0}\right\} $
is of the degree $m$. This can be explained by observing that in
the non-linear case of $m\times m$ matrix polynomial $A\left(s\right)$
we are dealing effectively with many more $m\times m$ matrices $A$
than just a single matrix $A_{0}$.
\end{rem}

Another problem of our particular interest related to the theory of
matrix polynomials is eigenvalues and eigenvectors degeneracy and
consequently the existence of non-trivial Jordan blocks, that is Jordan
blocks of dimensions higher or equal to 2. The general theory addresses
this problem by introducing so-called ``Jordan chains'' which are
intimately related to the theory of system of differential equations
expressed as $A\left(\partial_{t}\right)x\left(t\right)=0$ and their
solutions of the form $x\left(t\right)=p\left(t\right)e^{st}$ where
$p\left(t\right)$ is a vector polynomial, see Appendix \ref{sec:dif-jord}
and \cite[I, II]{GoLaRo}, \cite[9]{Baum}. Avoiding the details of
Jordan chains developments we simply notice that an important to us
point of Theorem \ref{thm:matpol-eigvec} is that there is one-to-one
correspondence between solutions of equations $A\left(s\right)x=0$
and $\left(s\mathsf{B}-\mathsf{A}\right)\mathsf{x}=0$, and it has
the following immediate implication.
\begin{cor}[equality of the dimensions of eigenspaces]
\label{cor:dim-ker} Under the conditions of Theorem \ref{thm:matpol-eigvec}
for any eigenvalue $s_{0}$, that is $\det\left\{ A\left(s_{0}\right)\right\} =0$,
we have
\begin{equation}
\dim\left\{ \ker\left\{ s_{0}\mathsf{B}-\mathsf{A}\right\} \right\} =\dim\left\{ \ker\left\{ A\left(s_{0}\right)\right\} \right\} .\label{eq:CBAx2e}
\end{equation}
In other words, the geometric multiplicities of the eigenvalue $s_{0}$
associated with matrices $A\left(s_{0}\right)$ and $s_{0}\mathsf{B}-\mathsf{A}$
are equal. In view of identity (\ref{eq:CBAx2e}) the following inequality
holds for the (algebraic) multiplicity $m\left(s_{0}\right)$
\begin{equation}
m\left(s_{0}\right)\geq\dim\left\{ \ker\left\{ A\left(s_{0}\right)\right\} \right\} .\label{eq:CBAx2f}
\end{equation}
\end{cor}

The next statement shows that if the geometric multiplicity of an
eigenvalue is strictly less than its algebraic one than there exist
non-trivial Jordan blocks, that is Jordan blocks of dimensions higher
or equal to 2.
\begin{thm}[non-trivial Jordan block]
\label{thm:Jord-block} Assuming notations introduced in Theorem
\ref{thm:matpol-eigvec} let us suppose that the multiplicity $m\left(s_{0}\right)$
of eigenvalue $s_{0}$ satisfies
\begin{equation}
m\left(s_{0}\right)>\dim\left\{ \ker\left\{ A\left(s_{0}\right)\right\} \right\} .\label{eq:CBAAx3a}
\end{equation}
Then the Jordan canonical form of companion polynomial $\mathsf{C}_{A}\left(s\right)=s\mathsf{B}-\mathsf{A}$
has a least one nontrivial Jordan block of the dimension exceeding
2.

In particular, if 
\begin{equation}
\dim\left\{ \ker\left\{ s_{0}\mathsf{B}-\mathsf{A}\right\} \right\} =\dim\left\{ \ker\left\{ A\left(s_{0}\right)\right\} \right\} =1,\label{eq:CBAAx3b}
\end{equation}
 and $m\left(s_{0}\right)\geq2$ then the Jordan canonical form of
companion polynomial $\mathsf{C}_{A}\left(s\right)=s\mathsf{B}-\mathsf{A}$
has exactly one Jordan block associated with eigenvalue $s_{0}$ and
its dimension is $m\left(s_{0}\right)$.
\end{thm}

The proof of Theorem \ref{thm:Jord-block} follows straightforwardly
from the definition of the Jordan canonical form and its basic properties.
Notice that if equations (\ref{eq:CBAAx3b}) hold that implies that
the eigenvalue $0$ is cyclic (nonderogatory) for matrix $A\left(s_{0}\right)$
and eigenvalue $s_{0}$ is cyclic (nonderogatory) for matrix $\mathsf{B}^{-1}\mathsf{A}$
provided $\mathsf{B}^{-1}$ exists, see Appendix  \ref{sec:co-mat}.

\section{Vector differential equations and the Jordan canonical form\label{sec:dif-jord}}

In this section we relate the vector ordinary differential equations
to the matrix polynomials reviewed in Appendix \ref{sec:matpol} following
to \cite[5.1, 5.7]{GoLaRo2}, \cite[II.8.3]{GoLaRo}, \cite[III.4]{Hale},
\cite[7.9]{MeyCD}.

Equation $A\left(s\right)x=0$ with polynomial matrix $A\left(s\right)$
defined by equations (\ref{eq:Aux1a}) corresponds to the following
$m$-vector $\nu$-th order ordinary differential
\begin{equation}
A\left(\partial_{t}\right)x\left(t\right)=0,\text{ where }A\left(\partial_{t}\right)=\sum_{j=0}^{\nu}A_{j}\partial_{t}^{j},\label{eq:Adtx1}
\end{equation}
where $A_{j}=A_{j}\left(t\right)$ are $m\times m$ matrices. Introducing
$m\nu$-column-vector function
\begin{equation}
Y\left(t\right)=\left[\begin{array}{c}
x\left(t\right)\\
\partial_{t}x\left(t\right)\\
\vdots\\
\partial_{t}^{\nu-2}x\left(t\right)\\
\partial_{t}^{\nu-1}x\left(t\right)
\end{array}\right]\label{eq:yxt1b}
\end{equation}
and under the assumption that matrix $A_{\nu}\left(t\right)$ is the
identity matrix the differential equation (\ref{eq:Adtx1}) can be
recast and the first order differential equation

\begin{equation}
\partial_{t}Y\left(t\right)=\mathsf{A}Y\left(t\right),\label{eq:yxt1ba}
\end{equation}
where $\mathsf{A}$ is $m\nu\times m\nu$ matrix defined by
\begin{gather}
\mathsf{A}=\mathsf{A}\left(t\right)=\left[\begin{array}{ccccc}
0 & \mathbb{I} & \cdots & 0 & 0\\
0 & 0 & \mathbb{I} & \cdots & 0\\
0 & 0 & 0 & \cdots & \vdots\\
\vdots & \vdots & \ddots & 0 & \mathbb{I}\\
-A_{0}\left(t\right) & -A_{1}\left(t\right) & \cdots & -A_{\nu-2}\left(t\right) & -A_{\nu-1}\left(t\right)
\end{array}\right],\quad A_{\nu}\left(t\right)=\mathbb{I}.\label{eq:yxt1bb}
\end{gather}

\subsection{Constant coefficients case}

Let us consider an important special case of equation (\ref{eq:Adtx1})
when matrices $A_{j}$ are $m\times m$ that do not depended on $t$.
Then equation (\ref{eq:Adtx1}) can be recast as
\begin{equation}
\mathsf{B}\partial_{t}Y\left(t\right)=\mathsf{A}Y\left(t\right),\label{eq:yxt1a}
\end{equation}
where $\mathsf{A}$ and $\mathsf{B}$ are $m\nu\times m\nu$ companion
matrices defined by equations (\ref{eq:CBA1b}) and

In the case when $A_{\nu}$ is an invertible $m\times m$ matrix equation
(\ref{eq:yxt1a}) can be recast further as
\begin{equation}
\partial_{t}Y\left(t\right)=\dot{\mathsf{A}}Y\left(t\right),\label{eq:yxt1c}
\end{equation}
where
\begin{gather}
\dot{\mathsf{A}}=\left[\begin{array}{ccccc}
0 & \mathbb{I} & \cdots & 0 & 0\\
0 & 0 & \mathbb{I} & \cdots & 0\\
0 & 0 & 0 & \cdots & \vdots\\
\vdots & \vdots & \ddots & 0 & \mathbb{I}\\
-\dot{A}_{0} & -\dot{A}_{1} & \cdots & -\dot{A}_{\nu-2} & -\dot{A}_{\nu-1}
\end{array}\right],\quad\dot{A}_{j}=A_{\nu}^{-1}A_{j},\quad0\leq\nu-1.\label{eq:yxt1d}
\end{gather}
Notice one can interpret equation (\ref{eq:yxt1c}) as particular
case of equation (\ref{eq:yxt1a}) where matrices $A_{\nu}$ and $\mathsf{B}$
are identity matrices of the respective dimensions $m\times m$ and
$m\nu\times m\nu$, and that polynomial matrix $A\left(s\right)$
defined by equations (\ref{eq:Aux1a}) becomes monic matrix polynomial
$\dot{A}\left(s\right)$, that is
\begin{gather}
\dot{A}\left(s\right)=\mathbb{I}s^{\nu}+\sum_{j=0}^{\nu-1}\dot{A}_{j}s^{j},\quad\dot{A}_{j}=A_{\nu}^{-1}A_{j},\quad0\leq\nu-1.\label{eq:yxt1e}
\end{gather}
Notice that in view of equation (\ref{eq:yxt1b}) one recovers $x\left(t\right)$
from $Y\left(t\right)$ by the following formula
\begin{equation}
x\left(t\right)=P_{1}Y\left(t\right),\quad P_{1}=\left[\begin{array}{ccccc}
\mathbb{I} & 0 & \cdots & 0 & 0\end{array}\right],\label{eq:yxt2b}
\end{equation}
where $P_{1}$ evidently is $m\times m\nu$ matrix.

Observe also that, \cite[Prop. 5.1.2]{GoLaRo2}, \cite[14]{LanTsi}
\begin{gather}
\left[\dot{A}\left(s\right)\right]^{-1}=P_{1}\left[\mathbb{I}s-\dot{\mathsf{A}}\right]^{-1}R_{1},\quad P_{1}=\left[\begin{array}{ccccc}
\mathbb{I} & 0 & \cdots & 0 & 0\end{array}\right],\quad R_{1}=\left[\begin{array}{c}
0\\
0\\
\vdots\\
0\\
\mathbb{I}
\end{array}\right],\label{eq:yxt2a}
\end{gather}
where $P_{1}$ and $R_{1}$ evidently respectively $m\times m\nu$
and $m\nu\times m$ matrices.

The general form for the solution to vector differential equation
(\ref{eq:yxt1c}) is
\begin{equation}
Y\left(t\right)=\exp\left\{ \dot{\mathsf{A}}t\right\} Y_{0},\quad Y_{0}\in\mathbb{C}^{m\nu}.\label{eq:yxt2c}
\end{equation}
Then using the formulas (\ref{eq:yxt2b}), (\ref{eq:yxt2c}) and Proposition
\ref{prop:gen-eig} we arrive the following statement.
\begin{prop}[solution to the vector differential equation]
\label{prop:dif-sol-g} Let $\dot{\mathsf{A}}$ be $m\nu\times m\nu$
companion matrix defined by equations (\ref{eq:yxt1d}), $\zeta_{1},\ldots,\zeta_{p}$
be its distinct eigenvalues, and $M_{\zeta_{1}},\ldots,M_{\zeta_{p}}$
be the corresponding generalized eigenspaces of the corresponding
dimensions $r\left(\zeta_{j}\right)$, $1\leq j\leq p$. Then the
$m\nu$ column-vector solution $Y\left(t\right)$ to differential
equation (\ref{eq:yxt1c}) is of the form
\begin{gather}
Y\left(t\right)=\exp\left\{ \dot{\mathsf{A}}t\right\} Y_{0}=\sum_{j=1}^{p}e^{\zeta_{j}t}p_{j}\left(t\right),\quad Y_{0}=\sum_{j=1}^{p}Y_{0,j},\quad Y_{0,j}\in M_{\zeta_{j}},\label{eq:yxt2d}
\end{gather}
where $m\nu$-column-vector polynomials $p_{j}\left(t\right)$ satisfy
\begin{gather}
p_{j}\left(t\right)=\sum_{k=0}^{r\left(\zeta_{j}\right)-1}\frac{t^{k}}{k!}\left(\dot{\mathsf{A}}-\zeta_{j}\mathbb{I}\right)^{k}Y_{0,j},\quad1\leq j\leq p.\label{eq:yxt2e}
\end{gather}
Consequently, the general $m$-column-vector solution $x\left(t\right)$
to differential equation (\ref{eq:Adtx1}) is of the form
\begin{gather}
x\left(t\right)=\sum_{j=1}^{p}e^{\zeta_{j}t}P_{1}p_{j}\left(t\right),\quad P_{1}=\left[\begin{array}{ccccc}
\mathbb{I} & 0 & \cdots & 0 & 0\end{array}\right].\label{eq:yxt2f}
\end{gather}
\end{prop}

Notice that $\chi_{\dot{\mathsf{A}}}\left(s\right)=\det\left\{ s\mathbb{I}-\dot{\mathsf{A}}\right\} $
is the characteristic function of the matrix $\dot{\mathsf{A}}$ then
using notations of Proposition \ref{prop:dif-sol-g} we obtain
\begin{equation}
\chi_{\dot{\mathsf{A}}}\left(s\right)=\prod_{j=1}^{p}\left(s-\zeta_{j}\right)^{r\left(\zeta_{j}\right)}.\label{eq:yxt3a}
\end{equation}
Notice also that for any values of complex-valued coefficients $b_{k}$
we have
\begin{gather}
\left(\partial_{t}-\zeta_{j}\right)^{r\left(\zeta_{j}\right)}\left[e^{\zeta_{j}t}p_{j}\left(t\right)\right]=0,\quad p_{j}\left(t\right)=\sum_{k=0}^{r\left(\zeta_{j}\right)-1}b_{k}t^{k},\label{eq:yxt3b}
\end{gather}
implying together with representation (\ref{eq:yxt3a})
\begin{gather}
\chi_{\dot{\mathsf{A}}}\left(\partial_{t}\right)\left[e^{\zeta_{j}t}p_{j}\left(t\right)\right]=0,\quad p_{j}\left(t\right)=\sum_{k=0}^{r\left(\zeta_{j}\right)-1}b_{k}t^{k}.\label{eq:yxt3c}
\end{gather}
Combing now Proposition \ref{prop:dif-sol-g} with equation (\ref{eq:yxt3c})
we obtain the following statement.
\begin{cor}[property of a solution to the vector differential equation]
\label{cor:dif-sol-g} Let $x\left(t\right)$ be the general $m$-column-vector
solution $x\left(t\right)$ to differential equation (\ref{eq:Adtx1}).
Then $x\left(t\right)$ satisfies
\begin{equation}
\chi_{\dot{\mathsf{A}}}\left(\partial_{t}\right)x\left(t\right)=0.\label{eq:yxt3d}
\end{equation}
\end{cor}

\section{Floquet theory\label{sec:floquet}}

We provide here a concise review of the Floquet theory following to
\cite[III]{DalKre}, \cite[III.7]{Hale} and \cite[II.2]{YakSta}.
The primary subject of the Floquet theory is the general form of solutions
to the ordinary differential equations with periodic coefficients.
With that in mind suppose that: (i) $z$ is real valued variable,
(ii) $x\left(z\right)$ is $n$-vector valued function of $z$, (iii)
$A\left(z\right)$ is $n\times n$ matrix valued $\varsigma$-periodic
function of $z$ and consider the following homogeneous linear periodic
system:
\begin{equation}
\partial_{z}x\left(z\right)=A\left(z\right)x\left(z\right),\quad A\left(z+\varsigma\right)=A\left(z\right),\quad\varsigma>0.\label{eq:floq1a}
\end{equation}
We would like to give a complete characterization of the general structure
of the solutions to equation (\ref{eq:floq1a}). We start with the
following statement showing how to define the logarithm $B$ of a
matrix $C$ so that $C=\mathrm{e}^{B}$.
\begin{lem}[logarithm of a matrix]
\label{lem:log-mat} Let $C$ be is $n\times n$ matrix with $\det\left\{ C\right\} \neq0$.
Suppose $C=Z^{-1}JZ$ where $J$ is Jordan canonical form of $C$
as described in Proposition \ref{prop:jor-can}. Then using the block
representation (\ref{eq:Jork1b}) for $J$, that is
\begin{equation}
J=J=\mathrm{diag}\,\left\{ J_{n_{1}}\left(\zeta_{1}\right),J_{n_{2}}\left(\zeta_{2}\right),\ldots,J_{n_{r}}\left(\zeta_{r}\right)\right\} ,\quad n_{1}+n_{2}+\cdots n_{q}=n,\label{eq:floq1b}
\end{equation}
we decompose $J$ into its diagonal and nilpotent components:
\begin{equation}
J=\mathrm{diag}\,\left\{ \lambda_{1}\mathbb{I}_{n_{1}},\lambda_{2}\mathbb{I}_{n_{2}},\ldots,\lambda_{q}\mathbb{I}_{n_{q}}\right\} +K\label{eq:floq1c}
\end{equation}
where
\begin{gather}
D=\mathrm{diag}\,\left\{ \lambda_{1}\mathbb{I}_{n_{1}},\lambda_{2}\mathbb{I}_{n_{2}},\ldots,\lambda_{q}\mathbb{I}_{n_{q}}\right\} ,\quad K=\mathrm{diag}\,\left\{ K_{n_{1}},K_{n_{2}},\ldots,K_{n_{q}}\right\} ,\label{eq:floq1d}\\
K_{n_{j}}=J_{n_{j}}\left(\lambda_{j}\right)-\lambda_{j}\mathbb{I}_{n_{j}},\quad1\leq j\leq q.\nonumber 
\end{gather}
Then let $\ln\left(\ast\right)$ be a branch of the logarithm and
let
\begin{equation}
H=\ln J=\mathrm{diag}\,\left\{ \ln\left(\lambda_{1}\right)\mathbb{I}_{n_{1}},\ln\left(\lambda_{2}\right)\mathbb{I}_{n_{2}},\ldots,\ln\left(\lambda_{q}\right)\mathbb{I}_{n_{q}}\right\} +S\label{eq:floq1e}
\end{equation}
where $\mathbb{I}_{n_{j}}$ are identity matrices of identified dimensions
and
\begin{equation}
S=\mathrm{diag}\,\left\{ S_{n_{1}},S_{n_{2}},\ldots,S_{n_{q}}\right\} ,\quad S_{n_{j}}=\sum_{m=1}^{n_{j}-1}\left(-1\right)^{m-1}\frac{1}{m\lambda_{j}^{m}}K_{n_{j}}^{m},\quad1\leq j\leq q.\label{eq:floq1g}
\end{equation}
Then
\begin{equation}
C=\mathrm{e}^{B},\quad B=\ln C=Z^{-1}HZ,\label{eq:floq1h}
\end{equation}
where matrix $H$ is defined by equation (\ref{eq:floq1e}).
\end{lem}

Note that matrix $S$ in equations (\ref{eq:floq1e}) and (\ref{eq:floq1g})
is associated with the nilpotent part of Jordan canonical form $J$.
The expression for $S_{n_{j}}$ originates in the series
\begin{equation}
\ln\left(1+s\right)=\sum_{m=1}^{\infty}\left(-1\right)^{m-1}\frac{1}{m}s^{m}=s-\frac{s^{2}}{2}+\frac{s^{3}}{3}+\cdots,\label{eq:floq2a}
\end{equation}
and it is a finite sum since $K_{n_{j}}$ is a nilpotent matrix such
that 
\begin{equation}
K_{n_{j}}^{m}=0,\quad m\geq n_{j},\quad1\leq j\leq q.\label{eq:floq2b}
\end{equation}

An $n\times n$ matrix $\Phi\left(z\right)$ is called \emph{matrizant
(matriciant)} of equation (\ref{eq:floq1a}) if it satisfies the following
equation:
\begin{equation}
\partial_{z}\Phi\left(z\right)=A\left(z\right)\Phi\left(z\right),\quad\Phi\left(0\right)=\mathbb{I},\quad A\left(z+\varsigma\right)=A\left(z\right),\quad\varsigma>0,\label{eq:floq2c}
\end{equation}
where $\mathbb{I}$ is the $n\times n$ identity matrix. Matrix $\Phi\left(z\right)$
is also called \emph{principal fundamental matrix} solution to equation
(\ref{eq:floq1a}). Evidently $x\left(z\right)=\Phi\left(z\right)x_{0}$
is the a solution to equation (\ref{eq:floq1a}) with the initial
condition $x\left(0\right)=x_{0}$. Using the fundamental solution
$\Phi\left(z\right)$ we can represent any matrix solution $\Psi\left(z\right)$
to equation (\ref{eq:floq1a}) based on its initial values as follows
\begin{equation}
\partial_{z}\Psi\left(z\right)=A\left(z\right)\Psi\left(z\right),\quad\Psi\left(z\right)=\Phi\left(z\right)\Psi\left(0\right).\label{eq:floq2d}
\end{equation}
In the case of $\varsigma$-periodic matrix function $A\left(z\right)$
the matrix function $\Psi\left(z\right)=\Phi\left(z+\varsigma\right)$
is evidently a solution to equation (\ref{eq:floq2d}) and consequently
\begin{equation}
\Phi\left(z+\varsigma\right)=\Phi\left(z\right)\Phi\left(\varsigma\right).\label{eq:floq2ca}
\end{equation}
It turns out that matrix $M_{\varsigma}=\Phi\left(\varsigma\right)$
called the \emph{monodromy matrix} is of particular importance for
the analysis of solutions to equation (\ref{eq:floq2c}) with $\varsigma$-periodic
matrix function $A\left(z\right)$. 

The monodromy matrix is integrated into the formulation of the main
statement of the Floquet theory describing the structure of solutions
to equation (\ref{eq:floq2d}) for $\varsigma$-periodic matrix function
$A\left(z\right)$.
\begin{thm}[Floquet]
\label{thm:floquet} Suppose that $A\left(z\right)$ is a $\varsigma$-periodic
continuous function of $z$. Let $\Phi\left(z\right)$ be the matrizant
of equation (\ref{eq:floq2c}) and let $M_{\varsigma}=\Phi\left(\varsigma\right)$
be the corresponding monodromy matrix. Using the statement of Lemma
\ref{lem:log-mat} we introduce matrix $\Gamma$ defined by
\begin{equation}
\Gamma=\frac{1}{\varsigma}\ln M_{\varsigma}=\frac{1}{\varsigma}\ln\Phi\left(\varsigma\right),\text{implying }M_{\varsigma}=\Phi\left(\varsigma\right)=\mathrm{e}^{\Gamma\varsigma}.\label{eq:floq2f}
\end{equation}
Then matrizant $\Phi\left(z\right)$ satisfies the following equation
called Floquet representation
\begin{equation}
\Phi\left(z\right)=P\left(z\right)\mathrm{e}^{\Gamma z},\quad P\left(z+\varsigma\right)=P\left(z\right),\quad P\left(0\right)=\mathbb{I},\label{eq:floq2g}
\end{equation}
where $P\left(z\right)$ is a differentiable $\varsigma$-periodic
matrix function of $z$. 
\end{thm}

\begin{proof}
Let us define matrix $P\left(z\right)$ by the following equation
\begin{equation}
P\left(z\right)=\Phi\left(z\right)\mathrm{e}^{-\Gamma z}.\label{eq:floq3a}
\end{equation}
Then combining representation (\ref{eq:floq3a}) for $P\left(z\right)$
with equations (\ref{eq:floq2ca}) and (\ref{eq:floq2f}) we obtain
\begin{equation}
P\left(z+\varsigma\right)=\Phi\left(z+\varsigma\right)\mathrm{e}^{-\Gamma\left(z+\varsigma\right)}=\Phi\left(z\right)\Phi\left(\varsigma\right)\mathrm{e}^{-\Gamma\varsigma}\mathrm{e}^{-\Gamma z}=\Phi\left(z\right)\mathrm{e}^{-\Gamma z}=P\left(z\right),\label{eq:floq3b}
\end{equation}
that is $P\left(z\right)$ is a differentiable $\varsigma$-periodic
matrix function of $z$. Equality $P\left(0\right)=\mathbb{I}$ readily
follows from equation (\ref{eq:floq3a}) and equality $\varPhi\left(0\right)=\mathbb{I}$.
\end{proof}
The eigenvalues of the monodromy matrix $\Phi\left(\varsigma\right)=\mathrm{e}^{\Gamma\varsigma}$
are called\emph{ Floquet (characteristic) multipliers} and their logarithms
(not uniquely defined) are called\emph{ characteristic exponents.}
\begin{defn}[Floquet multipliers, characteristic exponents and eigenmodes]
\label{def:floqmul} Using notation of Theorem \ref{thm:floquet}
let us consider complex numbers $\kappa$, $s_{\kappa}$ and vector
$y_{\kappa}$ satisfying the following equations
\begin{equation}
\Gamma y_{\kappa}=\kappa y_{\kappa},\quad\Phi\left(\varsigma\right)y_{\kappa}=\mathrm{e}^{-\Gamma\varsigma}y_{\kappa}=s_{\kappa}y_{\kappa},\quad s_{\kappa}=\mathrm{e}^{\kappa\varsigma},\label{eq:Gamyk1a}
\end{equation}
where evidently $\kappa$ and $y_{\kappa}$ are respectively an eigenvalue
and the corresponding eigenvector of matrix $\Gamma$. We refer to
$\kappa$ and $s_{\kappa}$ respectively as the \emph{Floquet characteristic
exponent} and the \emph{Floquet (characteristic) multiplier}.

Using $\kappa$ and $y_{\kappa}$ defined above we introduce the following
special solution to the original differential equation (\ref{eq:floq1a}):
\begin{equation}
\psi_{\kappa}\left(z\right)=p_{\kappa}\left(z\right)\mathrm{e}^{\kappa z}=\Phi\left(z\right)y_{\kappa}=P\left(z\right)\mathrm{e}^{\Gamma z}y_{\kappa},\quad p_{\kappa}\left(z\right)=P\left(z\right)y_{\kappa},\label{eq:Gamyk1b}
\end{equation}
and refer to it as the \emph{Floquet eigenmode}. Note that $p_{\kappa}\left(z\right)$
in equations (\ref{eq:Gamyk1b}) is $\varsigma$-periodic vector-function
of $z$.
\end{defn}

\begin{rem}[Floquet eigenmodes]
 If $\psi_{\kappa}\left(z\right)$ is the Floquet eigenmode defined
by equations (\ref{eq:Gamyk1b}) and $\Re\left\{ \kappa\right\} >0$
or equivalently $\left|s_{\kappa}\right|>1$ then $\psi_{\kappa}\left(z\right)$
grows exponentially as $z\rightarrow+\infty$ and we refer to such
$\psi_{\kappa}\left(z\right)$ as \emph{exponentially growing Floquet
eigenmode}. In the case when $\Re\left\{ \kappa\right\} =0$ or equivalently
$\left|s_{\kappa}\right|=1$ function $\psi_{\kappa}\left(z\right)$
is bounded and we refer to such $\psi_{\kappa}\left(z\right)$ as
an \emph{oscillatory Floquet eigenmode}. 
\end{rem}

\begin{rem}[dispersion relations]
 \label{rem:disprel} In physical applications of the Floquet theory
$\varsigma$-periodic matrix valued function $A\left(z\right)$ in
differential equation (\ref{eq:floq1a}) depends on the frequency
$\omega$, that is $A\left(z\right)=A\left(z,\omega\right)$. In this
case we also have $\kappa=\kappa\left(\omega\right)$. If we naturally
introduce the wave number $k$ by
\begin{equation}
k=k\left(\omega\right)=-\mathrm{i}\kappa\left(\omega\right),\label{eq:Gamyk1c}
\end{equation}
then the relation between $\omega$ and $k$ provided by equation
(\ref{eq:Gamyk1c}) is called the \emph{dispersion relation}.
\end{rem}

\section{Hamiltonian systems of linear differential equations\label{sec:Ham}}

We follow here to \cite[I.8, V.1]{DalKre} and \cite[III]{YakSta}.
We introduce first\emph{ indefinite scalar product} $\left\langle x,y\right\rangle $
on the vector space $\mathbb{C}^{n}$ associated with a nonsingular
Hermitian $n\times n$ matrix $G$, namely
\begin{equation}
\left\langle x,y\right\rangle =\overline{\left\langle y,x\right\rangle }=x^{*}Gy,\quad G^{*}=G,\quad\det\left\{ G\right\} \neq0,\quad x,y\in\mathbb{C}^{n}.\label{eq:Gindef1a}
\end{equation}
We refer to matrix $G$ \emph{metric matrix}. We define then for any
$n\times n$ matrix $A$ another matrix $A^{\dagger}$ called adjoint
by the following relations:
\begin{equation}
\left\langle Ax,y\right\rangle =\left\langle x,A^{\dagger}y\right\rangle \text{ or equivalently }A^{\dagger}=G^{-1}A^{*}G.\label{eq:Gindef1b}
\end{equation}
Notice that relations (\ref{eq:Gindef1b}) readily imply
\begin{equation}
\left(AB\right)^{\dagger}=B^{\dagger}A^{\dagger}.\label{eq:Gindef1c}
\end{equation}
\begin{table}
\centering{}%
\begin{tabular}{|c|c|c|}
\hline 
$G$-unitary & $G$-skew-Hermitian & $G$-Hermitian\tabularnewline
\hline 
\hline 
$\left\langle Ax,Ay\right\rangle =\left\langle x,y\right\rangle $ & $\left\langle Ax,y\right\rangle =-\left\langle x,Ay\right\rangle $ & $\left\langle Ax,y\right\rangle =\left\langle x,Ay\right\rangle $\tabularnewline
\hline 
$A^{\dagger}A=G^{-1}A^{*}GA=\mathbb{I},$ & $A^{\dagger}=G^{-1}A^{*}G=-A$, & $A^{\dagger}=G^{-1}A^{*}G=A,$\tabularnewline
\hline 
$A^{*}=GA^{-1}G^{-1}$ & $GA+A^{*}G=0$ & $GA-A^{*}G=0$\tabularnewline
\hline 
$A^{*}GA=G$ & $A=\mathrm{i}G^{-1}H,\;H=H^{*}$ & $A=G^{-1}H,\;H=H^{*}$\tabularnewline
\hline 
\end{tabular}\vspace{0.3cm}
\caption{\label{tab:Gunit}$G$-unitary, $G$-skew-Hermitian and $G$-Hermitian
matrices. }
\end{table}

Let $G$ and $H\left(t\right)$ be Hermitian $n\times n$ matrices
and suppose that matrix $G$ is nonsingular. We defined \emph{Hamiltonian
system of equations} to be a system of the form.
\begin{equation}
-\mathrm{i}G\partial_{t}x\left(t\right)=H\left(t\right)x\left(t\right),\quad H^{*}\left(t\right)=H\left(t\right).\label{eq:Gindefe1d}
\end{equation}
If based on matrices $G$ and $H\left(t\right)$ we introduce $G$-skew-Hermitian
matrix
\begin{equation}
A\left(t\right)=\mathrm{i}G^{-1}H\left(t\right),\label{eq:Gindefe1da}
\end{equation}
we can recast the Hamiltonian system (\ref{eq:Gindefe1d}) in the
following equivalent form,
\begin{equation}
\partial_{t}x\left(t\right)=A\left(t\right)x\left(t\right),\quad A^{\dagger}\left(t\right)=-A\left(t\right).\label{eq:Gindefe1e}
\end{equation}
It turns out that the matrizant $\Phi\left(t\right)$ of equation
(\ref{eq:Gindefe1e}) with $G$-skew-Hermitian matrix $A\left(t\right)$
is a $G$-unitary matrix for each value of $t$. Indeed, using equation
(\ref{eq:Gindefe1e}) together with equations (\ref{eq:Gindef1b}),
(\ref{eq:Gindef1c}) we obtain
\begin{gather}
\partial_{t}\left[\Phi^{\dagger}\left(t\right)\Phi\left(t\right)\right]=\left\{ \partial_{t}\left[\Phi\left(t\right)\right]\right\} ^{\dagger}\Phi\left(t\right)+\Phi^{\dagger}\left(t\right)\partial_{t}\left[\Phi\left(t\right)\right]=\label{eq:Gindefe1f}\\
=-\Phi^{\dagger}\left(t\right)A\left(t\right)\Phi\left(t\right)+\Phi^{\dagger}\left(t\right)A\left(t\right)\Phi\left(t\right)=0,\nonumber 
\end{gather}
implying that matrizant $\Phi\left(t\right)$ satisfies
\begin{equation}
\Phi^{\dagger}\left(t\right)\Phi\left(t\right)=\mathbb{I},\text{or equivalently }\Phi^{*}\left(t\right)G\Phi\left(t\right)=G,\label{eq:Gindefe2a}
\end{equation}
implying that $\Phi\left(t\right)$ is a $G$-unitary matrix for each
value of $t$. Identity (\ref{eq:Gindefe2a}) implies in turn that
for any two solutions $x\left(t\right)$ and $y\left(t\right)$ to
the Hamiltonian system (\ref{eq:Gindefe1d}) we always have
\begin{equation}
\left\langle x\left(t\right),y\left(t\right)\right\rangle =x^{*}\left(t\right)Gy\left(t\right)=x^{*}\left(0\right)\Phi^{*}\left(t\right)G\Phi\left(t\right)y\left(0\right)=\left\langle x\left(0\right),y\left(0\right)\right\rangle ,\label{eq:Gindefe2aa}
\end{equation}
that is $\left\langle x\left(t\right),y\left(t\right)\right\rangle $
does not depend on $t$.

\subsection{Symmetry of the spectra}

$G$-unitary, $G$-skew-Hermitian and $G$-Hermitian matrices have
special properties described in Table \ref{tab:Gunit}. These properties
can viewed as symmetries and not surprising they imply consequent
symmetries of the spectra of the matrices. Let $\sigma$$\left\{ A\right\} $
denote the spectrum of matrix $A$. It is a straightforward exercise
to verify based on matrix properties described in Table \ref{tab:Gunit}
that the following statements hold.
\begin{thm}[spectral symmetries]
\label{thm:sp-sym} Suppose that matrix $A$ is either $G$-unitary
or $G$-skew-Hermitian or $G$-Hermitian. Then the following statements
hold:
\begin{enumerate}
\item If $A$ is $G$-unitary then $\sigma$$\left\{ A\right\} $ is symmetric
with respect to the unit circle, that is
\begin{equation}
\zeta\in\sigma\left\{ \Phi\right\} \Rightarrow\frac{1}{\bar{\zeta}}\in\sigma\left\{ \Phi\right\} .\label{eq:Gindef2b}
\end{equation}
\item If $A$ is $G$-skew-Hermitian then $\sigma$$\left\{ A\right\} $
is symmetric with respect the imaginary axis, that is
\begin{equation}
\zeta\in\sigma\left\{ \Phi\right\} \Rightarrow-\bar{\zeta}\in\sigma\left\{ \Phi\right\} .\label{eq:Gindef2ba}
\end{equation}
\item If $A$ is $G$-Hermitian then $\sigma$$\left\{ A\right\} $ is symmetric
with respect to real axis, that is
\begin{equation}
\zeta\in\sigma\left\{ \Phi\right\} \Rightarrow\bar{\zeta}\in\sigma\left\{ \Phi\right\} .\label{eq:Gindef2bb}
\end{equation}
\end{enumerate}
\end{thm}

The following statement describes $G$-orthogonality of invariant
subspaces of $G$-unitary, $G$-skew-Hermitian and $G$-Hermitian
matrices, \cite[1.8]{DalKre}.
\begin{thm}[eigenspaces]
\label{thm:G-eig} Suppose that matrix $A$ is either $G$-unitary
or $G$-skew-Hermitian or $G$-Hermitian. Then the following statements
hold. Let $\Lambda\subset\sigma\left\{ A\right\} $ be a subset of
the spectrum $\sigma\left\{ A\right\} $ of the matrix $A$, and let
$\tilde{\Lambda}$ be the relevant symmetric image of $\Lambda$ defined
by
\[
\tilde{\Lambda}=\left\{ \begin{array}{rrr}
\left\{ \frac{1}{\bar{\zeta}}:\zeta\in\Lambda\right\}  & \text{if} & A\text{ is \ensuremath{G}-unitary }\\
\left\{ -\bar{\zeta}:\zeta\in\Lambda\right\}  & \text{if} & A\text{ is \ensuremath{G}-skew-Hermitian }\\
\left\{ \bar{\zeta}:\zeta\in\Lambda\right\}  & \text{if} & \text{ is \ensuremath{G}-Hermitian }
\end{array}\right..
\]
Let $\Lambda_{1},\Lambda_{2}\subset\sigma\left\{ A\right\} $ be two
subsets of the spectrum $\sigma\left\{ A\right\} $ so that $\tilde{\Lambda}_{1}$
and $\Lambda_{2}$ are separated from each other by non-intersecting
contours $\tilde{\Gamma}_{1}$ and $\varGamma_{2}$. Then the invariant
subspaces $E_{1}$ and $E_{1}$ of the matrix $A$ corresponding to
$\Lambda_{1}$ and $\Lambda_{2}$ are $G$-orthogonal.
\end{thm}

The statement below describes a special property of eigenvectors of
a $G$-unitary matrix.
\begin{lem}[isotropic eigenvector]
\label{lem:isoeig} Let $A$ be a $G$-unitary matrix and $\zeta$
be its eigenvalue that does not lie on the unit circuit, that $\left|\zeta\right|\neq1$.
Then if $x$ is the eigenvector corresponding to $\zeta$ it is isotropic,
that is
\begin{equation}
\left\langle x,x\right\rangle =x^{*}Gx=0,\quad Ax=\zeta x,\quad\left|\zeta\right|\neq1.\label{eq:Gindef2d}
\end{equation}
\end{lem}

\begin{proof}
Since $Ax=\zeta x$ and $A$ is a $G$-unitary we have

\[
\left\langle Ax,Ax\right\rangle =\left\langle \zeta x,\zeta x\right\rangle =\left|\zeta\right|^{2}\left\langle x,x\right\rangle ,\quad\left\langle Ax,Ax\right\rangle =\left\langle x,x\right\rangle .
\]
Combining the two equation above with $\left|\zeta\right|\neq1$ we
conclude that $\left\langle x,x\right\rangle =0$ which is the desired
equation (\ref{eq:Gindef2b}).
\end{proof}
\textbf{\vspace{0.1cm}
}

\textbf{DATA AVAILABILITY:} The data that supports the findings of
this study are available within the article.

\end{document}